\documentclass[a4paper,11pt]{amsart}
\usepackage[dvips]{graphicx}
\usepackage{amssymb}
\usepackage{amsmath}
\usepackage[utf8]{inputenc}
\usepackage[T1]{fontenc}
\usepackage{graphicx}
\usepackage{amsmath}
\usepackage{amsthm}
\usepackage{amsfonts}
\usepackage{amssymb}
\usepackage{enumerate}
\usepackage{vaucanson-g}
\usepackage{setspace}
\usepackage{enumitem}
\usepackage{dsfont}
\usepackage{mathrsfs}
 \usepackage{soul}
 \usepackage{wrapfig}
\usepackage{remreset}
  \usepackage{enumitem}
  \usepackage{mdframed}
  \usepackage{longtable}
  \usepackage{marginnote}
  \usepackage{lineno}
  \setlistdepth{6}
 \renewlist{itemize}{itemize}{6}
\makeatletter
\@namedef{subjclassname@2020}{%
  \textup{2020} Mathematics Subject Classification}
\makeatother
\usepackage{xcolor}

\usepackage{tikz}
\usetikzlibrary{shapes.geometric}
\usetikzlibrary{automata,positioning}
\usepackage{pgfplots}
\pgfplotsset{compat=1.15}
\usepackage{mathrsfs}
\usetikzlibrary{arrows}

\newcommand\blfootnote[1]{%
  \begingroup
  \renewcommand\thefootnote{}\footnote{#1}%
  \addtocounter{footnote}{-1}%
  \endgroup
}

\usepackage[labelformat=simple]{subcaption}

\definecolor{liens}{rgb}{1,0,0}
\usepackage[colorlinks=true, linkcolor=blue, 
hyperfootnotes=true,citecolor=blue,urlcolor=black]{hyperref}

\newtheorem{theo}{Theorem}[]
\newtheorem{prop}[theo]{Proposition}
\newtheorem{lem}[theo]{Lemma}
\newtheorem{coro}[theo]{Corollary}

\newtheorem{defi}[theo]{Definition}
\newtheorem{hypo}{Hypothesis}[]

\newtheorem{rem}[theo]{Remark}
\newtheorem{question}[theo]{Question}

\theoremstyle{definition} 
\newtheorem{ex}[theo]{Example}

\newtheorem{algorithm}[theo]{Algorithm}

\newcommand{\cacher}[1]{}

\newcommand{\logm}{l}

\newcommand{\tree}[1]{\mathcal T(#1)}

\newcommand{\thetalb}[1]{\theta_{#1}}

\newcommand{\ellmahl}{\ell}
\newcommand{\malop}[1]{\phi_{#1}}
\newcommand{\val}{\operatorname{val}}

\newcommand{\card}{\sharp}

\newcommand{\alphaa}{\alpha}
\newcommand{\betaa}{\beta}

\newcommand{\x}{p}

\newcommand{\commentaireJulienPetit}[1]{
}

\newcommand{\univ}[1]{\mathscr{U}}

\newcommand{\pzintro}{(z)}
\newcommand{\pz}{}
\newcommand{\res}[1]{\operatorname{\bullet}_{\vert #1}}

\newcommand{\Ecal}{\mathcal{E}}
\newcommand{\Hahn}{\mathscr{H}}

\newcommand{\supp}{\operatorname{supp}}

\renewcommand{\theta}{\breve{\epsilon}}
\renewcommand{\kappa}{K}
\newcommand{\vareps}{\breve{\epsilon}}

\newcommand{\hgt}{N}

\newcommand{\iotaprime}{\breve{\iota}}

\newcommand{\Rr}{{\mathcal R}}
\newcommand{\Rrinfty}{{\mathcal R}}
\newcommand{\Qr}{{\mathcal Q}}

\renewcommand{\P}{{\mathcal P}}
\newcommand{\Cond}[1]{\mathcal C_{#1}}
\newcommand{\Sol}{\operatorname{Sol}}

\newcommand{\minP}{\psi}

\newcommand{\Q}{{\mathbb Q}}
\newcommand{\Qbar}{\overline{\mathbb Q}}
\newcommand{\C}{\mathbb{C}}
\newcommand{\Z}{\mathbb{Z}}

\newcommand{\R}{\mathbb R}

\newcommand{\D}{\mathcal{D}}
\newcommand{\E}{\mathcal{E}}
\renewcommand{\S}{\mathcal{S}}
\newcommand{\V}{\mathcal{V}}
\newcommand{\Vmaj}[1]{\mathcal{V}_{#1,\leq 8}}
\newcommand{\W}{\mathcal{V}}

\newcommand{\kappaprime}{\kappa_{0}}

\newcommand{\bK}{\mathbf{K}}
\newcommand{\wplus}{w^{+}}
\newcommand{\wmoins}{w^{-}}
\newcommand{\minepsilon}{\tau}
\newcommand{\minepsilonlb}{\breve{\tau}}

\newcommand{\strictementcroissant}{increasing}
\newcommand{\strictementdecroissant}{decreasing}
\newcommand{\croissant}{nondecreasing}

\title{Hahn series and Mahler equations: algorithmic aspects}
\author{C. Faverjon}
\address{Universite Claude Bernard Lyon 1, CNRS, Ecole Centrale de Lyon, INSA Lyon, Université Jean Monnet, ICJ UMR5208, 69622 Villeurbanne, France.}
\email{faverjon@math.univ-lyon1.fr}
\author{J. Roques }
\address{Universite Claude Bernard Lyon 1, CNRS, Ecole Centrale de Lyon, INSA Lyon, Université Jean Monnet, ICJ UMR5208, 69622 Villeurbanne, France.}
\email{Julien.Roques@univ-lyon1.fr}
\date{\today}
\begin{document}

\blfootnote{This is the accepted version of the following article : {\it Faverjon, C. and Roques, J. (2024), Hahn series and Mahler equations: Algorithmic aspects. J. London Math. Soc., 110: e12945},  which has been published in final form at \url{https://doi.org/10.1112/jlms.12945}.}

\begin{abstract}
Many articles have recently been devoted to Mahler equations, partly because of their links with other branches of mathematics such as automata theory. 
Hahn series (a generalization of the Puiseux series allowing arbitrary exponents of the indeterminate as long as the set that supports them is well-ordered) play a central role in the theory of Mahler equations. In this paper, we address the following fundamental question: is there an algorithm to calculate the Hahn series solutions of a given linear Mahler equation?  What makes this question interesting is the fact that the Hahn series appearing in this context can have complicated supports with infinitely many accumulation points. 
Our (positive) answer to the above question involves among other things the construction of a computable well-ordered receptacle for the supports of the potential Hahn series solutions.  
\end{abstract}

\maketitle
\setcounter{tocdepth}{1}
\tableofcontents

\section{Introduction}
Let $\bK$ be a field (of any characteristic and not necessarily algebraically closed). A linear Mahler equation with coefficients in $\bK(z)$ is a functional equation of the form
\begin{equation}\label{eq mahl intro}
 a_{n}(z) y(z^{\ellmahl^{n}}) + a_{n-1}(z) y(z^{\ellmahl^{n-1}}) + \cdots + a_{0}(z) y(z) = 0
\end{equation}
for some $\ellmahl \in \Z_{\geq 2}$, $n \in \Z_{\geq 0}$ and $a_{0}(z),\ldots,a_{n}(z) \in \bK(z)$ with $a_{0}(z)a_{n}(z) \neq 0$.  

These equations are named after K. Mahler who wrote  influential papers on the arithmetic nature of the values taken by solutions of such equations at algebraic points; see \cite{MahlerArith1929,MahlerArith1930,MahlerUber1930}. Since then, the theory has undergone many developments, in various directions and is nowadays a very active field of research with many facets. The interactions between the theory of Mahler equations and other fields of mathematics have been fruitful in recent years. This is well illustrated by the work of Sh\"afke and Singer in \cite{SchafkeSinger} which gives a new proof of a conjecture of Loxton and van der Poorten -- previously established by Adamczewski and Bell in \cite{BorisAboutMahler} -- and, therefore, a new proof of Cobham's theorem in automata theory by using tools coming from the theory of functional equations.  Here are some references  
\cite{Kubota,LoxtVanderPort,MasserVanishTheoPowSeries,RandeThese,DumasThese,BeckerKReg,NishiokaLNM1631,DF96,ZannierOFE,CorvajaZannierSNA,AS03,pellarinAIMM,NG,NGT,PhilipponGaloisNA,AlgIndMahlFunc,BorisAboutMahler,BorisFaverjonMethodeMahler17,BorisFaverjonMethodeMahler,DHRMahler,CompSolMahlEq,BeckerConj,FernandesMMCarNN,BorisMahlerSelecta,SchafkeSinger,RoquesLSMS,PouletDensity,ADHHLDE,RoquesFrobForMahler,FaverjonPoulet,BorisBellSmertnigGap}.  

The Hahn series play a fundamental role in the theory of Mahler equations.
Let us first look at the following simple but instructive example: 
\begin{equation}\label{eq:ex intro}
z^{\ellmahl}f(z^{\ellmahl^{2}}) - (z^{\ellmahl}+z)f(z^{\ellmahl})+zf(z)=0.
\end{equation}
This equation has the obvious constant solution $f_{1}\pzintro=1$ and any other solution in the field of formal Laurent series $\bK((z))$ or even in the field of Puiseux series $\mathscr{P}=\bigcup_{d \in \Z_{\geq 1}} \bK((z^{\frac{1}{d}}))$ is of the form $\lambda f_{1}\pzintro$ for some $\lambda \in \bK$. However, a new solution can be found in the field $\Hahn$ of Hahn series\footnote{See section \ref{sec:hahn series} for the concept of Hahn series.} with coefficients in $\bK$ and value group $\Q$, namely  
\begin{equation}\label{f2 intro}
f_{2}\pzintro=\sum_{k \geq 1} z^{-\frac{1}{\ellmahl^{k}}}. 
\end{equation}
Hence, working in the field of Hahn series, we have found two $\bK$-linearly independent solutions of the above linear Mahler equation of order $n=2$ ({\it i.e.}, as many $\bK$-linearly independent solutions as the order of the equation), which is satisfactory.

Actually\footnote{In fact, the main results of \cite{RoquesLSMS} and their proofs extend {\it mutatis mutandis} to an arbitrary algebraically closed field $\bK$.}, when $\bf K=\Qbar$,  it follows from \cite{RoquesLSMS} that the difference field $(\Hahn,\phi_{\ellmahl})$, where $\phi_{\ellmahl}$ is the field automorphism of $\Hahn$ sending $f (z)$ on $f(z^{\ellmahl})$, has a difference ring extension $(\mathcal{A},\phi_{\ellmahl})$ such that
\begin{itemize}
 \item for any $c \in \Qbar^{\times}$, there exists $e_{c} \in \mathcal{A}$ satisfying $\phi_{\ellmahl}(e_{c})=ce_{c}$;
\item there exists $\logm \in \mathcal{A}$ satisfying $\phi_{\ellmahl}(\logm)=\logm+1$; 
\item any linear Mahler equation of the form \eqref{eq mahl intro} has $n$ $\Qbar$-linearly independent solutions $y_{1},\ldots,y_{n} \in \mathcal{A}$ of the form 
\begin{equation}\label{form yi intro}
 y_{i} = \sum_{(c,j) \in \Qbar^{\times} \times \Z_{\geq 0}}  f_{i,c,j} e_{c} \logm^{j} 
\end{equation}
where the sum is finite and the $f_{i,c,j}$ belong to $\Hahn$.
\end{itemize}
This leads to the following fundamental question to which this article is devoted. 

\begin{question}\label{main question: un algo pour sol hahn?}
 Is there an algorithm to calculate the Hahn series solutions of an equation of the form \eqref{eq mahl intro}? 
\end{question}

Before formulating this question more formally, let us say a few words about the calculation of the solutions of linear Mahler equations such as  \eqref{eq mahl intro} in the more usual ring of formal power series $\bK[[z]]$. By ``calculating'' the solutions of an equation of the form \eqref{eq mahl intro} in $\bK[[z]]$, we usually mean calculating the formal power series solutions truncated to a specified order, {\it i.e.}, $\hgt \in \Z_{\geq 0}$ being given, we want to determine the $\sum_{k \in \{0,\ldots,\hgt\}} f_{k} z^{k} \in \bK[z]$ for which there exists a solution 
$$\widetilde{f}\pzintro=\sum_{k\in \Z_{\geq 0}} \widetilde{f}_{k} z^{k} \in \bK[[z]]$$ 
of \eqref{eq mahl intro} such that
\begin{equation}\label{eq pour def quest 1 cas power series}
 \sum_{k \in \{0,\ldots,\hgt\}} \widetilde{f}_{k} z^{k} 
=\sum_{k \in \{0,\ldots,\hgt\}} f_{k} z^{k}.
\end{equation}

A natural formalization of Question \ref{main question: un algo pour sol hahn?} is obtained by replacing the sets of indices $\Z_{\geq 0}$ and $\{0,\ldots,\hgt\}$ by $\Q$ and by an arbitrary finite subset $\E$ of $\Q$ respectively. 
More explicitly, this leads to the following formalization of Question \ref{main question: un algo pour sol hahn?}: a finite subset $\mathcal E$ of $\Q$ being given, we want to determine the $\sum_{\gamma \in \E} f_{\gamma} z^{\gamma} \in \Hahn$  for which there exists a solution 
$$\widetilde{f}\pzintro=\sum_{\gamma \in \Q} \widetilde{f}_{\gamma} z^{\gamma} \in \Hahn$$ 
of \eqref{eq mahl intro} such that
\begin{equation}\label{eq pour def quest 1}
 \sum_{\gamma \in \mathcal E} \widetilde{f}_{\gamma} z^{\gamma}
=\sum_{\gamma \in \mathcal E} f_{\gamma} z^{\gamma}.
\end{equation}

\begin{rem}
(1) An arbitrary Hahn series truncated at an given order has infinitely many nonzero coefficients in general. For instance, the truncation at order $0$ of the Hahn series $f_{2}\pzintro$ given by \eqref{f2 intro} is $f_{2}\pzintro$ itself and has infinitely many nonzero coefficients. This is why Question \ref{main question: un algo pour sol hahn?} is not stated in terms of truncated Hahn series. 

(2) The truncation $\sum_{k \in \{0,\ldots,\hgt\}} \widetilde{f}_{k} z^{k}$ of  $\widetilde{f}\pzintro=\sum_{k \in \Z_{\geq 0}} \widetilde{f}_{k} z^{k} \in \bK[[z]]$ can be interpreted as what remains of $\widetilde{f}\pzintro$ when only the indices $k\in \Z_{\geq 0}$ such that $H(k) \leq \hgt$ are retained, where $H$ denotes the naive height function defined, for any rational number $x=a/b$ where $a \in \Z$, $b\in \Z \setminus \{0\}$ are coprime, by $H(x)=\max\{\vert a\vert, \vert b\vert \}$. This leads to an alternative formulation of Question \ref{main question: un algo pour sol hahn?} similar to that given above but with  $\E$ replaced by $\E_{\hgt}=\{\gamma \in \Q \ \vert \ H(\gamma) \leq \hgt\}$. Since $\E_{\hgt}$ is a finite subset of $\Q$ and since any finite subset $\E$ of $\Q$ is a subset of $\E_{\hgt}$ for some $\hgt \in \Z_{\geq 0}$, the formulation of Question \ref{main question: un algo pour sol hahn?} given above is equivalent to this one.  
\end{rem}

Note that, since we are not imposing any conditions on $\mathcal{E}$, the solution $\widetilde{f}\pzintro$ may not be uniquely determined by \eqref{eq pour def quest 1}. 
 Fortunately, there is a simple condition guaranteeing that $\widetilde{f}\pzintro$ is uniquely determined by \eqref{eq pour def quest 1}: for this to be true, it suffices that $-\S \subset \mathcal E$ where $\S$ is the (finite and explicit) set of slopes of \eqref{eq mahl intro} defined in section \ref{sec: ext newt pol}; this follows directly from Corollary \ref{sol eq si eq on - pentes}.  
So, $- \S$ can serve as a set of indices for ``initial coefficients'' of all Hahn series solutions of \eqref{eq mahl intro} and, in the special case $\E=-\S$, Question \ref{main question: un algo pour sol hahn?} aims to describe the possible ``initial coefficients''. 
However, we draw the reader's attention to the fact that, even if one knows explicitly the ``initial part'' $\sum_{\gamma \in -\S} \widetilde{f}_{\gamma} z^{\gamma}$ of a solution $\widetilde{f}\pzintro=\sum_{\gamma \in \Q} \widetilde{f}_{\gamma} z^{\gamma} \in \Hahn$ of \eqref{eq mahl intro}, it is not obvious at all to compute the value of $\widetilde{f}_{\gamma}$ for a given $\gamma \in \Q \setminus -\S$ from it, whence the importance of allowing an arbitrary finite set $\E$ in Question~\ref{main question: un algo pour sol hahn?}.

An algorithm to find the solutions in the field of Puiseux series $\mathscr{P}=\bigcup_{d \in \Z_{\geq 1}} \bK((z^{\frac{1}{d}}))$ of a given Mahler equation has been given in \cite{CompSolMahlEq}\footnote{We mention for the interested reader that, when $\bK \subset \C$, any Puiseux series solution is actually convergent; see \cite[Lem.\,4]{BCR2013} for example.}. It consists in bounding  the ramification of these solutions in order to reduce the problem to the search of the solutions in a specific field of ramified Laurent series $\bK((z^{\frac{1}{d}}))$ for an explicit $d \in \Z_{\geq 1}$. What makes the search of the Hahn series solutions interesting is precisely the fact that one cannot reduce the problem to the search of (ramified) Laurent series solutions: 
one has to deal with Hahn series that might have rather involved supports. The support $\supp f_{2}\pzintro =\{-\frac{1}{\ellmahl^{k}} \ \vert \ k \in \Z_{\geq 1} \}$ of the Hahn series $f_{2}\pzintro$ given by \eqref{f2 intro} is one of the simplest support one can expect for a (non Puiseux) Hahn series solution of a linear Mahler equation. Much more complicated supports may arise. For instance, the Hahn series $f_{2}\pzintro^{2}$ satisfies a linear Mahler equation of order $3$ and, if the characteristic of $\bK$ is not equal to $2$, its support $\supp f_{2}\pzintro^{2} =\{-\frac{1}{\ellmahl^{k}}- \frac{1}{\ellmahl^{k'}}\ \vert \ k,k' \in \Z_{\geq 1} \}$ has infinitely many accumulation points, namely any element of $\{0\} \cup \{-\frac{1}{\ellmahl^{k}} \ \vert \ k \in \Z_{\geq 1} \}$.  These complicated supports induce many difficulties.

\subsection{Outline of our answer to Question \ref{main question: un algo pour sol hahn?}}\label{sec:outline algo}

Our approach to answer Question \ref{main question: un algo pour sol hahn?} relies on the following two ingredients. 
\begin{enumerate}
 \item \label{ingred 1} We introduce a subset $\V$ of $\Q$ satisfying the following properties:
\begin{itemize}
	\item[-] $\V$ contains the support of any Hahn series solution of \eqref{eq mahl intro};
	 \item[-] $\V$ is well-ordered; 
	\item[-] $\V$ is computable in the sense that there exists an algorithm to determine whether a given rational number belongs to $\V$ or not;
	\item[-] $\V$ satisfies a technical but important condition that we do not state here.
\end{itemize}
\item  \label{ingred 2} A finite set $\mathcal E \subset \V$ being given, we show that we can compute algorithmically a finite subset $\Rr$ of $\V$ containing $\E$  such that, for any $f\pzintro = \sum_{\gamma \in \Rr} f_\gamma z^\gamma \in \Hahn$, the following properties are equivalent: 
\begin{itemize}
 \item[-] there exists a solution $\widetilde{f}\pzintro=\sum_{\gamma \in \Q} \widetilde{f}_{\gamma} z^{\gamma} \in \Hahn$ of \eqref{eq mahl intro} such that 
 $$
\sum_{\gamma \in \Rr } \widetilde{f}_{\gamma} z^{\gamma}= \sum_{\gamma \in \Rr } f_{\gamma} z^{\gamma};
 $$
 \item[-] the support of the Hahn series
\begin{equation}\label{une eq de l intro}
a_{n}(z) f(z^{\ellmahl^{n}}) + a_{n-1}(z) f(z^{\ellmahl^{n-1}}) + \cdots + a_{0}(z) f(z)  
\end{equation}
is disjoint from $\minP(\Rr)$ where $\minP:\mathbb Q \to \mathbb Q$ is an explicit map defined in section \ref{sec:the maps Psi minP pi}.
\end{itemize}
\end{enumerate}

This reduces Question \ref{main question: un algo pour sol hahn?}  to a question of linear algebra. 
Indeed, up to multiplying \eqref{eq mahl intro} by a suitable nonzero polynomial, one can assume that the $a_{i}(z)$ are polynomials. 
Then, one can compute an explicit family of linear maps $F_{\delta} : \bK^{\Rr} \rightarrow \bK$ such that, for any $f\pzintro=\sum_{\gamma \in \Rr} f_{\gamma} z^{\gamma} \in \Hahn$, 
$$
a_{n}(z) f(z^{\ellmahl^{n}}) + a_{n-1}(z) f(z^{\ellmahl^{n-1}}) + \cdots + a_{0}(z) f(z)=\sum_{\delta \in \mathbb Q} F_{\delta}((f_{\gamma})_{\gamma \in \Rr}) z^{\delta}.
$$ 
The fact that the support of \eqref{une eq de l intro} is disjoint from $\minP(\Rr)$ is equivalent to the fact that, for all $\delta \in \minP(\Rr)$, $F_{\delta}((f_{\gamma})_{\gamma \in \Rr}) =0$. This is an (explicit) system of linear equations in the $(f_{\gamma})_{\gamma \in \Rr}$ that can be solved algorithmically. This solves Question \ref{main question: un algo pour sol hahn?}. 

\subsection{Organization of the paper} In section \ref{sec:hahn series}, we recall basic definitions and properties of the Hahn series. In section \ref{sec: ext newt pol}, we first  recall the notions of Newton polygons and of slopes. We then state and prove several results used elsewhere in the paper. In section \ref{sec:recept supp sol L bis}, we give an algorithm to compute a set $\V$ having the properties listed in section \ref{sec:outline algo} above. In sections \ref{sec:approx sol} and \ref{sec constr R}, we give an algorithm to compute a set $\Rr$ having the properties listed in section~\ref{sec:outline algo} above. 
In section \ref{sec: ancsw quest}, we describe an algorithm that answers Question \ref{main question: un algo pour sol hahn?} in the affirmative. In section \ref{illustration main algo on RS}, we apply our main algorithm to a classical equation. 
\vskip 10 pt 
\noindent \textbf{Acknowledgements.} Our warmest thanks go to the referees for their careful reading and their many suggestions, which have considerably improved the readability of this paper. 
The work of the second author was supported by the ANR De rerum natura project, grant ANR-19-CE40-0018 of the French Agence Nationale de la Recherche.

\section{The ring of Hahn series}  \label{sec:hahn series}

We denote by 
$
\Hahn
$ 
the field of Hahn series with coefficients in the field $\bK$ and with value group $\Q$ (see \cite{HahnSeriesHahn}). An element of $\Hahn$ is an $(f_{\gamma})_{\gamma \in \Q} \in \bK^{\Q}$ whose support 
$$
\supp (f_{\gamma})_{\gamma \in \Q}=\{\gamma\in \Q \ \vert \ f_{\gamma}\neq 0\}
$$ 
is well-ordered, {\it i.e.}, such that any nonempty subset of this support has a least element. An element $(f_{\gamma})_{\gamma \in \Q}$ of $
\Hahn
$
is usually (and will be) denoted by 
 $$
 f\pz =\sum _{{\gamma \in \Q }}f_{\gamma}z^{\gamma}. 
 $$
The sum and product of two elements 
$f\pz =\sum _{{\gamma\in \Q }}f_{\gamma}z^{\gamma}$ and 
$g\pz =\sum _{{\gamma\in \Q }}g_{\gamma}z^{\gamma}$ of $\Hahn$ 
are respectively defined by
$$
f\pz +g\pz =\sum _{{\gamma\in \Q }}(f_{\gamma}+g_{\gamma})z^{\gamma}
$$
and
$$
f\pz g\pz =\sum _{{\gamma\in \Q }}\left(\sum _{{\gamma'+\gamma''=\gamma}}f_{{\gamma'}}g_{{\gamma''}}\right)z^{\gamma}.
$$
(Note that there are only finitely many $(\gamma',\gamma'') \in \Q \times \Q$ such that $\gamma'+\gamma''=\gamma$ and $f_{{\gamma'}}g_{{\gamma''}}\neq 0$.) For a proof that $\Hahn$ endowed with this ring structure is a field, we refer to \cite[Th.\,5.7]{Neumann_Hahn}. 

Since the support of any Hahn series is well-ordered,  one can define the $z$-adic valuation 
\begin{eqnarray*}
 \val : \Hahn & \rightarrow & \Q \cup \{+\infty\} \\
 f\pz & \mapsto & \val f\pz =\min \supp f\pz 
\end{eqnarray*}
with the convention $\min \emptyset = +\infty$. It satisfies the usual properties of a valuation, namely~:
\begin{itemize}
\item $\forall f\pz\in \Hahn$, $(\val f\pz =+\infty \Longleftrightarrow f\pz=0);$
\item $\forall f\pz,g\pz\in \Hahn$, 
\begin{equation}\label{val prod}
 \val (f\pz g\pz)=\val f\pz+\val g\pz 
\end{equation}
and
\begin{equation}\label{val ult metric}
 \val(f\pz+g\pz)\geqslant \min \{\val f\pz,\val g\pz\}.
\end{equation}
\end{itemize}

For any subset $\Qr$ of $\Q$, we let $\Hahn_{\vert \Qr}$ be the $\bK$-vector space of Hahn series with support in $\Qr$, {\it i.e.},
$$
\Hahn_{\vert \Qr} = \{f\pz  \in \Hahn \ \vert \ \supp f\pz  \subset \Qr \}. 
$$ 
We have a natural $\bK$-linear map 
$$ 
\begin{array}{cccc}
\res{\Qr} : & \Hahn &\rightarrow& \Hahn_{\vert \Qr} \\
&     f\pz = \sum_{\gamma \in \mathbb Q} f_{\gamma} z^{\gamma} &\mapsto& f_{\vert \Qr}\pz := \sum_{\gamma \in \Qr} f_{\gamma} z^{\gamma} .
\end{array}
$$
For any $f\pz ,g\pz \in \Hahn$ and any $\Qr \subset \Q$, we will say that ``$f\pz =g\pz $ on $\Qr$'' if $f_{\vert \Qr}\pz=g_{\vert \Qr}\pz$.

\section{Newton polygons}\label{sec: ext newt pol}

Let $\malop{\ellmahl}$ be the field automorphism of $\Hahn$ sending $f(z)$ on $f(z^{\ellmahl})$. 
We denote by 
$$
\mathcal{D}_{\bK[z]}=\bK[z] \langle \malop{\ellmahl} \rangle 
$$ 
the Ore algebra of noncommutative polynomials with coefficients in $\bK[z]$ such that, for all $f\pz\in \bK[z]$, 
$
\malop{\ellmahl} f = \malop{\ellmahl}(f) \malop{\ellmahl}.
$ 
An element of $\mathcal{D}_{\bK[z]}$ will be called a Mahler operator. 

In what follows, we consider an inhomogeneous Mahler equation 
\begin{equation}\label{eq mahl slopes}
 a_{n}(z) y(z^{\ellmahl^{n}}) + a_{n-1}(z)  y(z^{\ellmahl^{n-1}}) + \cdots + a_{0}(z)  y\pz  = a_{-\infty}(z) 
\end{equation}
with $a_{0}(z),\ldots,a_{n}(z) \in \bK[z]$ and $a_{-\infty}(z) \in \Hahn$ such that $a_{0}(z) a_{n}(z) \neq 0$. This equation can be rewritten as 
$$
L (y\pz ) =a_{-\infty}\pz 
$$
where 
\begin{equation}\label{mahl op assoc bis}
L=a_{n}\pz \malop{\ellmahl}^{n} + a_{n-1}\pz  \malop{\ellmahl}^{n-1} + \cdots  + a_{0}\pz \in \D_{\bK[z]}.
\end{equation}
 
 \subsection{Newton polygon}\label{subsec:slopes}
Following \cite{CompSolMahlEq}, we define the Newton polygon $\mathcal N(L,a_{-\infty}\pz )$  of \eqref{eq mahl slopes}  
as the lower convex hull of the set
$$
\P(L,a_{-\infty}\pz )=\{(\ellmahl^{i},j) \ \vert \ i \in \{-\infty,0,\ldots,n\}, \ j \in \supp a_{i }\pz\} \subset \R^{2}
$$
with the convention $\ellmahl^{-\infty}=0$. 
 In other terms, $\mathcal N(L,a_{-\infty}\pz )$ is the convex hull of the set
\begin{equation}\label{convex hull of this set is newton}
 \{(\ellmahl^{i},j) \ \vert \ i \in \{-\infty,0,\ldots,n\}, \ j  \geq \val a_{i}\pz \} \subset \R^{2}. 
\end{equation}

\subsection{Slopes}\label{sec : def slopes} 
The polygon $\mathcal N(L,a_\infty)$ is delimited by two vertical half lines and by finitely many nonvertical vectors 
having pairwise distinct slopes, called the slopes of \eqref{eq mahl slopes}. The set of slopes of \eqref{eq mahl slopes} will be denoted by $\mathcal{S}(L,a_{-\infty}\pz )$.
The following result gives an useful characterization of these slopes.

\begin{lem}\label{rem charact slopes 1}
The following properties relative to $\mu \in \mathbb Q $ are equivalent:
\begin{itemize}
 \item[(i)] $\mu $ belongs to $\mathcal{S}(L,a_{-\infty}\pz )$;
 \item[(ii)] there exist distinct $i_1,i_2\in \{-\infty,0,\ldots,n\}$ such that 
\begin{equation}\label{eq:caract slopes}
  \val a_{i_{1}}\pz -\ellmahl^{i_{1}} \mu =\val a_{i_{2}}\pz -\ellmahl^{i_{2}} \mu =\min_{i \in \{-\infty,0,\ldots,n\}} \val a_{i}\pz -\ellmahl^{i} \mu . 
\end{equation}
\end{itemize}
Moreover, if $\mu $ belongs to $\mathcal{S}(L,a_{-\infty}\pz )$, then the equality \eqref{eq:caract slopes} is satisfied if and only if $(\ellmahl^{i_{1}},\val a_{i_1})$ and $(\ellmahl^{i_{2}},\val a_{i_2})$ belong to the edge of slope $\mu $ of $\mathcal{N}(L,a_{-\infty}\pz )$.
\end{lem}

\begin{proof}
The fact that $\mu $ satisfies (i) is equivalent to the fact that $\mu $ is the slope of a nonvertical edge of $\mathcal{N}(L,a_{-\infty}\pz )$ which is in turn equivalent to the fact that there exists $b \in \R$ such that the affine line $y-\mu x=b$ contains an edge of $\mathcal{N}(L,a_{-\infty}\pz )$. 

But, since $\mathcal{N}(L,a_{-\infty}\pz )$ is the convex hull of the set 
\eqref{convex hull of this set is newton}, the affine line $y-\mu x=b$ contains an edge of $\mathcal{N}(L,a_{-\infty}\pz )$ if and only  the following properties are satisfied: 
\begin{itemize}
 \item there exist distinct $i_1,i_2\in \{-\infty,0,\ldots,n\}$ such that $(\ellmahl^{i_{1}},\val a_{i_{1}}\pz )$ and $(\ellmahl^{i_{2}},\val a_{i_{2}}\pz )$ belong to the line $y-\mu x=b$;
 \item for any $i \in  \{-\infty,0,\ldots,n\}$, if $\val a_{i}\pz < +\infty$, then $(\ellmahl^{i},\val a_{i}\pz )$ belongs to the half space $y-\mu x\geq b$. 
\end{itemize}
The latter two properties are of course equivalent to the fact that there exist distinct $i_1,i_2\in \{-\infty,0,\ldots,n\}$ such that \begin{equation*}
b=  \val a_{i_{1}}\pz -\ellmahl^{i_{1}} \mu =\val a_{i_{2}}\pz -\ellmahl^{i_{2}} \mu =\min_{i \in \{-\infty,0,\ldots,n\}} \val a_{i}\pz -\ellmahl^{i} \mu . 
\end{equation*}
 This shows that (i) and (ii) are equivalent. 
 
 Last, if $\mu $ belongs to $\mathcal{S}(L,a_{-\infty}\pz )$, then the previous discussion shows that the equation of the line containing the edge of slope $\mu $ of $\mathcal{N}(L,a_{-\infty}\pz )$ is $y-\mu x=b$ with $b=\min_{i \in \{-\infty,0,\ldots,n\}} \val a_{i}\pz -\ellmahl^{i} \mu $.   The last assertion of the lemma follows directly from this. 
\end{proof}

\subsection{Newton polygon and slopes in the homogeneous case}

In the homogeneous case, that is when $a_{-\infty}\pz =0$, we will omit $a_{-\infty}\pz $ in the previous notations and terminologies. For instance, $\mathcal N(L,0)$ will simply be denoted by $\mathcal N(L)$ and will be called the Newton polygon of $L$.

\begin{rem}\label{rem:PL_finite}
  Note, for later use, the following simple but important fact: since the coefficients $a_{0}\pz,\ldots,a_{n}\pz$ of $L$ are in $\bK[z]$, the set $\P(L)$ is finite.
\end{rem}

We denote by 
$$
\mu_{1} < \cdots < \mu_{\kappa} 
$$  
the slopes of $L$, so that 
$$
\mathcal{S}(L)=\{\mu_{1},\ldots,\mu_{\kappa}\}.
$$

We let
$$
\x_{0},\ldots,\x_{\kappa} \in \Z_{\geq 0} \times \Z
$$
be the vertices, ordered by \strictementcroissant{} abscissa, of the polygon $\mathcal N(L)$. 
For any $k \in \{0,\ldots,\kappa\}$, we let $\alpha_{k}$
be the unique element of $\{0,\ldots,n\}$ and $\beta_{k}$ be the unique element of $\Z$ such that 
$$
\x_k=(\ellmahl^{\alpha_{k}},\beta_k)=(\ellmahl^{\alpha_{k}},\val a_{\alpha_{k}}\pz ). 
$$
Note that $\alpha_{0}=0$  and that $\alpha_{\kappa}=n$. With these notations, the edge of $\mathcal{N}(L)$ with slope $\mu_k$ has $\x_{k-1}=(\ellmahl^{\alpha_{k-1}},\beta_{k-1}\pz )$ as its left endpoint and $\x_k=(\ellmahl^{\alpha_{k}},\beta_{k}\pz )$ as its right endpoint.

\begin{ex}\label{ex:L pour example}
The main algorithm presented in this paper will be illustrated in section \ref{illustration main algo on RS} on the Mahler operator of order $2$ given by 
\begin{equation}\label{L pour example}
L=z\malop{2}^{2}+(z-1)\malop{2}-2.
\end{equation} 
The set $\P(L)$, the Newton polygon $\mathcal N(L)$ and the vertices $\x _{k}$  associated to this specific $L$ 
are given in section \ref{sec:example Newton et slopes} and represented in Figure \ref{fig:rudin shapiro polygon}.
\end{ex}

We note the following result for further use. 
	
\begin{lem}\label{lem:ineq slopes N(L)}
For all $(\ellmahl^{i},j) \in \P(L)$ and all $k \in \{1,\ldots,\kappa\}$, we have: 
	\begin{equation}\label{une ineg utile}
-\ellmahl^{i} \mu_{k} + j \geq -\ellmahl^{\alpha_{k}}  \mu_{k} +\beta_{k}=-\ellmahl^{\alpha_{k-1}}  \mu_{k} +\beta_{k-1}. 
\end{equation}
In geometric terms, this means that the minimum of the ordinates of the projections of points of $\P(L)$ along a line of slope $\mu_{k}$ on the $y$-axis is reached at $\x_{k-1}$ and at $\x_k$.
\end{lem}

\begin{proof}
Lemma \ref{rem charact slopes 1} applied with $\mu =\mu_k$, $(\ellmahl^{i_{1}},\val a_{i_{1}})=p_{k}$ and $(\ellmahl^{i_{2}},\val a_{i_{2}})=p_{k-1}$ ensures that $-\ellmahl^{\alpha_{k}}  \mu_{k} +\beta_{k}=-\ellmahl^{\alpha_{k-1}}  \mu_{k} +\beta_{k-1}=\min_{i \in \{0,\ldots,n\}} -\ellmahl^{i} \mu_{k} + \val a_{i}  $. The inequality \eqref{une ineg utile} follows from this and from the fact that, for all $(\ellmahl^{i},j) \in \P(L)$, we have $j \geq \val a_{i}$ and, hence, $-\ellmahl^{i} \mu_{k} + j \geq -\ellmahl^{i} \mu_{k} + \val a_{i}$. 
\end{proof}

\subsection{The maps $\Psi$, $\minP$ and $\pi$}\label{sec:the maps Psi minP pi}
 In the rest of the paper, we will intensively use the following three maps:
\begin{equation}\label{def psi}
\begin{array}{lrclc}
&\Psi : \Q &\rightarrow& \{\,\text{Finite subsets of $\Q$}\,\}  \\
 &    v &\mapsto& \{v \ellmahl^{i}+j \ \vert \  (\ellmahl^{i},j) \in \P(L)\}
\\ &&&&
\\
  &\minP : \Q &\rightarrow& \Q \\
   &  v &\mapsto& \min \Psi(v)  =\min \{v\ellmahl^{i}+j \ \vert \  (\ellmahl^{i},j) \in \P(L)\} \\
    & && \hskip 46pt  =\min \{v\ellmahl^{i}+\val a_{i}\pz \ \vert \  i \in \{0,\ldots,n\}\}
\\
{\text{and}}&&&
\\
  &\pi : \Q &\rightarrow& \Q \\
   &  q &\mapsto& \max \left \{\frac{q-j}{\ellmahl^{i}} \ \vert \ (\ellmahl^{i},j) \in \P(L)\right\}=- \min \S(L,z^q)\,.
\end{array}
\end{equation}
These maps are well-defined because $\P(L)$ is finite according to Remark \ref{rem:PL_finite}. 
In geometric terms: 
\begin{itemize}
 \item $\Psi(v)$ is the set of ordinates of the projection of the elements of $\P(L)$ along a line of slope $-v$ onto the $y$-axis;
 \item $\minP(v)$ is the least of these ordinates;
 \item $\pi(q)$ is the opposite of the minimum of the slopes of the lines passing through $(0,q)$ and an element of $\P(L)$.
\end{itemize}

Let us now give a more computational interpretation of these maps. Setting, for any $i \in \{0,\ldots,n\}$, 
$$a_{i}=\sum_{j \in \supp a_{i}} a_{i,j} z^{j},
$$ 
we have 
\begin{equation}\label{eq:Lzv}
L(z^{v}) = \sum_{i=0}^{n} a_{i,j} z^{v \ellmahl^{i} +j}=\sum_{(\ellmahl^{i},j) \in \P(L)} a_{i,j} z^{v \ellmahl^{i} +j}. 
\end{equation}
This formula shows that $\Psi(v)$ is a natural receptacle for the support of $L(z^{v})$. Indeed, we have 
\begin{equation}\label{eq:suppincPsi}
\supp L(z^v) \subset  \Psi(v),  
\end{equation}
this inclusion being an equality for all but finitely many $v \in \Q$, {\it e.g.}, for all $v \in \Q$ such that the exponents $v \ellmahl^{i} +j$ involved in \eqref{eq:Lzv} are two by two distinct, because there is no cancellation between terms on the right-hand side of  \eqref{eq:Lzv} in this case. It follows immediately from these remarks that 
$$\minP(v) \leq \val L(z^v)
$$ 
and that this inequality is an equality for all but finitely many $v \in \Q$ (actually, Lemma \ref{lem:val Lf} below ensures that this is the case for all $v \in \Q \setminus -\S(L)$). 
Last, it is easily seen that, for all but finitely many $q \in \Q$, the equation $\val L(z^{w})=q$ has a unique solution $w \in \Q$ and that it is given by $w=\pi(q)$.

 It will be convenient to set 
$$
\minP(+\infty)=\pi(+\infty)=+\infty. 
$$

We shall now give several properties of the maps $\minP$ and $\pi$ which will shall use later.

\begin{lem}\label{lem proprietes pi minP}
	The maps $\minP : \mathbb Q \rightarrow \mathbb Q$ and $\pi : \mathbb Q \rightarrow \mathbb Q$ are \strictementcroissant
\footnote{In the whole paper, a function $f:\mathbb Q \rightarrow \mathbb Q$ is said to be \strictementcroissant{} if, for all $x,y \in \mathbb Q$, ($y>x \Rightarrow f(y)>f(x)$). It is \croissant{} if, for all $x,y \in \mathbb Q$, ($y \geq x \Rightarrow f(y) \geq f(x)$). We will use similar terminologies for sequences of real numbers or of sets.} bijections and inverse of each other. 	\end{lem}
	
\begin{proof}
	The fact that these maps are \strictementcroissant{}
 is immediate from their definitions. In particular, $\minP$ and $\pi$ are injective. In order to prove that $\minP$ and $\pi$ are inverse of each other, it is thus sufficient to prove that $\minP(\pi(q))=q$ for every $q \in \Q$. Let us prove this. On the one hand, by definition of $\pi$, there exists $(\ellmahl^{i},j) \in \P(L)$ such that 
	$$
	 \pi(q)=  
	 \frac{q-j}{\ellmahl^{i}}.
	$$
	Since $(\ellmahl^{i},j) \in \P(L)$, it follows from the definition of $\minP$ that 
	$$
	\minP(\pi(q)) 
	\leq \ellmahl^{i} \pi(q) + j
	= \ellmahl^{i} \frac{q-j}{\ellmahl^{i}} + j=q.
	$$
On the other hand, by definition of $\minP$, there exists $(\ell^{i'},j') \in \P(L)$ such that $$
\minP(\pi(q)) =\ellmahl^{i'} \pi(q) + j'.
$$ 
But, since $(\ellmahl^{i'},j') \in \P(L)$, it follows from the definition of $\pi$ that 
$$
	\pi(q) \geq \frac{q-j'}{\ellmahl^{i'}}.
$$

So, we have 
$$\minP(\pi(q))
	\geq \ellmahl^{i'}  \frac{q-j'}{\ellmahl^{i'}} + j' = q.
	$$ Finally, we obtain $\minP(\pi(q)) =q$.
\end{proof}

\begin{lem}\label{lem psi}
Consider $v\in\Q$. If $k \in \{1,\ldots,\kappa+1\}$ is such that $-\mu_{k} \leq v \leq-\mu_{k-1}$, with the conventions $\mu_{0}=-\infty$  and $\mu_{\kappa+1}=+\infty$, then 
\begin{equation}\label{eq:formule pour minP}
\minP( v)= \ellmahl^{\alpha_{k-1}} v + \beta_{k-1} \,.
\end{equation}
In geometric terms, the formula \eqref{eq:formule pour minP} means that $\minP(v)$ is the ordinate of the projection of  $\x_{k-1}=(\ellmahl^{\alpha_{k-1}} ,\beta_{k-1})$ along a line of slope $-v$ on the $y$-axis.
\end{lem}

\begin{proof}
By definition, $\minP(v)  =\min \{\ellmahl^{\alphaa} v + \betaa  \ \vert \  (\ellmahl^{\alphaa},\betaa) \in \P(L)\}$ so, in order to prove the lemma, it is sufficient  to prove that, for all $(\ellmahl^{\alphaa},\betaa) \in \P(L)$, for all $v \in [-\mu_{k} , -\mu_{k-1}]$, $\ellmahl^{\alphaa} v + \betaa \geq  \ellmahl^{\alpha_{k-1}}  v+\beta_{k-1}$. In other terms, we have to prove that, for $(\ellmahl^{\alphaa},\betaa) \in \P(L)$, the map
	\begin{eqnarray*}
		\delta : \R & \rightarrow & \R \\
		v &\mapsto & \ellmahl^{\alphaa} v + \betaa - (\ellmahl^{\alpha_{k-1}}  v+\beta_{k-1} )
	\end{eqnarray*}
takes nonnegative values on $[-\mu_{k} , -\mu_{k-1}]$. Let us prove this.

\medskip
 Let us first assume that $k \in \{2,\ldots,\kappa\}$.  
	It follows from Lemma \ref{lem:ineq slopes N(L)}  (applied with $k-1$ instead of $k$ for the second inequality)  that 
	\begin{equation*}
	\label{eq:ineg_pentes}
	-\ellmahl^{\alphaa} \mu_{k} + \betaa \geq -\ellmahl^{\alpha_{k-1}}  \mu_{k} +\beta_{k-1}
	 \ \text{ and } -\ellmahl^{\alphaa} \mu_{k-1} + \betaa \geq -\ellmahl^{\alpha_{k-1}}  \mu_{k-1} +\beta_{k-1}.
	\end{equation*}
	Thus $\delta(-\mu_{k}) \geq 0$ and $\delta(-\mu_{k-1})\geq 0$. Since $\delta$ is affine, it follows that, for all $v \in [-\mu_{k} , -\mu_{k-1}]$, $\delta(v) \geq 0$, as wanted.

\medskip
	Let us now consider the case $k=1$. It follows from Lemma \ref{lem:ineq slopes N(L)} that 
	\begin{equation*}
	 -\ellmahl^{\alphaa} \mu_{1} + \betaa \geq -\ellmahl^{\alpha_{0}}  \mu_{1} +\beta_{0}.
	\end{equation*} 
	In other words, $\delta(-\mu_{1}) \geq 0$.  Since $\alphaa_0=0$, the function $\delta$ is either increasing or constant. Thus, for any $v \geq -\mu_1$, $\delta(v)\geq \delta(-\mu_1)\geq 0$, as wanted.

	\medskip
	Let us eventually consider the case $k=\kappa+1$. It follows from Lemma \ref{lem:ineq slopes N(L)} that 
	\begin{equation*}
	 -\ellmahl^{\alphaa} \mu_{\kappa} + \betaa \geq -\ellmahl^{\alpha_{\kappa}}  \mu_{\kappa} +\beta_{\kappa}.
	\end{equation*}
	In other words, $\delta(-\mu_{\kappa}) \geq 0$. Since $\alphaa_{\kappa}=n \geq \alpha$, the function $\delta$ is either decreasing or constant. Thus, for any $v \leq -\mu_{\kappa}$, $\delta(v)\geq \delta(-\mu_{\kappa})\geq 0$, as wanted.
\end{proof}

\begin{ex}
An illustration of Lemma \ref{lem psi} for the operator $L$ given by \eqref{L pour example} is given by Figure \ref{fig:rudin shapiro polygon} in Section \ref{illustration main algo on RS}. Indeed, with the hypotheses and notations of Figure \ref{fig:rudin shapiro polygon}, we have $\kappa=2$, $\mu_{1}=0$, $\mu_{2}=1/2$ and Lemma \ref{lem psi} ensures that:
\begin{itemize}
 \item if $0\leq v$, then $\minP(v)$ is the ordinate of the projection of  $\x _{0}$ along a line of slope $-v$ on the $y$-axis;
 \item if $-1/2 \leq v \leq 0$, then $\minP(v)$ is the ordinate of the projection of  $\x _{1}$ along a line of slope $-v$ on the $y$-axis;
  \item if $ v \leq -1/2$, then $\minP(v)$ is the ordinate of the projection of  $\x _{2}$ along a line of slope $-v$ on the $y$-axis. 
\end{itemize}
This is indeed what we see on Subfigure \ref{fig:subfigA} (resp. \ref{fig:subfigB}, \ref{fig:subfigC}) when $v=1/4$ (resp. $-1/4$, $-3/4$).
\end{ex}

\begin{lem}\label{lem pi}
 If $-\mu_{k} \leq \pi(q) \leq-\mu_{k-1}$ for some $k \in \{1,\ldots,\kappa+1\}$ with the conventions $\mu_{0}=-\infty$ and $\mu_{\kappa+1}=+\infty$,  then 
\begin{equation}\label{eq: form pour pi(q)}
 \pi(q)= \frac{ q - \beta_{k-1}}{\ellmahl^{\alpha_{k-1}}} \,.
\end{equation}
In geometric terms, the formula \eqref{eq: form pour pi(q)} means that $\pi(q)$ is the opposite of the slope of the line passing through $(0,q)$ and $\x_{k-1}=(\ellmahl^{\alpha_{k-1}},\beta_{k-1})$.
\end{lem}

\begin{proof}
Applying Lemma \ref{lem psi} with $v= \pi(q)$ we obtain 
$$
\minP(\pi(q))=\ellmahl^{\alpha_{k-1}} \pi(q) + \beta_{k-1}.
$$
The result follows from this formula since, according to Lemma \ref{lem proprietes pi minP}, we have $q=\minP(\pi(q))$. 
\end{proof}

\subsection{Supports, valuations and the maps $\Psi$, $\minP$ and $\pi$}\label{sec:slopes and sol}
Roughly speaking, the results presented in this section aim to relate the valuations and the supports of $L(f\pz)$ and of $f\pz$ {\it via} the maps $\Psi$, $\minP$ and $\pi$ introduced in section \ref{sec:the maps Psi minP pi}. These results will be used extensively in the remainder of this paper.

\begin{lem}\label{lem:val Lf}
For any $f\pz  \in \Hahn$, we have 
$$
\val L(f\pz ) \geq \minP(\val f\pz ).
$$
If $\val f\pz  \not \in -\mathcal{S}(L)$, then  
 $$
 \val L(f\pz ) = \minP(\val f\pz ).$$
\end{lem}

\begin{proof}
The result is obvious if $f\pz=0$. In the rest of the proof, we assume that $f\pz \neq 0$. 
As a preliminary remark, note the following obvious but important formula:
\begin{equation}\label{formula psi val}
 \min_{i \in \{0,\ldots,n\}} \val(a_{i}  \malop{\ellmahl}^{i}(f)) 
 =\min_{i \in \{0,\ldots,n\}} \val a_{i}\pz +\ellmahl^{i} \val f\pz  
 = \minP(\val f\pz ).
\end{equation}

We have 
\begin{equation}\label{eq tt a gauche bis bis}
L(f)=a_{n}  \malop{\ellmahl}^{n}(f) + a_{n-1}  \malop{\ellmahl}^{n-1}(f) + \cdots + a_{0}  f.
\end{equation}
Using \eqref{val ult metric}, we get  
$$
\val L(f) \geq \min_{i \in \{0,\ldots,n\}} \val(a_{i} \malop{\ellmahl}^{i}(f)). 
$$
Combining the latter inequality with \eqref{formula psi val}, we get the first assertion of the Lemma.

Let us prove the contrapositive of the second assertion. Assume that  $\val L(f\pz ) \neq \minP(\val f\pz )$. The first part of the Lemma ensures that $\val L(f\pz ) > \minP(\val f\pz )$.  
Using \eqref{formula psi val}, the latter inequality can be rewritten as 
$$  
\val L(f )>\min_{i \in \{0,\ldots,n\}} \val(a_{i}  \malop{\ellmahl}^{i}(f)). 
$$
Combining the latter inequality with \eqref{eq tt a gauche bis bis}, we see that there exist distinct indices $i_{1},i_{2} \in \{0,\ldots,n\}$ such that 
$$ 
\val(a_{i_{1}}  \malop{\ellmahl}^{i_{1}}(f)) =\val(a_{i_{2}} \malop{\ellmahl}^{i_{2}}(f)) 
 =\min_{i \in \{0,\ldots,n\}} \val(a_{i}  \malop{\ellmahl}^{i}(f)), 
$$
{\it i.e.}, such that
$$
 \val a_{i_{1}}\pz +\ellmahl^{i_{1}} \val f\pz =\val a_{i_{2}}\pz +\ellmahl^{i_{2}} \val f\pz 
 =\min_{i \in \{0,\ldots,n\}} \val a_{i}\pz +\ellmahl^{i} \val f\pz . 
$$
 Using Lemma \ref{rem charact slopes 1}, we see that $\val f\pz \in -\mathcal{S}(L)$. 
\end{proof}

\begin{lem}\label{lem:val Lf eg g}
 For any solution $f\pz  \in \Hahn$ of \eqref{eq mahl slopes}, we have  
	$$
 \val f\pz  \in \{\pi(\val a_{-\infty}\pz)\} \cup -\S(L).$$
\end{lem}

\begin{proof}
If $f\pz=0$, then $a_{-\infty}\pz=L(f\pz)=L(0)=0$ and, hence,  $\val f\pz  =\val a_{-\infty}\pz = +\infty$. This proves the lemma in this case. 

From now on, we assume that $f\pz \neq 0$. 
We have
\begin{equation}\label{eq tt a gauche}
 a_{n} \malop{\ellmahl}^{n}(f)  + a_{n-1}  \malop{\ellmahl}^{n-1}(f) + \cdots + a_{0}  f  - a_{-\infty} =0. 
\end{equation}
The equality \eqref{eq tt a gauche} ensures that there exist distinct indices $i_{1},i_{2} \in \{-\infty,0,\ldots,n\}$ such that 
$$ 
\val(a_{i_{1}} \malop{\ellmahl}^{i_{1}}(f) ) =\val(a_{i_{2}} \malop{\ellmahl}^{i_{2}}(f)) 
 =\min_{i \in \{-\infty,0,\ldots,n\}} \val(a_{i}  \malop{\ellmahl}^{i}(f) ), 
$$
{\it i.e.}, such that
$$
 \val a_{i_{1}}\pz +\ellmahl^{i_{1}} \val f\pz  =\val a_{i_{2}}\pz +\ellmahl^{i_{2}} \val f\pz = \min_{i \in \{-\infty,0,\ldots,n\}} \val a_{i}\pz +\ellmahl^{i} \val f\pz . 
$$
In what precedes, when $i=-\infty$, by $\val(a_{i} \malop{\ellmahl}^{i}(f) )$ and by $\val a_{i}\pz +\ellmahl^{i} \val f\pz$, we mean $\val a_{-\infty}\pz$. 
 Using Lemma \ref{rem charact slopes 1}, we see that $\val f\pz  \in -\mathcal{S}(L,a_{-\infty}\pz )$. We conclude by using the following inclusion that follows directly from the definitions:
 $$
 -\mathcal{S}(L,a_{-\infty}\pz ) \subset \{\pi(\val a_{-\infty}\pz)\} \cup -\S(L).
 $$ 
\end{proof}

\begin{coro}\label{sol eq si eq on - pentes}
		Two solutions $f\pz,g\pz \in \Hahn$ of \eqref{eq mahl slopes} are equal if and only if they are equal on $-\S(L)$. 		
\end{coro}
		
		 \begin{proof}
		 If two solutions $f\pz$ and $g\pz$ of \eqref{eq mahl slopes} are equal on $-\S(L)$, then $h\pz=g\pz-f\pz$ is equal to $0$ on $-\S(L)$ and, hence, $\val h\pz \not \in -\S(L)$. But, since $h\pz$ satisfies  $L(h\pz)=0$,  Lemma \ref{lem:val Lf eg g} ensures that $\val h\pz \in \{+\infty\} \cup -\S(L)$. Therefore, $\val h\pz = +\infty$ and, hence, $h\pz=0$, {\it i.e.}, $f\pz=g\pz$ as claimed. 
		\end{proof}

\begin{lem}\label{lem:supp Lf sub P supp f}
For any $f\pz  \in \Hahn$, we have $\supp L(f\pz) \subset \bigcup_{\gamma \in  \supp f\pz} \Psi(\gamma)$.
\end{lem}

\begin{proof}	
We set, for any $i \in \{0,\ldots,n\}$, $a_{i}\pz=\sum_{j \in \supp a_{i}\pz} a_{i,j}z^{j}$. We have 
$$
L(f\pz)= \sum_{i=0}^{n} a_{i}\malop{\ellmahl}^{i}(f)=\sum_{i=0}^{n}\sum_{j \in \supp a_{i}\pz} a_{i,j}z^{j}\malop{\ellmahl}^{i}(f),
$$
so 
\begin{eqnarray*}
 \supp L(f\pz) &\subset& \bigcup_{i=0}^{n}\bigcup_{j \in \supp a_{i}\pz} \supp z^{j}\malop{\ellmahl}^{i}(f) \\ 
&=&\bigcup_{i=0}^{n}\bigcup_{j \in \supp a_{i}\pz} \bigcup_{\gamma \in \supp f\pz} \{\gamma\ellmahl^{i}+j\}\\
&=&\bigcup_{\gamma \in \supp f\pz} \bigcup_{i=0}^{n}\bigcup_{j \in \supp a_{i}\pz}  \{\gamma\ellmahl^{i}+j\}\\ 
&=&\bigcup_{\gamma \in \supp f\pz} \{\gamma\ellmahl^{i}+j \ \vert \ (\ellmahl^{i},j) \in \P(L)\} = \bigcup_{\gamma \in \supp f\pz} \Psi(\gamma). 
\end{eqnarray*}
\end{proof}

\subsection{A map of fundamental importance}

The map 
\begin{equation}\label{eq:pi circ Psi}
 \begin{array}{rclc}
\Q &\rightarrow& \{\,\text{Finite subsets of $\Q$}\,\}  \\
     v &\mapsto& \pi(\Psi(v))
\end{array}
\end{equation}
will play central role in this paper. For a graphic illustration of this map for the operator $L$ given by \eqref{L pour example}, we refer to Figure \ref{fig:rudin shapiro polygon}.  
 
Let us briefly explain how this map naturally arises when we seek to understand the support of the Hahn series solutions of Mahler equations. 
Consider $f\pz=\sum_{\gamma \in \Q} f_{\gamma}z^{\gamma} \in \Hahn$ such that $L(f)=0$. To simplify the presentation, we assume that $\supp f$ has at least two elements and we ask: what are the possible values for the two least elements $\gamma_{0}<\gamma_{1}$ of $\supp f$?  
We have $\gamma_{0} = \val f $, $\gamma_{1}=\val f_{\vert  \Q_{>\gamma_{0}}}$ and 
$$
f\pz =f_{\vert \Q_{>\gamma_{1}}}\pz +f_{\gamma_{1}}z^{\gamma_{1}} + f_{\gamma_{0}}z^{\gamma_{0}}. 
$$ 

What are the possible values for $\gamma_{0}$? Lemma \ref{lem:val Lf eg g} provides an immediate answer to this question: $\gamma_{0} \in -\S(L)$. 

What are the possible values for $\gamma_{1}$? In order to answer this question, note that  
the equations $L(f\pz)=0$ and $f=f_{\vert \Q_{>\gamma_{0}}}\pz+f_{\gamma_{0}}z^{\gamma_{0}}$ imply $L(f_{\vert \Q_{>\gamma_{0}}}\pz)=-f_{\gamma_{0}}L(z^{\gamma_{0}})$. Lemma \ref{lem:val Lf eg g} ensures that  $\gamma_{1}=\val f_{\vert \Q_{>\gamma_{0}}}\pz \in \{\pi(\val L(z^{\gamma_{0}}))\} \cup -\S(L)$. But, it follows from Lemma \ref{lem:supp Lf sub P supp f} that $\supp L(z^{\gamma_{0}}) \subset \Psi(\gamma_{0})$. So, $\gamma_{1} \in \pi( \Psi(\gamma_{0})) \cup -\S(L)$. Since $\gamma_{0} \in -\S(L)$ and since $-\S(L)\subset \bigcup_{v\in -\S(L)}\pi( \Psi(v))$ as a consequence of the first assertion of Lemma \ref{min Psi(s)=s} below, we get 
$$\gamma_{1} \in \bigcup_{v\in -\S(L)}\pi( \Psi(v)).$$
We see the map \eqref{eq:pi circ Psi} naturally appear here. 
Iterating this in order to reach more and more elements of the support of $f$, we can guess that the map \eqref{eq:pi circ Psi} and its iterates should play a central role in the study of the support of the Hahn series solution of \eqref{eq mahl slopes}. We will see in section \ref{sec:recept supp sol L bis} and, more precisely, in Theorem \ref{theo: un support bis} that this is indeed the case.

Note the following result for further use.

\begin{lem}\label{min Psi(s)=s}
 For all $v\in \Q$, we have 
 $$v \in \pi(\Psi(v))$$
 and 
\begin{equation}\label{eq:min Psi(s)=s}
  \min \pi(\Psi(v))=\pi(\min \Psi(v))=\pi(\minP(v))=v.
\end{equation}
\end{lem}

\begin{proof}
In order to prove the lemma, it is sufficient to prove \eqref{eq:min Psi(s)=s}. The first equality in \eqref{eq:min Psi(s)=s} follows from the fact that $\pi$ is increasing by Lemma \ref{lem proprietes pi minP}. The second equality in \eqref{eq:min Psi(s)=s} follows from the definition of $\minP$. The third equality in \eqref{eq:min Psi(s)=s} follows form the fact that $\pi$ and $\minP$ are inverse of each other by Lemma~\ref{lem proprietes pi minP}.
\end{proof}

\section{A receptacle $\V$ for the support of the solutions}\label{sec:recept supp sol L bis}

Throughout this section, we consider a Mahler operator 
\begin{equation}\label{operateur L}
L= a_{n}\pz \malop{\ellmahl}^{n} + a_{n-1}\pz\malop{\ellmahl}^{n-1} + \cdots + a_{0}\pz
\end{equation}
with coefficients $a_{0}\pz,\ldots,a_{n}\pz \in \bK[z]$ such that $a_{0}\pz a_{n}\pz \neq 0$. 
We let  
$$
\Sol(L,\Hahn) = \{f\pz  \in \Hahn \ \vert \ L(f\pz )=0 \}
$$ 
be the $\bK$-vector space of solutions of $L$ in $\Hahn$. 
In what follows, we will use the following terminology:

\begin{defi}\label{def comp bis}
	We say that a subset $\Qr$ of $\Q$ is \textit{computable} if there exists an algorithm which takes a rational number $q$ as input and returns whether it belongs to $\Qr$ or not. 
\end{defi}
The aim of this section is to describe a computable well-ordered subset $\V$ of $\Q$ containing the support of any Hahn series solution of $L$, {\it i.e.}, such that 
\begin{equation}\label{sol inc h V}
 \Sol(L,\Hahn) \subset \Hahn_{\vert \V}
\end{equation}  
(and satisfying a technical but important stability condition with respect to the map \eqref{eq:pi circ Psi}).

\begin{theo}\label{theo: un support bis}
Let $(\W_{i})_{i\geq 0}$ be the sequence of finite subsets of $\Q$ defined as follows:
\begin{itemize}
 \item $\W_{0}=-\S(L)$;
 \item $\forall i \geq 0$, $\W_{i+1}=\bigcup_{v \in \W_{i}}\pi (\Psi(v))$.
\end{itemize}
The sequence $(\V_{i})_{i\geq 0}$ is \croissant{}
 and the set 
$
\W=\bigcup_{i\geq 0} \W_{i}
$  
has the following properties: 
\begin{enumerate}
 \item \label{supp sol dans S bis}
 $
\Sol(L,\Hahn) \subset \Hahn_{\vert \W}
$; 
 \item \label{S well ordered bis} $\W$ is well-ordered; 
 \item \label{S computable bis} $\W$ is computable;
 \item \label{S contient opp pentes bis} $-\S(L) \subset \W$;
 \item \label{stab pi Psi bis}  
 $
\bigcup_{v \in \W} \pi (\Psi(v)) = \W
$.
 \end{enumerate}
\end{theo}

Property \eqref{S contient opp pentes bis} is obvious and Property \eqref{stab pi Psi bis} follows immediately from Lemma \ref{min Psi(s)=s} and the construction of $\V$. They will be freely used in the proofs of the other assertions of Theorem \ref{theo: un support bis} given in the next subsections.  
The fact that $(\V_{i})_{i\geq 0}$ is \croissant{} is proved in section \ref{sec:proof Vi increasing}. Section \ref{sec:basic prop V} gives a couple of basic properties of $\V$ that will be used for the proofs of properties \eqref{supp sol dans S bis}, \eqref{S well ordered bis} and \eqref{S computable bis} of Theorem \ref{theo: un support bis} but also latter in the paper.  Properties \eqref{supp sol dans S bis}, \eqref{S well ordered bis} and \eqref{S computable bis} of Theorem \ref{theo: un support bis} are proved in sections \ref{sec:proof supp sol dans S bis}, \ref{sec:proof S well ordered bis} and \ref{sec:proof S computable bis} respectively.  

\begin{rem}
The existence  of a subset $\V$ of $\Q$ satisfying properties \eqref{supp sol dans S bis} and \eqref{S well ordered bis} is obvious:  the support of any element of $\Sol(L,\Hahn)$ is included in the union of the supports of the elements of an arbitrary basis of $\Sol(L,\Hahn)$ (which is a finite dimensional sub-$\bK$-vector space of $\Hahn$ with dimension at most $n$). But, this is not sufficient for our purpose, all the properties listed in Theorem~\ref{theo: un support bis} will be used. 
In particular: 
\begin{itemize}
 \item property \eqref{S computable bis} is of fundamental importance for the algorithmic considerations of this paper; it is one of the crucial ingredients that makes our answer to Question \ref{main question: un algo pour sol hahn?}, given by Theorem \ref{theo algo}, possible;
 \item  property \eqref{S contient opp pentes bis} is a precaution to ensure that property \eqref{supp sol dans S bis} is satisfied;
\item property \eqref{stab pi Psi bis} is essential in section \ref{sec:approx sol} to build the set $\Rr$ mentioned in the introduction.
\end{itemize}
\end{rem}

\begin{ex}\label{ex:Vi pour L pour example}
The sets $\V_{0},\V_{1},\V_{2}$ are computed in section \ref{the Vi for RS example} for the operator $L$ given by \eqref{L pour example}.
\end{ex}

In what follows, we will use the following notations: 
\begin{itemize}
 \item we continue with the notations $\mu_{1},\ldots,\mu_{\kappa}$, $\alphaa_k$, $\betaa_k$, {\it etc},  from section~\ref{sec: ext newt pol};
\item we let $d \in \Z_{\geq 1}$ be a common multiple of the denominators of the slopes $\mu_{1}, \ldots, \mu_{\kappa}$; 
 \item for any subset $\Qr$ of $\mathbb Q$ and any $\gamma \in \mathbb Q$, we set $\Qr_{> \gamma}=\Qr \cap \Q_{> \gamma}$, $\Qr_{\geq \gamma}=\Qr \cap \Q_{\geq \gamma}$, $\Qr_{< \gamma}=\Qr \cap \Q_{< \gamma}$ and $\Qr_{\leq \gamma}=\Qr \cap \Q_{\leq \gamma}$.
\end{itemize}

\subsection{Proof of the fact that $(\V_{i})_{i\geq 0}$ is \croissant{} in Theorem \ref{theo: un support bis}}\label{sec:proof Vi increasing}
Consider $i \in \Z_{\geq 0}$ and $v \in \V_{i}$. Lemma \ref{min Psi(s)=s} ensures that $v \in \pi(\Psi(v))$. But, $\pi(\Psi(v)) \subset \pi(\bigcup_{w \in \W_{i}}\Psi(w))=\V_{i+1}$. So, $v \in \V_{i+1}$. This shows that $\V_{i} \subset \V_{i+1}$.

\subsection{Basic properties of $\V$}\label{sec:basic prop V}

\begin{lem}\label{min V=-muk}
The set $\W$ has a minimal element given by $\min \W=\min \W_{0}=\min -\S(L)=-\mu_{\kappa}$. 
\end{lem}

\begin{proof}
We claim that, for all $i \in \Z_{\geq 0}$, 
\begin{equation}\label{eq:min Vi+1 Vi}
 \min \W_{i+1}=\min \W_{i}.
\end{equation}
Indeed, for all $i \in \Z_{\geq 0}$, we have $\W_{i+1}=\pi (\bigcup_{v \in \W_{i}}\Psi(v))=\bigcup_{v \in \W_{i}}\pi (\Psi(v))$. But, Lemma \ref{min Psi(s)=s} ensures that,  for all $v \in \Q$, $\min \pi(\Psi(v))=v$. So,  
$$\min \W_{i+1}=\min \bigcup_{v \in \W_{i}}\pi (\Psi(v)) =\min_{v \in \W_{i}} \min \pi (\Psi(v))=\min_{v \in \W_{i}} v=\min \W_{i},$$
whence our claim. Now, the equality $\min \W=\min \W_{0}$ follows clearly from the fact that $\W=\bigcup_{i\geq 0} \W_{i}$ and from \eqref{eq:min Vi+1 Vi}. The equality $\min \W_{0}=\min -\S(L)$ follows from the fact that $\W_{0}=-\S(L)$ by definition. The equality $\min -\S(L)=-\mu_{\kappa}$ follows from the fact that $\mu_{\kappa}$ is the greatest slope of $L$. 
\end{proof}

\begin{lem}\label{lem:exist v0 vM}
Consider $v \in \V$. Let $M$ be the least element of $\Z_{\geq 0}$ such that $v \in \V_{M}$.  There exist $v_{0}<\cdots<v_{M-1}<v_{M}=v$ in $\W$ such that 
	\begin{itemize}
		\item $v_{0} \in -\S(L)$;
		\item $v_{i+1} \in \pi ( \Psi(v_{i}))$ for any $i \in \{0,\ldots,M-1\}$.
	\end{itemize}
\end{lem}
\begin{proof}
The case $M=0$ being obvious, we will assume in the remainder of this proof that $M \geq 1$. 
It follows immediately from the definition of the $\V_{i}$ that there exist $v_{0},\ldots,v_{M-1},v_{M}=v$ in $\V$ such that 
\begin{itemize}
		\item $v_{0} \in -\S(L)$;
		\item $v_{i+1} \in \pi(\Psi(v_{i}))$ for any $i \in \{0,\ldots,M-1\}$.
	\end{itemize}
 Lemma \ref{min Psi(s)=s} guarantees that, for all $i \in \{0,\ldots,M-1\}$, $\min \pi( \Psi(v_{i}))=v_{i}$, so $v_{i+1} \geq v_{i}$. Therefore, we have $v_{0}\leq \cdots \leq v_{M}$. If one these inequalities were an equality, then $v_{M}$ would belong to $\V_{M-1}$ and this would contradict the minimality of $M$. Thus, $v_{0}<\cdots<v_{M}$. This concludes the proof. 
\end{proof}

\subsection{Proof of \eqref{supp sol dans S bis} of Theorem \ref{theo: un support bis}}\label{sec:proof supp sol dans S bis}

Consider $f\pz \in \Sol(L,\Hahn)$. We want to prove that 
$\supp f\pz  \subset \W$. Assume on the contrary that $\supp f\pz  \not\subset \W$. Then, the set $\supp (f\pz) \setminus \W$ is  nonempty and well-ordered. In particular, we can consider 
$$
\gamma_{\min} = \min \left(\supp(f\pz ) \setminus \W\right).
$$ 
By minimality of $\gamma_{\min}$, we have 
\begin{equation}\label{supp < gamma inc V}
\supp f_{\vert \Q_{<\gamma_{\min}}}\pz  \subset \W. 
\end{equation}
Consider the decomposition 
$$
f\pz =f_{\vert \Q_{<\gamma_{\min}}}\pz +f_{\vert \Q_{\geq \gamma_{\min}}}\pz. 
$$ 
The equality $L(f)=0$ implies the equality
$
L(f_{\vert \Q_{\geq \gamma_{\min}}}\pz)=-L(f_{\vert \Q_{<\gamma_{\min}}}\pz)
$ 
and we infer from Lemma \ref{lem:val Lf eg g} that 
\begin{equation}\label{gammamin in}
\gamma_{\min} = \val f_{\vert \Q_{\geq \gamma_{\min}}}\pz \in \{\pi(\val L(f_{\vert \Q_{<\gamma_{\min}}}\pz))\}
 \cup -\S(L). 
\end{equation}
Since $\gamma_{\min} \in \supp(f\pz ) \setminus \V$ and $-\S(L) \subset \V$, we have 
	$\gamma_{\min} \notin -\S(L)$ so 
	 \eqref{gammamin in} gives
\begin{equation}\label{gammamin in bis}
\gamma_{\min} = \pi(\val L(f_{\vert \Q_{<\gamma_{\min}}}\pz)).
\end{equation} 
In particular, this implies that $L(f_{\vert \Q_{<\gamma_{\min}}}\pz)\neq 0$. 
Lemma \ref{lem:supp Lf sub P supp f} ensures that 
\begin{equation}\label{val M f<gammamin}
\val L(f_{\vert \Q_{<\gamma_{\min}}}\pz) \in \supp L(f_{\vert \Q_{<\gamma_{\min}}}\pz) \subset \bigcup_{\gamma \in \supp f_{\vert \Q_{<\gamma_{\min}}}\pz}\Psi(\gamma).
\end{equation}
Combining \eqref{gammamin in bis} and \eqref{val M f<gammamin}, we get  
\begin{equation}\label{eq gammamin in pi etc}
\gamma_{\min}  \in \pi\left(\bigcup_{\gamma \in \supp f_{\vert \Q_{<\gamma_{\min}}}\pz}\Psi(\gamma)\right). 
\end{equation}
The right-hand side of \eqref{eq gammamin in pi etc} is included in $\pi(\bigcup_{\gamma \in \V} \Psi(\gamma))$ by \eqref{supp < gamma inc V} and the latter set is equal to $\V$ by property \eqref{stab pi Psi bis} of Theorem \ref{theo: un support bis}. So, $\gamma_{\min} \in \V$. This is a contradiction.

\subsection{Proof of \eqref{S well ordered bis} of Theorem \ref{theo: un support bis}}\label{sec:proof S well ordered bis}
We argue by contradiction: we assume that $\W$ is not well-ordered. Thus, the set $D$  made up of the infinite \strictementdecroissant{} sequences with values in $\W$ is nonempty. Lemma \ref{min V=-muk} ensures that  $\W \subset \Q_{\geq -\mu_{\kappa}}$. Therefore, any element of $D$ has a limit in $\R_{\geq -\mu_{\kappa}}$. We let $E \subset \R_{\geq -\mu_{\kappa}}$ be the (nonempty) set made of these limits.   We set $w=\inf E \in \R_{\geq -\mu_{\kappa}}$. It is easily seen that $w \in E$. We let $(w_{m})_{m\geq 0} $ be an arbitrary element of $D$ tending to $w$. 

Note that, for all $m\geq 0$, $w_{m}>w_{m+1} \geq w$, so $w_{m} > w$.  Consider $k\in \{1,\ldots,\kappa\}$ such that $-\mu_{k} \leq w < -\mu_{k-1}$ with the convention $\mu_0 = -\infty$. Since $(w_{m})_{m\geq 0}$ is \strictementdecroissant{} and tends to $w$, up to replacing $(w_{m})_{m \geq 0}$ by $(w_{m+M})_{m\geq 0}$ with $M \in \Z_{\geq 0}$ large enough, we can assume that, for all $m \geq 0$, $-\mu_{k} < w_{m} < -\mu_{k-1}$, where the first inequality is legitimate because $w_{m} > w \geq -\mu_{k}$.

Lemma \ref{lem:exist v0 vM} ensures that, for any $M \geq 0$, there exist $r_{M}\in \Z_{\geq 0}$ and $v_{M,0}<\cdots<v_{M,r_{M}}=w_{M}$ in $\W$ such that 
	\begin{itemize}
		\item $v_{M,0} \in -\S(L)$;
		\item $v_{M,i+1} \in \pi ( \Psi(v_{M,i}))$ for any $i \in \{0,\ldots,r_{M}-1\}$.
	\end{itemize}

For any $M \geq 0$, we have $v_{M,r_{M}}=w_{M}>w$, so one can consider the least element $r'_{M}$ of $\{0,\ldots,r_{M}\}$ such that $v_{M,r'_{M}}>w$. 
It is important to notice that $r'_{M} \geq 1$ for all $M \geq 0$, 
because $-\mu_{k} \leq w < v_{M,r'_{M}} \leq  w_{M} < -\mu_{k-1}$, so $v_{M,r'_{M}} \not \in -\S(L)$ and, in particular, $v_{M,r'_{M}} \neq v_{M,0}$.
 The inequality $w_{M} \geq v_{M,r'_{M}}>w$ and the fact that $(w_{m})_{m\geq 0}$ tends to $w$ show  that $(v_{M,r'_{M}})_{M\geq 0}$ tends to $w$ and that, up to extracting a subsequence, we can assume that $(v_{M,r'_{M}})_{M\geq 0}$ is \strictementdecroissant. 	
	
	We set $q_{M} = \minP(v_{M,r'_{M}})$, so that $q_{M} \in \Psi(v_{M,r'_{M}})$ is such that $v_{M,r'_{M}} = \pi(q_{M})$ by Lemma \ref{lem proprietes pi minP}. 	 	
	 As $-\mu_{k} \leq w< v_{M,r'_{M}}= \pi(q_{M}) \leq w_{M} < -\mu_{k-1}$, it follows from Lemma \ref{lem pi} that, for all $M\geq 0$, 	
\begin{equation}\label{eq vnr'n qn}
v_{M,r'_{M}} = \pi(q_{M}) = \frac{q_{M} - \beta_{k-1}}{\ellmahl^{\alpha_{k-1}}} 
\end{equation}	
where $(\ellmahl^{\alpha_{k-1}},\beta_{k-1}) \in \P(L)$ is  the left endpoint of the edge of $\mathcal N(L)$ with slope $\mu_{k}$.  
	Moreover, since, for all $M \geq 0$,  $v_{M,r'_{M}} \in \pi(\Psi(v_{M,r'_{M}-1}))$ and since $\minP$ and $\pi$ are inverse of each other by Lemma  \ref{lem proprietes pi minP}, we have $q_{M}=\minP(v_{M,r'_{M}}) \in \minP(\pi( \Psi(v_{M,r'_{M}-1})))=\Psi(v_{M,r'_{M}-1})$ and, hence, there exists $(\ellmahl^{i_{M}},j_{M}) \in \P(L)$ such that  
	\begin{equation}\label{eq qn vnr'n-1}
 q_{M} = \ellmahl^{i_{M}}v_{M,r'_{M}-1} + j_{M}.
\end{equation}
	Since $\P(L)$ is finite, up to extracting a subsequence, we can assume that $i_{M} =\alphaa$ and $j_{M} = \betaa$ do not depend on $M\geq 0$. Combining \eqref{eq vnr'n qn} and \eqref{eq qn vnr'n-1}, we obtain, for all $M \geq 0$,  
	\begin{equation*}\label{eq:v_n-w_n}
	v_{M,r'_{M}} = \frac{\ellmahl^{\alphaa}v_{M,r'_{M}-1} + \betaa- \beta_{k-1}}{\ellmahl^{\alpha_{k-1}}}
	\end{equation*}
and, hence,
\begin{equation*}\label{eq:w_n-v_n}
	v_{M,r'_{M}-1} = \frac{\ellmahl^{\alpha_{k-1}}v_{M,r'_{M}} + \beta_{k-1}- \betaa}{\ellmahl^{\alphaa}}.
	\end{equation*}
The latter equality can be rewritten as $v_{M,r'_{M}-1} =\delta(v_{M,r'_{M}})$ where $\delta:\R\rightarrow \R$ is the \strictementcroissant{} affine function defined by $\delta(x)=\frac{\ellmahl^{\alpha_{k-1}}x+ \beta_{k-1}- \betaa}{\ellmahl^{\alphaa}}$.   Since $(v_{M,r'_{M}})_{M \geq 0}$ is \strictementdecroissant, this implies that $(v_{M,r'_{M}-1})_{M \geq 0}$ is \strictementdecroissant{} as well and, hence, belongs to $D$. Let $w'$ be the limit of $(v_{M,r'_{M}-1})_{M \geq 0}$. Since $(v_{M,r'_{M}-1})_{M \geq 0}$ is \strictementdecroissant{} and satisfies, for all $M\geq 0$, $v_{M,r'_{M}-1} \leq w$, we have $w' < w$. This contradicts the minimality of $w$ and concludes the proof of \eqref{S well ordered bis} of Theorem~\ref{theo: un support bis}.

\subsection{Proof of \eqref{S computable bis} of Theorem \ref{theo: un support bis}} \label{sec:proof S computable bis}

We set 
$$
\Z_{d,\ellmahl}=\bigcup_{i \in \Z_{\geq 0}} \frac{1}{d\ellmahl^{i}} \Z \subset \Q.
$$
Note that we do not require the integers $d$ and $\ellmahl$ to be coprime. We introduce the following two maps that will play an important role in this section: 
\begin{eqnarray*}
h \ \ : \ \  \Z_{d,\ellmahl} &\rightarrow& \Z_{\geq 0} \\
 v& \mapsto & h(v)= \min \big\{ i \in \Z_{\geq 0} \ \vert \ v \in \frac{1}{d\ellmahl^{i}}\Z \big\} 
\end{eqnarray*}
and 
\begin{eqnarray}\label{eq:epsilon_i}
\epsilon \ \ : \ \  \Q &\rightarrow& \Q_{>0} \cup \{+\infty\}\\
 v& \mapsto & \epsilon(v)=\min \V_{>v} - v \nonumber 
 \end{eqnarray}
with the convention $\epsilon(v)=+\infty$ if $\V_{>v}=\emptyset$. 
The fact that $\epsilon$ is well-defined and takes its values in $\Q_{>0} \cup \{+\infty\}$, {\it i.e.}, the fact that $\V_{>v}$ has a minimal element if it is not empty, follows from the fact that $\V$ is well-ordered.

\subsubsection{Computability of $\V$: description of an algorithm }
Let us temporarily admit the following result which will be proved in section \ref{sec:proof prop:majoration-iterationV}.
 
\begin{prop}\label{prop:majoration-iterationV}
For any $v \in \Q$, the following properties are equivalent: 
\begin{enumerate}
 \item \label{v in V} $ v \in \V$;
 \item \label{v in Vu} $v \in \Z_{d,\ellmahl}$, $v \geq -\mu_{\kappa}$ and $v \in \V_{\iota(v)}$ where 
\begin{equation}\label{def iota v}
 \iota(v)=\left\lfloor (n+1)\frac{v+\mu_{\kappa}}{\minepsilon} + h(v) \right\rfloor
\end{equation}
with 
\begin{equation}
\label{eq:epsilon}
\minepsilon:= \min\big\{\epsilon(-\mu_1),\ldots,\epsilon(-\mu_{\kappa}),(d\ellmahl^n)^{-1}\big\} \in \Q_{>0}. 
\end{equation}
\end{enumerate}
\end{prop}

The interest of this result lies in the fact that, while $\V$ is infinite, the set $\V_{\iota(v)}$ is finite and can be easily computed (directly from its definition) once $\iota(v)$ has been computed. 
Unfortunately, $\iota(v)$ is not easy to compute. More precisely, while in formula \eqref{def iota v} the values of $n$, $\mu_{\kappa}$ and $h(v)$ can be easily computed, the value of $\minepsilon$ cannot. This problem can be solved using section \ref{sec:algo plb epsilon} below, where an algorithm for computing a positive lower bound $\minepsilonlb$ on $\minepsilon$ is provided. Whence an upper bound 
$$
\iotaprime(v)=\left\lfloor (n+1)\frac{v+\mu_{\kappa}}{\minepsilonlb} + h(v) \right\rfloor
$$ on $\iota(v)$ for any $v \in \Z_{d,\ellmahl}$ such that $v \geq -\mu_{\kappa}$. Now, since $(\V_{i})_{i \geq 0}$ is \croissant, we have $\V_{\iota(v)} \subset \V_{\iotaprime(v)}$ and Proposition \ref{prop:majoration-iterationV} implies that the following properties are equivalent:
\begin{enumerate}
 \item $ v \in \V$;
 \item $v \in \Z_{d,\ellmahl}$, $v \geq -\mu_{\kappa}$ and $v \in \V_{\iotaprime(v)}$. 
\end{enumerate}

This leads to the following algorithm, which uses Algorithm \ref{algo:minepsilon} described at the end of section~\ref{sec:algo plb epsilon} for computing a positive lower bound  
$\minepsilonlb$ on $\minepsilon$.
\begin{algorithm} \label{algo v dans V}
$_{}$ \vskip 1 pt
\noindent \rule{\linewidth}{1mm}
{\bf Input}: $L$ a Mahler operator with coefficients in $\bK[z]$, $v \in \Q$. \\
{\bf Output}: whether or not $v$ is in $\V$.\vskip 1pt
\noindent \rule{\linewidth}{1mm}
\begin{itemize}
\item[] \textbf{if} $v \not \in \Z_{d,\ellmahl}$ or $v<-\mu_{\kappa}$ \textbf{then}
\begin{itemize}
 \item[] return ``$v$ is not in $\V$''
\end{itemize}
\item[] \textbf{otherwise} 
\begin{itemize}
\item[] compute $n$, $\mu_{\kappa}$ and $h(v)$
\item[] compute a positive lower bound $\minepsilonlb$ on $\minepsilon$ using Algorithm \ref{algo:minepsilon}
\item[] compute $\iotaprime=\left\lfloor (n+1)\frac{v+\mu_{\kappa}}{\minepsilonlb} + h(v) \right\rfloor$
\item[] compute $\V_{\iotaprime}$ 
\item[] \textbf{if} $v \in \V_{\iotaprime}$ \textbf{then} 
\begin{itemize}
 \item[] return ``$v$ is in $\V$''
\end{itemize} 
\item[] \textbf{otherwise} 
\begin{itemize}
\item[] return ``$v$ is not in $\V$'' 
\end{itemize}
\item[] \textbf{end if}
\end{itemize}
\item[] \textbf{end if}
\end{itemize}
\noindent \rule{\linewidth}{0.5mm}
\end{algorithm}

\subsubsection{Proof of Proposition \ref{prop:majoration-iterationV}}\label{sec:proof prop:majoration-iterationV} Proposition \ref{prop:majoration-iterationV} is proved at the very end of this section, after a series of lemmas.

\begin{lem}\label{lem:V_Zdl}
	We have 
	\begin{equation}\label{incl V Zdl}
\V \subset \Z_{d,\ellmahl}.
\end{equation}
Moreover, for any $v \in \Z_{d,\ell}$ and $w \in \pi(\Psi(v))$, we have $w \in \Z_{d,\ellmahl}$  and
\begin{equation}\label{eq h leq h n}
h(v) \leq h(w)+n. 
\end{equation}
\end{lem}

\begin{proof} 
Before proving  \eqref{incl V Zdl}, note that 
\begin{itemize}
 \item $\V_{0}\subset \Z_{d,\ellmahl}$; 
 \item for any $v \in \Z_{d,\ellmahl}$, $\pi(\Psi(v)) \subset \Z_{d,\ellmahl}$.
\end{itemize}
Indeed, the first property holds because $d$ is a common denominator of the elements of $\V_{0}=-\S(L)$. The second property holds because, by definition of $\pi$ and $\Psi$, for any $v \in \Z_{d,\ellmahl}$ and any $w\in \pi(\Psi(v))$, there exist $(\ellmahl^i,j),(\ellmahl^{i'},j') \in \P(L)$ such that 
	$$
          w=\frac{\ellmahl^i v + j - j'}{\ellmahl^{i'}} \in \frac{1}{\ellmahl^{i'}} \Z_{d,\ellmahl} = \Z_{d,\ellmahl}.
	$$
Let us now prove \eqref{incl V Zdl}. Since $\V=\cup_{i\geq 0} \V_{i}$, it is equivalent to prove that, for all $i \in \Z_{\geq 0}$, $\V_{i} \subset \Z_{d,\ellmahl}$. The latter property can be proved by induction on $i \in \Z_{\geq 0}$. Indeed, the two properties noticed at the beginning of the proof show that $\V_{0} \subset  \Z_{d,\ellmahl}$ and that if, for a given $i \in\Z_{\geq 0}$, we have $\V_{i} \subset \Z_{d,\ellmahl}$, then we have $\V_{i+1}=\bigcup_{v \in \W_{i}}\pi (\Psi(v)) \subset \Z_{d,\ellmahl}$.  

Let us now prove \eqref{eq h leq h n}. By definition of $h(w)$, we have $w \in \frac{1}{d\ellmahl^{h(w)}}\Z$, so 
	$$
v = \frac{\ellmahl^{i'} w  + j'- j}{\ellmahl^i} \in \frac{1}{\ellmahl^{i}}\frac{1}{d\ellmahl^{h(w)}}\Z=\frac{1}{d\ellmahl^{h(w)+i}}\Z.
$$
Thus, we have $h(v) \leq h(w)+i \leq h(w)+n$ and this proves \eqref{eq h leq h n}.

\end{proof}

\begin{rem} Consider two multiplicatively independent integers $\ellmahl_1,\ellmahl_2\geq 2$. In \cite{BorisAboutMahler},  Adamczewski and Bell give a proof of a conjecture of Loxton and van der Poorten asserting that any Puiseux series solution of both a $\ellmahl_1$-Mahler equation and of a $\ellmahl_2$-Mahler equation belongs to $\bigcup_{d \geq 1}\C(z^{1/d})$. As mentioned in the Introduction, an alternative proof was given later by Sh\"afke and Singer in \cite{SchafkeSinger}. The first part of Lemma \ref{lem:V_Zdl} can be used to extend this result to Hahn series. Indeed, let $f\pz \in \Hahn$ be solution of both a $\ellmahl_1$-Mahler equation and of a $\ellmahl_2$-Mahler equation with coefficients in $\C(z)$. Then, it follows from Lemma \ref{lem:V_Zdl} that $\supp f\pz$ is included in $\Z_{d_1,\ellmahl_1}$ and in $\Z_{d_2,\ellmahl_2}$ for some integers $d_1,d_2 \geq 1$. But, $\Z_{d_1,\ellmahl_1} \cap \Z_{d_2,\ellmahl_2} \subset \frac{1}{d} \Z$ for some integer $d \geq 1$ because $\ellmahl_1$ and $\ellmahl_2$ are multiplicatively independent. Thus, $f\pz$ is a Puiseux series solution of both a $\ellmahl_1$-Mahler equation and a $\ellmahl_2$-Mahler equation and, hence, $f\pz \in \bigcup_{d \geq 1}\C(z^{1/d})$.
\end{rem}

\begin{lem}\label{lem:steps}
Let $v \in \Z_{d,\ell}$ and $w \in \pi(\Psi(v))$ satisfy 
\begin{equation}\label{hyp encadrement muj}
 -\mu_{k} + \thetalb{k} \leq v < w <-\mu_{k-1}
\end{equation}
for some $k \in \{1,\ldots,\kappa\}$ and some lower bound $\thetalb{k}>0$ on $\epsilon(-\mu_{k})$. Then, at least one of the following properties holds:
	\begin{itemize}
		\item $w \geq v+\thetalb{k}$;
		\item $w \geq v + \frac{1}{d\ellmahl^n}$;
		\item $h(w)>h(v)$. 
	\end{itemize}
\end{lem}

\begin{proof}
By definition of $\Psi$,  there exists  $(\ellmahl^{\alphaa},\betaa) \in \P(L)$ such that  $w=\pi(\ellmahl^{\alphaa} v + \betaa)$.  Since $-\mu_{k} \leq w=\pi(\ellmahl^{\alphaa} v + \betaa) <-\mu_{k-1}$, Lemma \ref{lem pi} ensures that  
\begin{equation}\label{eq w en fonc de v}
w =\pi(\ellmahl^{\alphaa} v + \betaa)= \frac{\ellmahl^{\alphaa} v + \betaa - \beta_{k-1}}{\ellmahl^{\alpha_{k-1}}} \,.
\end{equation}
We now distinguish the cases $\alphaa=\alpha_{k-1}$, $\alphaa>\alpha_{k-1}$ and $\alphaa<\alpha_{k-1}$.

\vskip 5pt	
\noindent {\bf Case $\alphaa=\alpha_{k-1}$.} In this case, \eqref{eq w en fonc de v} can be rewritten as $w = v + \frac{\betaa - \beta_{k-1}}{\ellmahl^{\alpha_{k-1}}}$. Since $w>v$, we have $\betaa-\beta_{k-1}>0$. Since $\betaa$ and $\beta_{k-1}$ belong to $\Z$,  we have $\betaa-\beta_{k-1}\geq 1$. It follows that 
	 $$w \geq v+\frac{1}{\ellmahl^{\alpha_{k-1}}} \geq v+\frac{1}{\ellmahl^n}\geq v+\frac{1}{d\ellmahl^n}
	 $$ 
	 and the lemma holds in this case. 
	 
 \vskip 5pt	
\noindent {\bf Case $\alphaa >\alpha_{k-1}$.} In this case, since $(\ellmahl^{\alpha_{k-1}},\beta_{k-1})$ is the left endpoint of the edge of $\mathcal N(L)$ of slope $\mu_{k}$, the slope of the vector joining $(\ellmahl^{\alpha_{k-1}},\beta_{k-1})$ to $(\ellmahl^{\alphaa},\betaa)$ is greater than or equal to $\mu_{k}$, that is, 
	$$
	\frac{\betaa-\beta_{k-1}}{\ellmahl^{\alphaa} - \ellmahl^{\alpha_{k-1}}} \geq \mu_{k}. 
	$$ 
	So, we have
\begin{multline*}
 w - v =  \frac{(\ellmahl^{\alphaa}-\ellmahl^{\alpha_{k-1}}) v + \betaa - \beta_{k-1}}{\ellmahl^{\alpha_{k-1}}} 
	\geq 
	\frac{(\ellmahl^{\alphaa}-\ellmahl^{\alpha_{k-1}})v+(\ellmahl^{\alphaa}-\ellmahl^{\alpha_{k-1}})\mu_{k}}{\ellmahl^{\alpha_{k-1}}} \\
	=
	\frac{(\ellmahl^{\alphaa}-\ellmahl^{\alpha_{k-1}})(v+\mu_{k})}{\ellmahl^{\alpha_{k-1}}} 
	\geq \frac{(\ellmahl^{\alphaa}-\ellmahl^{\alpha_{k-1}})\thetalb{k}}{\ellmahl^{\alpha_{k-1}}}\geq \thetalb{k}
\end{multline*}	
and the lemma holds in this case. 
	
\vskip 5pt	
\noindent {\bf Case ${\alphaa}<\alpha_{k-1}$.} Write $v=\frac{M}{d\ellmahl^{h(v)}}$ with $M \in \Z$. Equation \eqref{eq w en fonc de v} can be rewritten as 
\begin{equation}\label{rewrit eq w}
w = \frac{M}{d\ellmahl^{h(v)+\alpha_{k-1}-\alphaa}} + \frac{\betaa-\beta_{k-1}}{\ellmahl^{\alpha_{k-1}}}. 
\end{equation}
We now distinguish two cases: 
\begin{itemize}
 \item if $h(v)-\alphaa>0$, then equation \eqref{rewrit eq w} shows that $h(w)= h(v)+\alpha_{k-1}-\alphaa>h(v)$ and the lemma holds; 
 \item if $h(v)\leq \alphaa$, then equation \eqref{rewrit eq w} shows that $h(w)\leq \alpha_{k-1}$; so $v$ and $w$ both belong to $\frac{1}{d\ellmahl^{\alpha_{k-1}}}\Z$; it follows that $w-v$ is a positive element of $\frac{1}{d\ellmahl^{\alpha_{k-1}}}\Z$, so $w-v \geq  \frac{1}{d\ellmahl^{\alpha_{k-1}}} \geq \frac{1}{d\ellmahl^n}$ and the lemma holds.
\end{itemize}
\end{proof}

\begin{lem}	\label{lem:longueur-chaine2} 
	Consider $k \in \{1,\ldots,\kappa\}$ and let $\thetalb{k}>0$   be a lower bound on $\epsilon(-\mu_{k})$.
	Consider $M \in \Z_{\geq 0}$ and $v_{0},\ldots,v_{M} \in \Z_{d,\ell}$ such that 	
\begin{itemize}
 \item $\forall i \in \{0,\ldots,M-1\}$, $v_{i+1} \in \pi(\Psi(v_i))$; 
 \item $-\mu_{k}+ \thetalb{k} \leq  v_{0} < \cdots < v_{M-1} < v_{M}  < -\mu_{k-1}$.
\end{itemize}
	 	Then, we have
	$$
	M < (n+1) \frac{v_{M} + \mu_{k}}{\min\left\{\thetalb{k},(d\ellmahl^n)^{-1}\right\}} + h(v_{M}).
	$$ 
\end{lem}

\begin{proof}  Set $$m_{k}=\min\big\{\thetalb{k},(d\ellmahl^n)^{-1}\big\}.$$
Consider 
\begin{equation}\label{def E1}
 E_1= \{i \in \{0,\ldots,M-1\} \ \vert \  h(v_{i})\geq  h(v_{i+1}) \}
\end{equation}
	 and 
\begin{equation}\label{def E2}
 E_2= \{i \in \{0,\ldots,M-1\} \ \vert \  h(v_{i}) <h(v_{i+1}) \}.
\end{equation}
Set $M_{1}=\sharp E_{1}$ and $M_{2}=\sharp E_{2}$. Since $\{E_{1},E_{2}\}$ is  a partition of $\{0,\ldots,M-1\}$, we have  
$$M=M_1+M_2.$$ 
	We have :
	\begin{itemize}
 \item $\forall i \in \{0,\ldots,M-1\}$, $h(v_{i}) \leq h(v_{i+1})+n$ by  Lemma \ref{lem:V_Zdl}; 
 \item if $i \in E_{2}$, then $h(v_i) < h(v_{i+1})$ and, hence, $h(v_i) \leq h(v_{i+1})-1$ because $h(v_i)$ and $h(v_{i+1})$ are integers.
\end{itemize}
It follows that 
$$
h(v_{0})\leq h(v_{M})+ n M_1 - M_2.
$$ 
Since $h(v_{0}) \geq 0$, we get  
\begin{equation}\label{ineg m2 en fction m1}
 M_2 \leq h(v_{M}) +nM_1.
\end{equation}
We shall now give an upper bound  on $M_{1}$.  	
On the one hand, since $-\mu_{k}+ \thetalb{k} \leq  v_{0} < \cdots < v_{M-1} < v_{M}$, 
we have 
\begin{equation}\label{sum 1}
 \sum_{i\in E_{1}} v_{i+1}-v_{i} \leq \sum_{i=0}^{M-1} v_{i+1}-v_{i} =v_{M}-v_{0} \leq v_{M}+\mu_{k} -  \thetalb{k} < v_{M}+\mu_{k}.
\end{equation}
On the other hand, for any $i \in E_{1}$, we have $h(v_i) \geq h(v_{i+1})$. Since  $-\mu_{k}+\thetalb{k} < v_i <v_{i+1}< -\mu_{k-1}$, Lemma \ref{lem:steps} ensures that $v_i \leq v_{i+1} - m_{k}$. It follows that 
\begin{equation}\label{sum 2}
\sum_{i\in E_{1}} v_{i+1}-v_{i}\geq \sum_{i\in E_{1}} m_{k} = M_{1} m_{k}. 
\end{equation}
Combining \eqref{sum 1} and \eqref{sum 2}, we get
	\begin{equation}\label{ineq pour M1}
 M_1 < \frac{v_{M} + \mu_{k}}{m_{k}}.
\end{equation}
Finally, combining \eqref{ineg m2 en fction m1} and \eqref{ineq pour M1}, we obtain  
	$$
	M = M_1 + M_2 \leq h(v_{M}) +(n+1)M_1 <  h(v_{M})+(n+1)\frac{v_{M} + \mu_{k}}{m_{k}}. 
	$$
\end{proof}

\begin{lem}	\label{lem:longueur-chaine1} 
Consider $M \in \Z_{\geq 0}$ and $v_{0},\ldots,v_{M} \in \V$ such that 
\begin{itemize}
 \item $v_0 < \cdots < v_{M-1} < v_{M}$;
 \item $\forall i \in \{0,\ldots,M-1\}$, $v_{i+1} \in \pi( \Psi(v_i))$. 
\end{itemize}
Then, 
	$$
	M \leq (n+1) \frac{v_{M} + \mu_{\kappa}}{\minepsilon} + h(v_{M})
	$$ 
where 
$$
\minepsilon= \min\left\{\epsilon(-\mu_1),\ldots,\epsilon(-\mu_{\kappa}),(d\ellmahl^n)^{-1}\right\}>0. 
$$ 
\end{lem}

\begin{proof}
As in the proof of Lemma \ref{lem:longueur-chaine2}, 
we consider the sets $E_{1}$ and $E_{2}$ defined by formulas \eqref{def E1} and \eqref{def E2} respectively, 
we set $M_{1}=\sharp E_{1}$ and $M_{2}=\sharp E_{2}$ and we have  
$$M=M_1+M_2$$ and 
$$
M_2 \leq h(v_{M}) +nM_1. 
$$

As in the proof of Lemma \ref{lem:longueur-chaine2}, we shall now give an upper bound  on $M_{1}$. We first claim that, for all $i \in E_{1}$, we have 
\begin{equation}\label{ineq vi vi+1 mineps}
v_i \leq v_{i+1} - \minepsilon. 
\end{equation}
In order to prove this claim, let us first recall that $\V \subset \Q_{\geq -\mu_{\kappa}}$ by Lemma \ref{min V=-muk}, so one of the following cases is satisfied. 

\vskip 5pt
\noindent {\textbf{Case 1: there exists $k \in \{1,\ldots,\kappa\}$ such that $-\mu_{k} \leq v_{i} <v_{i+1} < -\mu_{k-1}$.}} We distinguish the following two subcases.
\vskip 5pt
\noindent {\textbf{Subcase 1.1: $-\mu_{k}+\minepsilon \leq v_{i}$.}}
 Since $i \in E_{1}$, we have $h(v_i)\geq  h(v_{i+1})$. Moreover, we have $-\mu_{k}+\minepsilon \leq v_{i}<v_{i+1} < -\mu_{k-1}$ by hypothesis. So, Lemma \ref{lem:steps} ensures that $v_i \leq v_{i+1} - \minepsilon$ as claimed. 
\vskip 5pt
\noindent {\textbf{Subcase 1.2:  $v_{i}<-\mu_{k}+\minepsilon $.}} In this case, since $v_i \in \V$, the definition of $\epsilon(-\mu_{k})$ ensures that $v_i=-\mu_{k}$. Since $v_{i+1} \in \V$, we have $v_i=-\mu_{k} \leq v_{i+1} - \minepsilon$ as claimed. 
\vskip 5pt
\noindent {\textbf{Case 2: there exists $k \in \{1,\ldots,\kappa-1\}$ such that $v_i < -\mu_{k} \leq v_{i+1}$.}} We distinguish the following two subcases. 
\vskip 5pt
\noindent {\textbf{Subcase 2.1: $v_{i+1}=-\mu_{k}$.}} In this case, we have $v_{i+1} \in \frac{1}{d} \Z$, so $h(v_{i+1})=0$. 
Lemma \ref{lem:V_Zdl} implies $h(v_{i}) \leq n$, {\it i.e.}, $v_{i} \in \frac{1}{d \ellmahl^{n}} \Z$. Therefore, $v_{i}$ and $v_{i+1}$ are elements of $\frac{1}{d \ellmahl^{n}} \Z$ such that $v_{i}<v_{i+1}$, so we have $v_i \leq v_{i+1} - \frac{1}{d \ellmahl^{n}} \leq v_{i+1} - \minepsilon$ as claimed. 
\vskip 5pt
\noindent {\textbf{Subcase 2.2: $v_{i+1}> -\mu_{k}$.}} Since $v_{i+1} \in \V$, the definition of $\epsilon(-\mu_{k})$ ensures that  $v_{i+1}\geq -\mu_{k} + \minepsilon$. But, by hypothesis, $-\mu_{k}>v_{i}$. So, $v_{i+1}>v_{i}+ \minepsilon$ and our claim is proved in this case as well. 
\vskip 5pt
Now that the inequality \eqref{ineq vi vi+1 mineps} is justified, we can argue as we did in Lemma \ref{lem:longueur-chaine2} for proving \eqref{ineq pour M1} in order to prove that 
$$
M_1 \leq \frac{v_{M} + \mu_{\kappa}}{\minepsilon}.
$$ 
	
Finally, we obtain
	$$
	M= M_1 + M_2 \leq h(v_{M}) +(n+1)M_1 \leq  h(v_{M}) + (n+1)\frac{v_{M} + \mu_{\kappa}}{\minepsilon}. 
	$$
\end{proof}

\begin{proof}[Proof of Proposition \ref{prop:majoration-iterationV}] Let us prove that \eqref{v in V}  implies \eqref{v in Vu} in Proposition \ref{prop:majoration-iterationV}. 
		Consider $v \in \V$. Lemma \ref{lem:V_Zdl} ensures that $v \in \Z_{d,\ellmahl}$. Lemma \ref{min V=-muk} ensures that $v \geq -\mu_{\kappa}$.  It remains to prove that $v \in \V_{\iota(v)}$. Let $M$ be the least positive integer such that $v \in \V_{M}$. Lemma \ref{lem:exist v0 vM} ensures that there exist $v_{0} < \cdots < v_{M}$ in $\V$ such that 	
\begin{itemize}
\item $v_{0} \in \V_{0}=-\S(L)$;
 \item $v_{i+1} \in \pi (\Psi(v_i))$ for any $i \in \{0,\ldots,M-1\}$;
 \item $v_{M}=v$.
\end{itemize}
	It follows from Lemma \ref{lem:longueur-chaine1} that $M \leq \iota(v)$. Since $(\V_{i})_{i \geq 0}$ is \croissant{} by Theorem \ref{theo: un support bis}, we have $v \in \V_{M} \subset \V_{\iota(v)}$. This proves that \eqref{v in V} implies \eqref{v in Vu} in Proposition \ref{prop:majoration-iterationV}. The converse implication is obvious. 
\end{proof}

\section{An algorithm for computing a positive lower bound on $\epsilon(v)$}\label{sec:algo plb epsilon}

We retain the notations from the previous section. In addition, we set $\mu_{0}=-\infty$.

\subsection{Structure of the algorithm}\label{sec:algo plb epsilon structure}

The aim of this section is to present an algorithm for computing a positive lower bound on $\epsilon(v)$ for any given $v \in \Z_{d,\ellmahl}$. This algorithm is recursive and its structure is as follows.
\begin{itemize}
\item The base case corresponds to the case when $v$ belongs to $]-\infty, -\mu_{\kappa}] \cap \Z_{d,\ellmahl}$. This case presents no difficulty because $\epsilon(v)$ can be computed explicitly as we will see in Proposition \ref{prop:form eps base case}. 
\item The recursive step is organized as follows. We consider an element $v$ of $]-\mu_{\kappa}, +\infty[ \cap \Z_{d,\ellmahl}$ and we distinguish two cases: 
\begin{itemize}
 \item[-]  if $v \not \in \{-\mu_{\kappa-1},\ldots,-\mu_{1}\}$, then there exists $k \in \{1,\ldots,\kappa\}$ such that $v  \in  ] -\mu_{k},-\mu_{k-1}[\cap \Z_{d,\ellmahl} $ and we will see in Proposition \ref{prop: step 2}
 how to compute a positive lower bound on $\epsilon(v)$ from positive lower bounds on $\epsilon(w)$ for finitely many explicit $w \in ] -\infty,-\mu_{k}]\cap \Z_{d,\ellmahl} $; 
 \item[-] if $v=-\mu_{k-1}$ for some $k \in \{2,\ldots,\kappa\}$, then we will see in Proposition \ref{prop: step 3} 
 how to compute a positive lower bound on $\epsilon(v)=\epsilon(-\mu_{k-1})$ from positive lower bounds on $\epsilon(w)$ for finitely many explicit $w \in ] -\infty,-\mu_{k-1}[\cap \Z_{d,\ellmahl}$. 
\end{itemize}
\end{itemize}
The algorithm is presented in pseudo-code form in Algorithm \ref{algo: lower bound epsilon} in section~\ref{sec: lbeps pc}. The theoretical results mentioned above, namely Proposition \ref{prop:form eps base case}, Proposition \ref{prop: step 2} and Proposition \ref{prop: step 3}, on which the algorithm is based, are the subject of the next three sections.

\subsection{Theoretical result for the base case} \label{sec: plb eps base case}

The following result allows us to compute a lower bound on $\epsilon(v)$ for any given $v \in ]-\infty, -\mu_{\kappa}] \cap \Z_{d,\ellmahl}$. 

\begin{prop}\label{prop:form eps base case}
We have: 
\begin{itemize}
\item for any $v \in ]-\infty, -\mu_{\kappa}[  \cap \Z_{d,\ellmahl}$, 
$
\epsilon(v) = -\mu_{\kappa} - v>0
$;
\item $\epsilon(-\mu_{\kappa})= \min (\V_{1} \setminus \{-\mu_{\kappa}\}) + \mu_{\kappa}$.
\end{itemize}
Therefore, 
a lower bound on $\epsilon(v)$ is given by 
\begin{itemize}
 \item  
$
-\mu_{\kappa} - v
$ if $v \in ]-\infty, -\mu_{\kappa}[  \cap \Z_{d,\ellmahl}$;
 \item $\min (\V_{1} \setminus \{-\mu_{\kappa}\}) + \mu_{\kappa}$ if $\V_{1} \setminus \{-\mu_{\kappa}\} \neq \emptyset$ and $1$ otherwise\footnote{Here, $1$ is an arbitrary choice, any positive value would work.}, if $v=-\mu_{\kappa}$. 
\end{itemize}
\end{prop}

\begin{ex}\label{ex:base cas RS}
In Section \ref{illustration main algo on RS}, we will illustrate our algorithm for computing a positive lower bound on $\epsilon(v)$, namely Algorithm \ref{algo: lower bound epsilon} presented in section~\ref{sec: lbeps pc} below, on  the operator defined by \eqref{L pour example}. We will have to compute positive lower bounds on $\epsilon(-\frac{3}{4})$ and $\epsilon(-\frac{1}{2})$. Let us explain how this can be done using Proposition \ref{prop:form eps base case}. We will see in section \ref{sec:example Newton et slopes} that $\kappa=2$, $\mu_{1}=0$ and $\mu_{2}=\frac{1}{2}$. Since $-\frac{3}{4}<-\mu_{2}$, Proposition \ref{prop:form eps base case} ensures that 
$$\epsilon(-\frac{3}{4})=-\mu_{2}-(-\frac{3}{4})=-\frac{1}{2}+\frac{3}{4}=\frac{1}{4}.
$$ 
Moreover, we will see in section \ref{the Vi for RS example} that $\V_1=\left\{-\frac{1}{2},-\frac{1}{4},0,1\right\}$, so $\V_{1} \setminus \{-\mu_{2}\}=\left\{-\frac{1}{4},0,1\right\} \neq \emptyset$ and Proposition \ref{prop:form eps base case} ensures that 
$$\epsilon(-\frac{1}{2})=\epsilon(-\mu_{2})=\min (\V_{1} \setminus \{-\mu_{2}\}) + \mu_{2}=-\frac{1}{4}+\frac{1}{2}=\frac{1}{4}.$$
\end{ex}

\begin{proof}[Proof of Proposition \ref{prop:form eps base case}]
Lemma \ref{min V=-muk} ensures that $\min \V=-\mu_{\kappa}$. Therefore, 
 \begin{itemize}
\item  for all $v \in ]-\infty, -\mu_{\kappa}[  \cap \Z_{d,\ellmahl}$, we have $\epsilon(v)= \min \V_{>v} - v = -\mu_{\kappa} - v>0$; 
\item $\epsilon(-\mu_{\kappa})=\min \V_{>-\mu_{\kappa}} + \mu_{\kappa}=\min (\V \setminus \{-\mu_{\kappa}\}) + \mu_{\kappa}$ and the desired equality $\epsilon(-\mu_{\kappa})= \min (\V_{1} \setminus \{-\mu_{\kappa}\}) + \mu_{\kappa}$ follows from Lemma \ref{lem:minV minV1} below.
\end{itemize}
\end{proof}

\begin{lem}\label{lem:minV minV1}
 We have $\min \V \setminus \{-\mu_{\kappa}\}=\min \V_{1} \setminus \{-\mu_{\kappa}\}$. 
\end{lem}

\begin{proof}
We set $w= \min \V \setminus \{-\mu_{\kappa}\}$. 
If $\V \setminus \{-\mu_{\kappa}\}=\emptyset$, then $\V_{1} \setminus \{-\mu_{\kappa}\}=\emptyset$ as well and the equality $w= \min \V_{1} \setminus \{-\mu_{\kappa}\}$ holds in this case. From now on, we assume that $\V \setminus \{-\mu_{\kappa}\}\neq\emptyset$.  In order to prove the lemma, it is sufficient to prove that $w \in \V_1$. 
Let $M$ be the least element of $\Z_{\geq 0}$ such that $w \in \V_{M}$. If $M=0$, then $w \in \V_0 \subset \V_1$ and the lemma is proved in this case. Suppose now that $M\geq 1$. We want to prove that $M=1$. According to Lemma \ref{lem:exist v0 vM}, there exist $v_0 , v_1 , \ldots , v_M \in \V$ such that $v_{0} \in -\S(L)$, $v_0  < v_1 < \cdots < v_M$ and $v_{M}=w$. Since $\min -\S(L)=-\mu_{\kappa}$ and $v_{0} \in -\S(L)$, the fact that $v_{1}>v_{0}$ implies that $v_{1}>-\mu_{\kappa}$ and, hence, $v_{1} \in \V \setminus \{-\mu_{\kappa}\}$. It follows that $v_{1} \geq \min \V \setminus \{-\mu_{\kappa}\}=w=v_{M}$ and, hence, $M=1$.  
 \end{proof}

\subsection{Theoretical results for the recursive step: case $v \not \in \{-\mu_{\kappa-1},\ldots,-\mu_{1}\}$}\label{sec: plb eps rec step 1}

In this section, we consider $k \in \{1,\ldots,\kappa\}$ and we assume that we are able to compute a positive lower bound $\vareps(w)$ on $\epsilon(w)$ for any $w \in ]-\infty, -\mu_{k}]  \cap \Z_{d,\ellmahl}$.
Our aim is to explain how one can compute,  for any $v  \in ] -\mu_{k},-\mu_{k-1}[\cap \Z_{d,\ellmahl}$, a positive lower bound on $\epsilon(v)$ 
by using finitely many of the values $\vareps(w)$, with $w \in ]-\infty,-\mu_{k}]\cap \Z_{d,\ellmahl}$. 

Our approach, detailed in Proposition \ref{prop: step 2} below, relies on a labeled rooted tree $\tree{k,\vareps(-\mu_{k}),v}$ that we shall now introduce.   
Its definition involves the finite (possibly empty) sets defined, for any $w \in \Z_{d,\ellmahl}$, by
$$
 \Delta(w) =\left\{\frac{\psi(w)-\betaa}{\ellmahl^{\alphaa}} \ \vert \ (\ellmahl^{\alphaa},\betaa) \in \P(L) \right\} \setminus\{w\}, 
$$
In geometric terms, $\Delta(w)$ is the set of $w' \in \Q \setminus\{w\}$ for which there exists a point in $\P(L)$ whose projection along a line of slope $-w'$ onto the $y$-axis is the point with coordinates $(0,\minP(w))$. It will also be useful to keep in mind that, in more computational terms, $\Delta(w)$ is the set of $w' \in \Q \setminus \{w\}$ such that $\minP(w) \in  \Psi(w')$ and that, according to the discussion at the beginning of section \ref{sec:the maps Psi minP pi}, $\Psi(w')$ is a natural receptacle for $\supp L(z^{w'})$.
 
\begin{defi}\label{defi:tree}
Let $k \in \{1,\ldots,\kappa\}$ and  $v  \in ] -\mu_{k},-\mu_{k-1}[\cap \Z_{d,\ellmahl}$. The labeled rooted tree $\tree{k,\vareps(-\mu_{k}),v}$ alluded to above is uniquely defined by the following properties:
\begin{itemize}
 \item the labels of $\tree{k,\vareps(-\mu_{k}),v}$ belong to $\Q$; 
 \item the label of the root of $\tree{k,\vareps(-\mu_{k}),v}$ is $v$; 
\item if a vertex has label $w < -\mu_{k}+\vareps(-\mu_{k})$, then it has no child, {\it i.e.}, it is a leaf;
 \item if a vertex has label $w \geq -\mu_{k}+\vareps(-\mu_{k})$, then it has $\card \Delta(w)$ children labeled by the elements of $\Delta(w)$. 
\end{itemize}
\end{defi}

Note that: 
\begin{itemize}
 \item In the last case of Definition \ref{defi:tree}, since $w > -\mu_{\kappa}$, it follows from Lemma \ref{lem:if delta empty then kappa=1 w=-mu1} below that $\Delta(w) \neq \emptyset$ and, hence, the vertex is not a leaf. 
\item In the tree $\tree{k,\vareps(-\mu_{k}),v}$, children have smaller labels than their parent; this is a direct consequence (of the first assertion) of Lemma \ref{lem:Delta} below.  
\end{itemize}

The following two lemmas give useful properties of the set $\Delta(w)$.

\begin{lem}\label{lem:if delta empty then kappa=1 w=-mu1}
If the set $\Delta(w)$ is empty for some $w \in \Z_{d,\ellmahl}$, then $\kappa=1$ and $w=-\mu_{\kappa}=-\mu_{1}$. 
\end{lem}

\begin{proof}
The following properties are obviously equivalent:
\begin{itemize}
 \item $\Delta(w)$ is empty; 
 \item for all $(\ellmahl^{\alphaa},\betaa) \in \P(L)$, $\frac{\psi(w)-\betaa}{\ellmahl^{\alphaa}}= w$; 
 \item for all $(\ellmahl^{\alphaa},\betaa) \in \P(L)$, $\psi(w)= \ellmahl^{\alphaa} w+\betaa$; 
\item  $\Psi(w)=\{\psi(w)\}$;
\item  $\Psi(w)$ is  a singleton. 
\end{itemize}
Since $\Psi(w)$ is the set of ordinates of the projection of the elements of $\P(L)$ along a line of slope $-w$ onto the $y$-axis, the fact that $\Psi(w)$ is a singleton is equivalent to the fact that the elements of $\P(L)$ belong to a single line of slope $-w$. In that case, $\mathcal N(L)$ has a single slope and this slope is equal to $-w$, {\it i.e.}, $\kappa=1$ and $w=-\mu_{\kappa}=-\mu_{1}$. 
\end{proof}

\begin{lem}\label{lem:Delta}
Let $w,w'\in \Z_{d,\ellmahl}$. If $w' \in \Delta(w)$, then $w'<w$ and $w \in \pi(\Psi(w'))$. Reciprocally\footnote{Strictly speaking, the reciprocal assertion should start with the hypothesis ``$w'<w$ and $w \in \pi(\Psi(w'))$'' instead of the seemingly weaker hypothesis  ``$w' \neq w$ and $w \in \pi(\Psi(w'))$''. Actually, these hypotheses are equivalent because $\min \pi(\Psi(w'))=w'$ according to Lemma~\ref{min Psi(s)=s}.}, if $w \in \pi(\Psi(w'))$ and $w \neq w'$, then $w' \in \Delta(w)$ and $w'<w$.
\end{lem}

\begin{proof}
Consider $w' \in \Delta(w)$. There exists $(\ellmahl^{\alphaa},\betaa) \in \P(L)$ such that 
$$
w'=\frac{\psi(w)  - \betaa}{\ellmahl^{\alphaa}}.
$$
Thus, we have 
$
\psi(w)=\ellmahl^{\alphaa}w'+\betaa
$
and, hence, since $\pi$ and $\minP$ are inverse of each other by Lemma \ref{lem proprietes pi minP}, we have 
$
w=\pi(\ellmahl^{\alphaa}w'+\betaa) \in \pi( \Psi(w')).
$
Then, it follows from Lemma \ref{min Psi(s)=s} that $w' \leq w$. Since $w'\neq w$, we have $w'<w$.
	
	Suppose now that $w \in \pi(\Psi(w'))$ and $w \neq w'$. Then, $w=\pi(\ellmahl^{\alphaa}w'+\betaa)$ for some $(\ellmahl^{\alphaa},\betaa) \in \P(L)$. Since $\pi$ and $\minP$ are inverse of each other by Lemma \ref{lem proprietes pi minP}, we have $
	\minP(w)=\ellmahl^{\alphaa}w'+\betaa
	$
and, hence,
	$
w'=\frac{\psi(w)  - \betaa}{\ellmahl^{\alphaa}}
	$. Since $w'\neq w$, we get $w' \in \Delta(w)$. The fact that $w'<w$ follows from the first part of the lemma.
\end{proof}

As announced above, the following result shows how to compute, for any $v  \in  ] -\mu_{k},-\mu_{k-1}[\cap \Z_{d,\ellmahl}$, a positive lower bound on $\epsilon(v)$ by using finitely many $\vareps(w)$ with $w \in ] -\infty,-\mu_{k}]\cap \Z_{d,\ellmahl}$. Our approach relies on the labeled rooted tree $\epsilon(\tree{k,\vareps(-\mu_{k}),v})$ obtained from $\tree{k,\vareps(-\mu_{k}),v}$ by applying $\epsilon$ to its labels.  It will be used in the following way: we will explain how to compute a positive lower bound on the label of any leaf of $\epsilon(\tree{k,\vareps(-\mu_{k}),v})$ and how to compute a positive lower bound on the label of any vertex of $\epsilon(\tree{k,\vareps(-\mu_{k}),v})$ from positive lower bounds on its children; this will enable us to compute a positive lower bound on the label of any vertex of $\epsilon(\tree{k,\vareps(-\mu_{k}),v})$ and, in particular, on the label of its root, which is nothing but $\epsilon(v)$. 

\begin{prop} \label{prop: step 2}
Let $k \in \{1,\ldots,\kappa\}$ and  $v  \in  ] -\mu_{k},-\mu_{k-1}[\cap \Z_{d,\ellmahl}$.
\vskip 5pt
\noindent (i) The tree $\tree{k,\vareps(-\mu_{k}),v}$ is finite and its height is less than or equal to 
\begin{equation}\label{bound height}
(n+1) \frac{v + \mu_{k}}{\min\left\{\vareps(-\mu_{k}),(d\ellmahl^n)^{-1}\right\}} + h(v)+1. 
\end{equation}
\vskip 5pt
\noindent (ii) Consider a leaf of $\tree{k,\vareps(-\mu_{k}),v}$ with label $w$. Then a positive lower bound on $\epsilon(w)$ is given by $\vareps(w)$ if $w < -\mu_k$  and 
	by $-\mu_{k} + \vareps(-\mu_{k})-w$
	 if $-\mu_k \leq w < -\mu_{k} + \vareps(-\mu_{k})$. 
	 
\vskip 5pt
\noindent (iii) Consider a vertex of $\tree{k,\vareps(-\mu_{k}),v}$ with label $w$ which is not a leaf. If,  for each $w'\in \Delta(w)$, we have a positive lower bound $m_{w'}$ on $\epsilon(w')$, then a positive lower bound on $\epsilon(w)$ is given by the minimum of
\begin{equation}\label{lower bound 3}
 \{m_{w'}\ellmahl^{d_{w,w'}-\alpha_{k-1}} \ \vert \ w' \in \Delta(w) \}\cup \{
 -\mu_{k-1} - w, \min(\pi(\Psi(w))\setminus \{w\}) - w\}
\end{equation}
where $d_{w,w'}$ is defined by \eqref{eq:def_dkw}. 

Consequently:
\begin{itemize}
\item (i) above ensures that $\epsilon(\tree{k,\vareps(-\mu_{k}),v})$ is a finite tree and gives an explicit upper bound on its height;  
 \item (ii) above allows us to compute a positive lower bound on the label of any leaf of $\epsilon(\tree{k,\vareps(-\mu_{k}),v})$; 
 \item (iii) above allows us to compute a positive lower bound on the label of a vertex of $\epsilon(\tree{k,\vareps(-\mu_{k}),v})$ which is not a leaf from positive lower bounds on its children. 
\end{itemize}
This allows us to compute a positive lower bound on the label of any vertex of $\epsilon(\tree{k,\vareps(-\mu_{k}),v})$ and, in particular, a positive lower bound on the label of its root, namely $\epsilon(v)$. 
\end{prop}

\begin{ex}\label{ex:rec step case 1}
In section \ref{illustration main algo on RS}, we will have to compute a positive lower bound on $\epsilon(-\frac{1}{4})$ for the operator defined by \eqref{L pour example}. Let us explain how this can be done using Proposition \ref{prop: step 2}.   We will see in section \ref{sec:example Newton et slopes} that $\kappa=2$, $\mu_{1}=0$ and $\mu_{2}=\frac{1}{2}$. Moreover, we have seen in Example \ref{ex:base cas RS} that $\epsilon(-\mu_{2})=\epsilon(-\frac{1}{2})=\frac{1}{4}$, thus 
$\vareps(-\mu_{2})=\vareps(-\frac{1}{2}):=\frac{1}{4}$ is a positive lower bound on 
$\epsilon(-\mu_{2})=\epsilon(-\frac{1}{2})$. 

We have $-\frac{1}{4}  \in ] -\mu_{2},-\mu_{1}[=] -\frac{1}{2},0[$, so, following the method presented in Proposition \ref{prop: step 2}, in order to compute a positive lower bound on $\epsilon(-\frac{1}{4})$, we first compute the tree $\tree{2,\vareps(-\mu_{2}),-\frac{1}{4}}=\tree{2,\frac{1}{4},-\frac{1}{4}}$. The result is shown in Figure \ref{fig:subfigAtree}. 

			\begin{figure}
        \centering
\begin{tikzpicture}[every node/.style={draw}]
	\node {$-\frac{1}{4}$}
	child {node {$-\frac{1}{2}$}}
	child {node {$-\frac{3}{4}$}}
	child {node {$-\frac{3}{8}$}};
	\end{tikzpicture}
        \caption{}   
        \label{fig:subfigAtree}
\end{figure}

We then compute lower bounds on the labels $\epsilon(-\frac{1}{2})$, $\epsilon(-\frac{3}{4})$ and $\epsilon(-\frac{3}{8})$ of the leaves of $\epsilon(\tree{2,\frac{1}{4},-\frac{1}{4}})$. 
Since $-\frac{3}{8}\geq-\mu_{2}=-\frac{1}{2}$, it follows from (ii) of Proposition \ref{prop: step 2} that a positive lower bound on $\epsilon(-\frac{3}{8})$ is given by 
$$
m_{-\frac{3}{8}}=-\mu_2+\theta(-\mu_{2}) - (-\frac{3}{8})=-\frac{1}{2}+\frac{1}{4}+\frac{3}{8}=\frac{1}{8}.
$$ 
Similarly, since $-\frac{1}{2}\geq-\mu_{2}=-\frac{1}{2}$, it follows from (ii) of Proposition \ref{prop: step 2} that a positive lower bound on $\epsilon(-\frac{1}{2})$ is given by 
$$
m_{-\frac{1}{2}}=-\mu_2+\theta(-\mu_{2}) - (-\frac{1}{2})=-\frac{1}{2}+\frac{1}{4}+\frac{1}{2}=\frac{1}{4}.
$$ 
Last, $-\frac{3}{4} < -\mu_{2}=-\frac{1}{2}$, 
and it has been shown in Example \ref{ex:base cas RS} that  
$m_{-\frac{3}{4}}=\frac{1}{4}$ is a positive lower bound on 
$\epsilon(-\frac{3}{4})$ respectively. 

Now that we have lower bounds on the labels of the leaves of $\epsilon(\tree{2,\frac{1}{4},-\frac{1}{4}})$, (ii) of Proposition \ref{prop: step 2} ensures that a positive lower bound on the root $\epsilon(-\frac{1}{4})$ of $\epsilon(\tree{2,\frac{1}{4},-\frac{1}{4}})$ is given by the minimum of \eqref{lower bound 3} with $w=-\frac{1}{4}$ and $k=2$. Computing this minimum requires the calculation of $\alpha_{1}$, of $d_{-\frac{1}{4},w'}$ for $w' \in  \{-\frac{1}{2},-\frac{3}{4},-\frac{3}{8}\}$ and of $\min(\pi(\Psi(-\frac{1}{4}))\setminus \{-\frac{1}{4}\})$. This presents no difficulty. Indeed, we will see in section \ref{sec:example Newton et slopes} that $\alpha_{1}=1$. 
Moreover, using the explicit formulas for $\P(L)$ and $\pi$ given in Section \ref{sec:example Newton et slopes} and Section \ref{sec:calc pi Psi RS}, we get 
\begin{multline*}
 d_{-\frac{1}{4},-\frac{1}{2}}=  \min \biggl\{\alphaa \in \{0,1,2\} \ \vert \ \exists (\ell^{\alphaa},\betaa) \in \{(1,0),(2,0),(2,1),(4,1)\},\\
  -\frac{1}{2}=\frac{\psi(-\frac{1}{4})-\betaa}{2^{\alphaa}}
=\frac{-\frac{1}{2}-\betaa}{2^{\alphaa}}\biggl\}=0. 
\end{multline*}
Similar calculations show that 
\begin{equation*}\label{eq: expl formulas for d w w'}
d_{-\frac{1}{4},-\frac{3}{4}}=1, \ \  d_{-\frac{1}{4},-\frac{3}{8}}=2. 
\end{equation*}
Last, the explicit formulas for $\pi$ and $\Psi$ given in section \ref{sec:calc pi Psi RS} show that 
$$\min(\pi(\Psi(-\frac{1}{4}))\setminus \{-\frac{1}{4}\})=-\frac{1}{8}.$$
Finally, we obtain that a positive lower bound of $\epsilon(-\frac{1}{4})$ is given by the minimum of the following numbers: 
\begin{itemize}[label=$\bullet$]
					\item  $m_{w'}\ellmahl^{d_{-\frac{1}{4},w'}-\alpha_{1}}=\frac{1}{4}\times2^{0-1}=\frac{1}{8}$ for $w'= -\frac{1}{2}$;
					\item $m_{w'}\ellmahl^{d_{-\frac{1}{4},w'}-\alpha_{1}}=\frac{1}{4}\times2^{1-1}=\frac{1}{4}$ for $w'= -\frac{3}{4}$;
					\item $m_{w'}\ellmahl^{d_{-\frac{1}{4},w'}-\alpha_{1}}=\frac{1}{8}\times2^{2-1}=\frac{1}{4}$ for $w'= -\frac{3}{8}$;
					\item $0 - \left(- \frac{1}{4}\right)=\frac{1}{4}$; 
					\item $\min(\pi(\Psi(-\frac{1}{4}))\setminus \{-\frac{1}{4}\}) + \frac{1}{4}=-\frac{1}{8}+\frac{1}{4} = \frac{1}{8}$; 
				\end{itemize}
this minimum is equal to $\frac{1}{8}$. Thus, a positive lower bound of $\epsilon(-\frac{1}{4})$ is given by $\vareps(-\frac{1}{4})=\frac{1}{8}$. 
\end{ex}

The proof of Proposition \ref{prop: step 2} is given below, after the following lemma.

\begin{lem}\label{lem:espilon_Delta}
	Let $w \in \mathbb Q_{>-\mu_{\kappa}}$ and let  $k \in \{1,\ldots,\kappa\}$ be such that $-\mu_k \leq w < -\mu_{k-1}$, with the convention $\mu_0=-\infty$. Then, a positive lower bound on $\epsilon(w)$ is the minimum of the following set 
	\begin{equation}\label{lower bound 2}
	\{\epsilon (w')\ellmahl^{d_{w,w'}-\alpha_{k-1}} \ \vert \ w' \in \Delta(w) \}\cup \{
	-\mu_{k-1} - w, \min(\pi(\Psi(w))\setminus \{w\}) - w\}
	\end{equation}
	where 
	 \begin{equation}\label{eq:def_dkw}
	d_{w,w'}= \min \left\{\alphaa \in \{0,\ldots,n\} \ \vert \ \exists (\ell^{\alphaa},\betaa) \in \P(L), w'=\frac{\psi(w)-\betaa}{\ellmahl^{\alphaa}}\right\}.
	\end{equation}
\end{lem}
\begin{proof}
	Let 
	$
	\wplus=\min \V_{>w}
	$  
	so that 
	$
	\epsilon(w)=\wplus-w.
	$ 
	Since $- \mu_{k-1} \in \V_{>w} \cup \{+\infty\}$, we have $ -\mu_{k-1} \geq \wplus$. Thus:
	$$
	-\mu_{k} \leq w<\wplus \leq -\mu_{k-1}.
	$$
	We shall now distinguish several cases. 
	\vskip 5pt
	\noindent {\textbf{Case 1: $\wplus =-\mu_{k-1}$.}} In this case, $\epsilon(w)=\wplus-w = -\mu_{k-1}-w$ and this quantity is bounded from below by 	the minimum of \eqref{lower bound 2}. \vskip 5pt
	\vskip 5pt
	\noindent {\textbf{Case 2: $\wplus < -\mu_{k-1}$.}} In this case, we have $-\mu_{k} < \wplus <-\mu_{k-1}$ and, hence, $\wplus \notin \V_0=-\S(L)$. It follows from Lemma \ref{lem:exist v0 vM} that there exists $\wmoins \in \V$ such that 
	\begin{equation}\label{prop wmoins}
	\wmoins< \wplus \text{ and } \wplus \in \pi( \Psi(\wmoins)). 
	\end{equation}
	Since $\wplus = \min \V_{>w}$, the facts that $\wmoins \in \V$ and that $\wmoins < \wplus$ ensures that $\wmoins \leq w$.  We now distinguish two subcases. 
	\vskip 5pt
	\noindent {\textbf{Subcase 2.1: $\wplus < -\mu_{k-1}$ and $\wmoins=w$.}} In this case, since $\wplus \in \pi( \Psi(\wmoins))$ by \eqref{prop wmoins}, we have 
	$$
	\epsilon(w)=\wplus-w \geq \min(\pi(\Psi(\wmoins))\setminus \{\wmoins\}) - w= \min(\pi(\Psi(w))\setminus \{w\}) - w 
	$$
	and this quantity is bounded from below by the minimum of  
	\eqref{lower bound 2}.
	\vskip 5pt	
	\noindent {\textbf{Subcase 2.2: $\wplus < -\mu_{k-1}$ and $\wmoins<w$.}} Since $-\mu_{k} \leq w < \wplus < -\mu_{k-1}$, Lemma \ref{lem psi} ensures that  
		\begin{equation}\label{eq:w+ and w}
		\wplus=\frac{\minP(\wplus)-\beta_{k-1}}{\ellmahl^{\alpha_{k-1}}} \text{ and } 	w=\frac{\minP(w)-\beta_{k-1}}{\ellmahl^{\alpha_{k-1}}}\,.
		\end{equation}
	Furthermore, since $\wplus \in \pi ( \Psi(\wmoins))$ by \eqref{prop wmoins} and since $\pi$ and $\minP$ are inverse of each other by Lemma \ref{lem proprietes pi minP}, we have $\minP(\wplus) \in \Psi(\wmoins)$ and, hence, there exists $(\ellmahl^{\alphaa},\betaa)\in \P(L)$ such that $\minP(\wplus)=\ellmahl^{\alphaa} \wmoins+ \betaa$. Therefore, we have   
		\begin{equation}\label{eq w' pq}
		\wmoins=\frac{\psi(\wplus) - \betaa}{\ellmahl^{\alphaa}}
		=\frac{\ellmahl^{\alpha_{k-1}} \wplus + \beta_{k-1} - \betaa}{\ellmahl^{\alphaa}}. 
		\end{equation}
			Let
			\begin{equation}\label{eq ws pq}
			w':=\frac{\psi(w) - \betaa}{\ellmahl^{\alphaa}}.
			\end{equation}
			It follows from Lemma \ref{lem proprietes pi minP} that $\minP$ is increasing so that we have
			\begin{equation}\label{eq:ineq_w'_w-}
			w'=\frac{\psi(w) - \betaa}{\ellmahl^{\alphaa}} < \frac{\psi(\wplus) - \betaa}{\ellmahl^{\alphaa}}=\wmoins<w
			\end{equation}
			In particular, this implies that $w' \in \Delta(w)$. 
	Using \eqref{eq:w+ and w}, we get 	
		$$
		\wplus-w = \frac{\minP(\wplus)-\beta_{k-1}}{\ellmahl^{\alpha_{k-1}}}-\frac{\minP(w)-\beta_{k-1}}{\ellmahl^{\alpha_{k-1}}}
		=\frac{\minP(\wplus)-\minP(w)}{\ellmahl^{\alpha_{k-1}}}.
		$$
Moreover, we infer from \eqref{eq w' pq} that 
$$
\psi(\wplus) = \ellmahl^{\alphaa}\wmoins+\betaa 
$$
and from \eqref{eq ws pq} that 
$$
\psi(w) = \ellmahl^{\alphaa}w'+\betaa. 
$$
So, we obtain 
	$$
		\wplus-w = \frac{(\ellmahl^{\alphaa}\wmoins+\betaa )-(\ellmahl^{\alphaa}w'+\betaa)}{\ellmahl^{\alpha_{k-1}}}
		=\frac{\ellmahl^{\alphaa}(\wmoins-w')}{\ellmahl^{\alpha_{k-1}}}.
		$$	
		But, since $w' < \wmoins$ by \eqref{eq:ineq_w'_w-} and since $\wmoins \in \V$, we have  $\wmoins \geq w' + \epsilon(w')$.
		Therefore, 
		$$
		\wplus-w =\frac{\ellmahl^{\alphaa}(\wmoins -w')}{\ellmahl^{\alpha_{k-1}}} \geq  \ellmahl^{\alphaa-\alpha_{k-1}} \epsilon(w')
		$$
		and the latter quantity is bounded from below by the minimum of \eqref{lower bound 2}.
\end{proof}

\begin{proof}[Proof of Proposition \ref{prop: step 2}]  
We start with the proof of (i). Since any vertex of $\tree{k,\vareps(-\mu_{k}),v}$ has finitely many children, in order to prove (i), it is sufficient to prove that the height of $\tree{k,\vareps(-\mu_{k}),v}$ is less than or equal to \eqref{bound height}. 
Consider an arbitrary vertex of $\tree{k,\vareps(-\mu_{k}),v}$ with label $w$ and depth $d$. Let  $w_{0}=w,w_{1},\ldots,w_{d}=v$ be the labels of the vertices encountered along the path from this vertex to the root of $\tree{k,\vareps(-\mu_{k}),v}$, so that, for any $i \in \{0,\ldots,d-1\}$, we have 
$$
w_{i} \in \Delta(w_{i+1}).
$$

We claim  that   
\begin{equation}\label{ineq wi}
-\mu_{k} + \vareps(-\mu_{k})  \leq w_{1} < \cdots < w_{d-1} <  w_d=v < -\mu_{k-1}
\end{equation}
and that, for all $i \in \{0,\ldots,d-1\}$, 
\begin{equation}\label{wi-1 dans pi Psi}
 w_{i+1} \in \pi( \Psi(w_i)).
\end{equation}
Indeed, for any $i \in \{0,\ldots,d-1\}$, we have $w_{i} \in \Delta(w_{i+1})$, so Lemma \ref{lem:Delta} ensures that  $w_{i+1}>w_{i}$ and that $w_{i+1} \in \pi( \Psi(w_i))$. So, we have justified \eqref{wi-1 dans pi Psi} and, in order to justify \eqref{ineq wi}, it only remains to prove that $w_{1}\geq -\mu_{k} + \vareps(-\mu_{k})$. The latter inequality follows from the fact that $w_{1}$ is the label of a vertex of $\tree{k,\vareps(-\mu_{k}),v}$ which is not a leaf because that vertex has a vertex with label $w_{0}$ as a child.

Now, applying Lemma \ref{lem:longueur-chaine2}, we get   
	$$
	d-1<(n+1) \frac{w_d + \mu_{k}}{\min\left\{\vareps(-\mu_{k}),(d\ellmahl^n)^{-1}\right\}} + h(w_d)
	$$	
	and, hence, $d$ is less than or equal to \eqref{bound height}. 
	This concludes the proof of (i).

Let us now prove (ii). Since $w$ is the label of a leaf of $\tree{k,\vareps(-\mu_{k}),v}$, we have  $w< -\mu_k + \vareps(-\mu_k)$.  If $w < -\mu_k$, there is nothing to prove. Suppose that $-\mu_k \leq w < -\mu_k + \vareps(-\mu_k)$. Then $-\mu_{k} + \vareps(-\mu_{k})-w>0$.  Moreover, for any $w' \in \V_{>w}$, we have $w' > w \geq -\mu_{k}$ and, hence, $w' \geq -\mu_{k} + \vareps(-\mu_{k})=w+(-\mu_{k} + \vareps(-\mu_{k})-w)$. We have shown that $-\mu_{k} + \vareps(-\mu_{k})-w$ is a positive lower bound on $\epsilon(w)$. 

Last, (iii) follows from Lemma \ref{lem:espilon_Delta}. Indeed, Lemma \ref{lem:espilon_Delta} ensures that a positive lower bound on $\epsilon(v)$ is given by the minimum of
\begin{equation}\label{avec epsilon}
 \{\epsilon(w')\ellmahl^{d_{w,w'}-\alpha_{k-1}} \ \vert \ w' \in \Delta(w) \}\cup \{
 -\mu_{k-1} - w, \min(\pi(\Psi(w))\setminus \{w\}) - w\}. 
\end{equation}
Since $m_{w'} \leq \epsilon(w')$ for any $w' \in \Delta(w)$, the latter minium is greater than or equal to the minimum of
\begin{equation}\label{avec vareps}
 \{m_{w'}\ellmahl^{d_{w,w'}-\alpha_{k-1}} \ \vert \ w' \in \Delta(w) \}\cup \{
 -\mu_{k-1} - w, \min(\pi(\Psi(w))\setminus \{w\}) - w\}
\end{equation}
(we emphasize that the only difference between \eqref{avec epsilon} and \eqref{avec vareps} is the first quantity, $\epsilon(w')$ versus $m_{w'}$), whence the desired result. 
\end{proof}

\subsection{Theoretical results for the recursive step: case $v \in \{-\mu_{\kappa-1},\ldots,-\mu_{1}\}$}\label{sec: plb eps rec step 2}

In this section, we consider $k \in \{1,\ldots,\kappa\}$ and we assume that we are able to compute a positive lower bound $\vareps(w)$ on $\epsilon(w)$ for any $w \in ]-\infty, -\mu_{k-1}[  \cap \Z_{d,\ellmahl}$.
The following result explains how one can compute a positive lower bound on $\epsilon(-\mu_{k-1})$ by using finitely many $\vareps(w)$ with $w \in ] -\infty,-\mu_{k-1}[\cap\Z_{d,\ellmahl}  $.  

\begin{prop}\label{prop: step 3}
 Let $k\in \{2,\ldots,\kappa\}$. A positive lower bound on $\epsilon(-\mu_{k-1})$ is given by the minimum of
\begin{eqnarray}\label{min cas 3}
\nonumber & \{\vareps(w')\ellmahl^{d_{-\mu_{k-1},w'} - \alpha_{k-2}} \ \vert \ w' \in \Delta(-\mu_{k-1}) \}&
 \\ &\text{ and } & 
 \\
\nonumber & \{
 \mu_{k-1} - \mu_{k-2}, \min(\pi(\Psi(-\mu_{k-1}))\setminus \{-\mu_{k-1}\}) +\mu_{k-1}\} &
\end{eqnarray}
where $d_{-\mu_{k-1},w'}$ is defined by \eqref{eq:def_dkw}.
\end{prop}

\begin{rem}
 It follows from (the first assertion of) Lemma \ref{lem:Delta} that any  $w'\in \Delta(-\mu_{k-1})$ satisfies $w'<-\mu_{k-1}$, hence, it is legitimate to consider $\vareps(w')$ in \eqref{min cas 3}. 
\end{rem}

\begin{ex}\label{ex:step 3 pour RS}
In section \ref{illustration main algo on RS}, we will have to compute a positive lower bound on $\epsilon(0)$ for the operator defined by \eqref{L pour example}. Let us explain how this can be done using Proposition \ref{prop: step 3}.   We will see in section \ref{sec:example Newton et slopes}  that $\kappa=2$, $\mu_{1}=0$ and $\mu_{2}=\frac{1}{2}$. 
Since $0=-\mu_{1}$, a lower bound on $\epsilon(0)$ is given by the minimum of the set \eqref{min cas 3} with $k=2$ and $-\mu_{k-1}=-\mu_{1}=0$. In order to compute this minimum, we have to compute $\Delta(0)$. We find $\Delta(0)=\{-\frac{1}{2},-\frac{1}{4}\}$. Then, we have to compute positive lower bounds $\vareps(-\frac{1}{2})$ and $\vareps(-\frac{1}{4})$ on $\epsilon(-\frac{1}{2})$ and $\epsilon(-\frac{1}{4})$ respectively. We have seen in Example \ref{ex:base cas RS} and Example \ref{ex:rec step case 1} that we can take $\vareps(-\frac{1}{2})=\frac{1}{4}$ and $\vareps(-\frac{1}{4})=\frac{1}{8}$. Using the calculations of section \ref{sec:calc pi Psi RS}, it is easily seen that the minimum of \eqref{min cas 3} with $k=2$ and $-\mu_{k-1}=-\mu_{1}=0$ is the minimum of the following numbers:
\begin{itemize}
\item $\vareps(w')\ellmahl^{d_{0,w'} - \alpha_{0}}=\frac{1}{4} 2^{1-0}=\frac{1}{2}$ for $w'=-\frac{1}{2}$; 
\item $\vareps(w')\ellmahl^{d_{0,w'} - \alpha_{0}}=\frac{1}{8} 2^{2-0}=\frac{1}{2}$ for $w'=-\frac{1}{4}$; 
 \item $\mu_{1} - \mu_{0}=+\infty$;
 \item $\min(\pi(\Psi(-\mu_{1}))\setminus \{-\mu_{1}\}) +\mu_{1}=\min(\pi(\Psi(0))\setminus \{0\})=1$.
\end{itemize}
This minimum is equal to $\frac{1}{2}$. 
Thus, a positive lower bound of $\epsilon(0)$ is given by $\vareps(0)=\frac{1}{2}$. 
\end{ex}

\begin{proof} 
This follows immediately from Lemma \ref{lem:espilon_Delta} when replacing $k$ with $k-1$ since $\vareps(w')\leq\epsilon(w')$ for any $w' \in \Delta(-\mu_{k-1})$.
\end{proof}

\subsection{Pseudo-code}\label{sec: lbeps pc}

Here is a pseudo-code transcription of the algorithm outlined in section \ref{sec:algo plb epsilon structure}. Its core is the function Lower\textunderscore Bound\textunderscore $\epsilon$\textunderscore param defined in Algorithm \ref{algo:lower bound eps avec param} below. 

\begin{algorithm} \label{algo: lower bound epsilon}
$_{}$ \vskip 1 pt
\noindent \rule{\linewidth}{1mm}
{\bf Input}: $L$ a Mahler operator with coefficients in $\bK[z]$, $v \in \Z_{d,\ellmahl}$. \\
{\bf Output}: positive lower bound on $\epsilon(v)$.\vskip 1pt
\noindent \rule{\linewidth}{1mm} 
\noindent \textbf{def} Lower\textunderscore Bound\textunderscore $\epsilon$ ($L,v$)
\begin{itemize}
\item[] \textbf{for}  $k$ \textbf{from} $\kappa$ \textbf{to} $1$
 \begin{itemize}
\item[] set $\theta_{k}$=Lower\textunderscore Bound\textunderscore $\epsilon$\textunderscore param $(L,k,(\theta_{i})_{k+1 \leq i \leq \kappa}, -\mu_{k})$
\end{itemize}
 \item[] \textbf{end for}
  \item[] return Lower\textunderscore Bound\textunderscore $\epsilon$\textunderscore param $(L,0,(\theta_{i})_{1 \leq i \leq \kappa}, v)$
\end{itemize}
\textbf{end def}\\
 \noindent \rule{\linewidth}{0.5mm}
\end{algorithm}

\begin{algorithm} \label{algo:lower bound eps avec param}
$_{}$ \vskip 1 pt
\noindent \rule{\linewidth}{1mm}
{\bf Input}: $L$ a Mahler operator with coefficients in $\bK[z]$, $\kappaprime \in \{0,\ldots,\kappa\}$, $(\theta_{i})_{\kappaprime +1 \leq i \leq \kappa} \in \mathbb Q_{>0}^{\kappa-\kappaprime}$,
$v \in \Z_{d,\ellmahl}$. \\
{\bf Output}: positive lower bound on $\epsilon(v)$ provided that:
\begin{itemize}
\item  either $\kappaprime=\kappa$ and $v \leq -\mu_{\kappaprime}=-\mu_{\kappa}$; 
\item or $\kappaprime \in \{0,\ldots,\kappa-1\}$,  $v \leq -\mu_{\kappaprime}$ and $\theta_{\kappaprime+1},\ldots,\theta_{\kappa}$ are positive lower bounds on $\epsilon(-\mu_{\kappaprime+1}),\ldots,\epsilon(-\mu_{\kappa})$ respectively.
\end{itemize}
\vskip 1pt
\noindent \rule{\linewidth}{1mm}

\noindent \textbf{def} Lower\textunderscore Bound\textunderscore $\epsilon$\textunderscore param $(L,\kappaprime,(\theta_{i})_{\kappaprime+1 \leq i \leq \kappa}, v)$ 

	\begin{linenumbers}
\begin{itemize}
	\item[] \textbf{if}  $v < -\mu_{\kappa}$ \textbf{then} \linelabel{debutStep1}
 \begin{itemize}
 	\item[] return $-\mu_{\kappa} - v$
 	\end{itemize}
	\item[] \textbf{end if}
 	\item[]
	\textbf{if}  $v = -\mu_{\kappa}$ \textbf{then} 
 	\begin{itemize}
  		\item[] set $S=\V_{1} \setminus \{-\mu_{\kappa}\}$
 		\item[] \textbf{if} $S \neq \emptyset $ \textbf{then}
  		\begin{itemize}
 			\item[] return ${\blue  \min S}+ \mu_{\kappa}$
 		\end{itemize}
  		\item[] \textbf{otherwise}
   		\begin{itemize}
 			\item[] return $1$
 		\end{itemize}
 		\item[] \textbf{end if}
	\end{itemize}
 	\item[] \textbf{end if} \linelabel{finStep1}
	\item[] 	\textbf{if} $v \in ]-\mu_{k},-\mu_{k-1}[$ for some $k \in \{\kappaprime+1,\ldots,\kappa\}$ \textbf{then} \linelabel{debutStep2}
	\begin{itemize}
	        \item[] set $(w',m)=$Lower\textunderscore Bound\textunderscore $\epsilon$\textunderscore interval $(L,k-1,(\theta_{i})_{k \leq i \leq \kappa}, v)$ 
	\item[] return $m$
		 \end{itemize}
 	\item[]  \textbf{end if} \linelabel{finStep2} 
	\item[] \textbf{if} $v=-\mu_{k-1}$ for some $k \in \{\kappaprime+1,\ldots,\kappa\}$ \textbf{then} \linelabel{debutStep3}
 	\begin{itemize}
	\item[] \textbf{for} $w' \in \Delta(-\mu_{k-1})$
	\begin{itemize}
	\item[]
	 set $m_{w'}$=Lower\textunderscore Bound\textunderscore $\epsilon$\textunderscore param $(L,k-1,(\theta_{i})_{k \leq i \leq \kappa},w')$\linelabel{line_reccall2}
	 \end{itemize}
	 \item[] \textbf{end for}
 		\item[] return the minimum of \linelabel{finalstepalgo}
\begin{eqnarray*}
\nonumber & \{m_{w'}\ellmahl^{d_{-\mu_{k-1},w'} - \alpha_{k-2}} \ \vert \ w' \in \Delta(-\mu_{k-1}) \}&
 \\ &\text{ and } & 
 \\
\nonumber & \{
 \mu_{k-1} - \mu_{k-2}, \min(\pi(\Psi(-\mu_{k-1}))\setminus \{-\mu_{k-1}\}) +\mu_{k-1}\} &
\end{eqnarray*}
where $d_{-\mu_{k-1},w'}$ is defined by \eqref{eq:def_dkw}.
 	\end{itemize}
	\textbf{end if}  \linelabel{finStep3}
 \end{itemize}
\end{linenumbers}
\textbf{end def} 
\\
\noindent \rule{\linewidth}{0.5mm}
\end{algorithm}

In Algorithm \ref{algo:lower bound eps avec param}: 
\begin{itemize}
\item the lines \ref{debutStep1}--\ref{finStep1} correspond to the base case considered in section \ref{sec: plb eps base case}; 
\item the lines \ref{debutStep2}--\ref{finStep2} correspond to the recursive step considered in section \ref{sec: plb eps rec step 1} which can itself be encoded as a recursive algorithm, namely Algorithm~\ref{algo:lower bound eps avec param interval} below;
\item the lines \ref{debutStep3}--\ref{finStep3} correspond to the recursive step considered in section \ref{sec: plb eps rec step 2}. 
\end{itemize}

\begin{algorithm} \label{algo:lower bound eps avec param interval}
$_{}$ \vskip 1 pt
\noindent \rule{\linewidth}{1mm}
{\bf Input}: $L$ a Mahler operator with coefficients in $\bK[z]$, $\kappaprime \in \{0,\ldots,\kappa-1\}$, $(\theta_{i})_{\kappaprime +1 \leq i \leq \kappa} \in \mathbb Q_{>0}^{\kappa-\kappaprime}$,
$w \in \Z_{d,\ellmahl}$. \\
{\bf Output}: positive lower bound on $\epsilon(w)$ provided that  $w<-\mu_{\kappaprime}$ 
and $\theta_{\kappaprime+1},\ldots,\theta_{\kappa}$ are positive lower bounds on $\epsilon(-\mu_{\kappaprime+1}),\ldots,\epsilon(-\mu_{\kappa})$ respectively.
 
\vskip 1pt
\noindent \rule{\linewidth}{1mm}

\noindent \textbf{def} Lower\textunderscore Bound\textunderscore $\epsilon$\textunderscore interval $(L,\kappaprime,(\theta_{i})_{\kappaprime+1 \leq i \leq \kappa}, w)$ 
\linenumbers[1]
	\begin{linenumbers}
\begin{itemize}
\item[] \textbf{if} $w<-\mu_{\kappaprime+1}+\theta_{\kappaprime+1}$ \textbf{then} \linelabel{finStep2a}
\begin{itemize}
\item[] \textbf{if} $w<-\mu_{\kappaprime+1}$ \textbf{then} \linelabel{finStep2b}
\begin{itemize}
\item[] set $m=$Lower\textunderscore Bound\textunderscore $\epsilon$\textunderscore param $(L,\kappaprime+1,(\theta_{i})_{\kappaprime+2 \leq i \leq \kappa}, w)$
\end{itemize}
\item[] \textbf{otherwise} \linelabel{finStep2c}
\begin{itemize}
\item[] set $m=-\mu_{\kappaprime+1}+\theta_{\kappaprime+1} - w$
\end{itemize}
\item[] \textbf{end if} 
\item[] return $(w,m)$
\end{itemize}
\item[] \textbf{end if} 
\item[] \textbf{if} $w\geq -\mu_{\kappaprime+1}+\theta_{\kappaprime+1}$ \textbf{then} \linelabel{finStep2d}
\begin{itemize}
 \item[] \textbf{for} each $v \in \Delta(w)$
\begin{itemize}
 \item[] compute $\boldsymbol b_{v}$=Lower\textunderscore Bound\textunderscore $\epsilon$\textunderscore interval $(L,\kappaprime,(\theta_{i})_{\kappaprime+1 \leq i \leq \kappa}, v)$
\end{itemize}
 \item[] \textbf{end for}
  \item[] set $m$ to the minimum of 
   \begin{eqnarray}\label{eq:setpourminvertexnotleaf}
  && \{m'\ellmahl^{d_{w,w'}-\alpha_{k}} \ \vert \ (w',m')=\boldsymbol b_{v} \text{ for some $v \in \Delta(w)$}\}\\
 &\cup& \{
 -\mu_{k} - w, \min(\pi(\Psi(w))\setminus \{w\}) - w\}\nonumber
\end{eqnarray}
where $d_{w,w'}$ is defined by \eqref{eq:def_dkw}
\item[] return $(w,m)$
\end{itemize}
\item[] \textbf{end if} 
\end{itemize}
\end{linenumbers}
\textbf{end def} 
\\
\noindent \rule{\linewidth}{0.5mm}
\end{algorithm}

\subsection{An algorithm for computing a positive lower bound on $\minepsilon$}\label{sec:algo low bound varepsilon}

To compute a lower bound on $\minepsilon=\min\big\{\epsilon(-\mu_1),\ldots,\epsilon(-\mu_{\kappa}),(d\ellmahl^n)^{-1}\big\}$, we can simply run Algorithm \ref{algo: lower bound epsilon} $\kappa$ times to compute positive lower bounds on $\epsilon(-\mu_1),\ldots,\epsilon(-\mu_{\kappa})$. However, the calculations would involve numerous redundancies. From an algorithmic point of view, it is better to use the following algorithm which eliminates these redundancies and which is an obvious modification of Algorithm \ref{algo: lower bound epsilon}.

	\begin{algorithm} \label{algo:minepsilon}
		$_{}$ \vskip 1 pt
		\noindent \rule{\linewidth}{1mm}
		{\bf Input}: $L$ a Mahler operator with coefficients in $\bK[z]$\\
		{\bf Output}: positive lower bound on $\minepsilon=\min\big\{\epsilon(-\mu_1),\ldots,\epsilon(-\mu_{\kappa}),(d\ellmahl^n)^{-1}\big\}$.\vskip 1pt
		\noindent \rule{\linewidth}{1mm} 
		\noindent \textbf{def} Lower\textunderscore Bound\textunderscore $\minepsilon$ ($L$)
	\begin{itemize}
			\item[] \textbf{for}  $k$ \textbf{from} $\kappa$ \textbf{to} $1$
			\begin{itemize}
				\item[] set $\theta_{k}$=
			Lower\textunderscore Bound\textunderscore $\epsilon$\textunderscore param $(L,k,(\theta_{i})_{k+1 \leq i \leq \kappa}, -\mu_{k})$	
			\end{itemize}
			\item[] \textbf{end for}
			\item[] return the minimum of $\theta_1,\ldots,\theta_{\kappa}$ and $\frac{1}{d\ellmahl^n}$
		\end{itemize}
		\textbf{end def}\\
		\noindent \rule{\linewidth}{0.5mm}
\end{algorithm}

\section{The property $\star_{\V}$}\label{sec:approx sol}

In this section, we consider a Mahler operator 
\begin{equation}\label{operateur L bla}
L= a_{n}\pz \malop{\ellmahl}^{n} + a_{n-1}\pz\malop{\ellmahl}^{n-1} + \cdots + a_{0}\pz 
\end{equation}
with coefficients $a_{0}\pz,\ldots,a_{n}\pz \in \bK[z]$ such that $a_{0}\pz a_{n}\pz \neq 0$. 
We recall the following notation:
$$
\Sol(L,\Hahn) = \{f\pz  \in \Hahn \ \vert \ L(f\pz )=0 \}. 
$$ 

\begin{hypo}\label{hyp sol}
Throughout this section, we let $\V$ be a subset of $\Q$ satisfying the following properties: 
\begin{enumerate}
 \item 
 $
\Sol(L,\Hahn) \subset \Hahn_{\vert \W}
$; 
 \item  $\W$ is well-ordered; 
 \item  $-\S(L) \subset \W$;
 \item  \label{hypo pi psi dans V}
 $
\bigcup_{v \in \W} \pi( \Psi(v)) =\W
$.
 \end{enumerate}
\end{hypo}

Theorem \ref{theo: un support bis} ensures that such a set $\V$ exists. 

\begin{defi}\label{defi:prop star s}
We say that a subset $\Rr$ of $\Q$ satisfies property $\star_{\V}$ if :
\begin{enumerate}
	\item[a.] $-\S(L) \subset \Rr \subset \V$; 
	\item[b.] 
 $\bigcup_{v \in \V \setminus \Rr} \pi(\Psi(v)) = \V \setminus \Rr$. 
\end{enumerate}
\end{defi}
The interest of property $\star_{\V}$ lies in the following result. 

\begin{theo}\label{theo: si star alors isom}
If $\Rr$ satisfies $\star_{\V}$, then the $\bK$-linear map 
\begin{equation}\label{eq:res hahn to rr}
\begin{array}{cccc}
\res{\Rr} : & \Hahn &\rightarrow& \Hahn_{\vert \Rr} \\
&     f\pz = \sum_{\gamma \in \mathbb Q} f_{\gamma} z^{\gamma} &\mapsto& f_{\vert \Rr}\pz = \sum_{\gamma \in \Rr} f_{\gamma} z^{\gamma} 
\end{array}
\end{equation}
induces a $\bK$-linear isomorphism  
\begin{equation}\label{eq:res sol to cond}
\Sol(L,\Hahn) \xrightarrow[]{\sim} \Cond{\Rr}
\end{equation}
where
\begin{eqnarray*}
 \Cond{\Rr} &=& \{f\pz  \in \Hahn_{\vert \Rr} \ \vert \ \pi(\supp L(f\pz )) \cap \Rr = \emptyset\} \\ 
 &=& \{f\pz  \in \Hahn_{\vert \Rr} \ \vert \ \supp(L(f\pz )) \cap \minP(\Rr) = \emptyset\}. 
\end{eqnarray*}
\end{theo}

The proof of this result is given in section \ref{preuve theo Lf eg g} below. It relies on certain preliminary results gathered in the next section. 

\begin{rem}
In Theorem \ref{th R star}, we will prove that there exists a finite subset $\Rr$ of $\Q$ satisfying property $\star_{\V}$ and we will give an algorithm to compute it. However, in this section, $\Rr$ is not required to be finite. 
\end{rem}

\subsection{Preliminary results}\label{sec: prelim res}

Through this section, we consider a subset $\Rr$ of $\Q$ satisfying $\star_{\V}$. 
We introduce the following sets:
$$
\Gamma=\V  \setminus \Rr \text{ and } \,
\Lambda = \psi(\Gamma).
$$
The principal aim of this section is to prove the following result. 

\begin{prop}\label{prop:for all g there is f}
 For all $g\pz  \in \Hahn_{\vert \Lambda}$, there exists $f\pz  \in \Hahn_{\vert \Gamma}$ such that $L(f\pz )=g\pz $.  
\end{prop}

The proof of this result is given at the end of this subsection. The following lemmas are technical results used in the proof of  Proposition \ref{prop:for all g there is f}. On first reading, the reader can admit these lemmas and read the proof of  Proposition \ref{prop:for all g there is f} directly.

\begin{lem}\label{lem Gamma Lambda}
	We have:  
\begin{itemize}
\item $\Gamma = \pi(\Lambda)$;
\item  $\Lambda=\bigcup_{\gamma \in \Gamma} \Psi(\gamma)$; 
\item $\Gamma \cap -\mathcal{S}(L) = \emptyset$.
\end{itemize}	
\end{lem}
\begin{proof}
	The equality $\Gamma = \pi(\Lambda)$ follows immediately from the fact that $\minP$ and $\pi$ are inverse of each other by Lemma \ref{lem proprietes pi minP}.
	
			By definition of $\Gamma$ and condition b. of Definition \ref{defi:prop star s}, we have
			 	$$
			 	\Gamma = \V \setminus \Rr = \bigcup_{{v} \in \V \setminus \Rr} \pi(\Psi(v)) = \bigcup_{{v} \in \Gamma} \pi(\Psi({v})) \,.
			 	$$
			 	Applying $\minP$ to the latter equality and using the fact that $\minP$ and $\pi$ are inverse of each other by Lemma \ref{lem proprietes pi minP}, we obtain $\Lambda=\bigcup_{\gamma \in \Gamma} \Psi(\gamma)$.
				
				Last, we have $\Gamma \cap -\mathcal{S}(L) = \emptyset$ because $-\S(L)\subset \Rr$ and $\Gamma \cap \Rr=\emptyset$ by definition of $\Gamma$. 		 
\end{proof}

\begin{lem}
 The subsets $\Gamma$ and $\Lambda$ of $\Q$ are well-ordered. 
\end{lem}

\begin{proof}
	Since $\Gamma$ is a subset of $\V$ which is well-ordered, $\Gamma$ is well-ordered.
Lemma \ref{lem proprietes pi minP} ensures that $\minP : \mathbb Q \rightarrow \mathbb Q$ is \strictementcroissant. Thus, the fact that $\Lambda=\minP(\Gamma)$ is well-ordered follows from the fact that $\Gamma$ is well-ordered.	
 \end{proof}

\begin{lem}\label{for all g there exists a gamma}
For any $g\pz  \in \Hahn_{\vert \Lambda}\setminus \{0\}$, there exists $a \in \bK $ such that 
$$
\val (L(a z^{\gamma})-g\pz ) > \val g\pz ,
$$ 
where $\gamma=\pi(\val g\pz )$. 
In particular, $\gamma \in \pi(\Lambda)=\Gamma$.
 \end{lem}

\begin{proof} Since $\supp g\pz \subset \Lambda$, we have $\val g\pz  \in \Lambda$ and, hence,  $\gamma=\pi(\val g\pz ) \in \pi(\Lambda)$. But,  $\pi(\Lambda) = \Gamma$ and $\Gamma \cap -\mathcal{S}(L) = \emptyset$  by Lemma \ref{lem Gamma Lambda}. So, $\gamma \not\in -\S(L)$ and Lemma \ref{lem:val Lf} ensures that $\val L(z^{\gamma}) = \minP(\gamma)$.  Since $\minP$ and $\pi$ are inverse of each other by Lemma \ref{lem proprietes pi minP}, we get  $\val L(z^{\gamma}) = \val g\pz$, {\it i.e.}, there exists $c \in \bK^{\times}$ such that 
$$
L(z^{\gamma}) = c z^{\val g\pz } + \text{Hahn series of higher-order valuation}.
$$
Therefore, $a=c^{-1} g_{\val g\pz}$ has the expected property. 
\end{proof}

\begin{lem}\label{lem:prop minp}
If $\gamma \in \Gamma$ and $\lambda \in \Lambda$ are such that, for all $x \in \Gamma \cap ]-\infty,\gamma[$, $\lambda > \minP(x)$, then $\lambda \geq \minP(\gamma)$.
\end{lem}

\begin{proof}
Since $\pi$ and $\minP$ are inverse of each other by Lemma \ref{lem proprietes pi minP}, we have 
 $\lambda = \minP(y) $ with $y=\pi(\lambda) \in \pi(\Lambda)=\Gamma$.  By hypothesis, for all $x \in \Gamma \cap ]-\infty,\gamma[$, we have $\minP(y) =\lambda > \minP(x)$. Since $\minP$ is \strictementcroissant{} by Lemma \ref{lem proprietes pi minP}, we get, for all $x \in \Gamma \cap ]-\infty,\gamma[$, $y>x$. So, $y \in \Gamma \cap [\gamma,+\infty[$. In particular, we have $y \geq \gamma$ and, since $\minP$ is \strictementcroissant{} by Lemma \ref{lem proprietes pi minP}, $\lambda=\minP(y)  \geq \minP(\gamma)$ as claimed. 
\end{proof}

\begin{lem}\label{lem:if val strict then eg}
 Suppose that $f\pz,f'\pz \in \Hahn_{\vert \Gamma}$ are such that $\val (L(f\pz)-g\pz)>\minP(y)$ and $\val (L(f'\pz)-g\pz)>\minP(y')$ for some $g\pz \in \Hahn$ and some $y,y' \in \Gamma$. Then, $f\pz=f'\pz$ on $\Gamma \cap ]-\infty , \min \{y,y'\}]$.
\end{lem}

\begin{proof}
We argue by contradiction. 
Assume on the contrary that $f\pz \neq f'\pz$ on $\Gamma \cap ]-\infty , \min \{y,y'\}]$. Up to interchanging the roles of $f\pz$ and $f'\pz$, we can assume that $y \leq y'$ and, hence, that $f\pz \neq f'\pz$ on $\Gamma \cap ]-\infty , y]$. Using the fact that $\Gamma$ is well ordered, we can assume that $y$ is minimal with respect to the property ``$f\pz \neq f'\pz$ on $\Gamma \cap ]-\infty , y]$''. On the one hand, one can characterize $y$ in terms of $h\pz=f'\pz-f\pz \in \Hahn_{\vert \Gamma} \setminus \{0\}$ as the minimal element of $\Gamma$ such that $h \pz \neq 0$ on $\Gamma \cap ]-\infty , y]$. So, $y=\val h \pz$. Since $\val h \pz=y \in \Gamma$ and $\Gamma \cap -\S(L)=\emptyset$ by Lemma \ref{lem Gamma Lambda}, we have $\val h \pz \not\in -\S(L)$ and it follows from Lemma \ref{lem:val Lf} that 
\begin{equation}\label{valLdelta=minPy}
\val L(h\pz)=\minP(y). 
\end{equation}
On the other hand, we have
$$
L(h\pz) = (L(f'\pz)-g\pz)-(L(f\pz)-g\pz). 
$$
Applying the $z$-adic valuation $\val$ to the latter equality and using \eqref{val ult metric}, we get 
$$
\val L(h\pz) \geq  \min \{\val (L(f\pz)-g\pz),\val (L(f'\pz)-g\pz)\}.
$$
But, by hypothesis, we have $\val (L(f\pz)-g\pz) > \minP(y)$ and $\val (L(f'\pz)-g\pz) > \minP(y') \geq \minP(y)$, the latter inequality $\minP(y') \geq \minP(y)$ following from the facts that $y'\geq y$ and that $\minP$ is \strictementcroissant{} by Lemma \ref{lem proprietes pi minP}. Therefore,  
$\val L(h\pz) > \minP(y).
$ 
This contradicts \eqref{valLdelta=minPy}. 
\end{proof}

\begin{lem}\label{lem:exists unique}
Consider $y \in \Gamma \cup \{+\infty\}$ and a family $(f_{x})_{x \in \Gamma \cap ]-\infty,y[}$ of Hahn series such that, for all $x \in \Gamma \cap ]-\infty,y[$, $f_x \in \Hahn_{\vert \Gamma}$. Suppose that, for all $x,x' \in \Gamma \cap ]-\infty,y[$, we have $f_{x}\pz=f_{x'}\pz$ on $\Gamma \cap ]-\infty , \min \{x,x'\}]$. Then, there exists 
$f\pz\in\Hahn_{\vert \Gamma \cap ]-\infty , y[}$ such that, for all $x \in \Gamma \cap ]-\infty , y[$, $f\pz=f_{x}\pz$ on $\Gamma \cap ]-\infty , x]$. \end{lem}

\begin{proof}
	We set, for all $x \in \Gamma \cap ]-\infty,y[$, $f_x\pz=\sum_{\gamma \in \Gamma} a_{x,\gamma} z^{\gamma}$. Set 
		$$f\pz=\sum_{\gamma \in \Gamma \cap ]-\infty , y[} a_{\gamma,\gamma} z^{\gamma} \in \Hahn_{\vert \Gamma \cap ]-\infty , y[}\,.
		$$
		Let $x \in \Gamma \cap ]-\infty,y[$ and $\gamma \in \Gamma \cap ]-\infty , x]$. Using the hypothesis of the lemma with $x'=\gamma$, we have $f_{x}\pz=f_{\gamma}\pz$ on $\Gamma \cap ]-\infty , \gamma]$. In particular, looking at the coefficients of $z^\gamma$, we obtain $a_{x,\gamma}=a_{\gamma,\gamma}$. Thus, $f\pz=f_{x}\pz$ on $\Gamma \cap ]-\infty , x]$. 
\end{proof}

\begin{proof}[Proof of Proposition \ref{prop:for all g there is f}]
Consider $g\pz \in \Hahn_{\vert \Lambda}$. We have to prove that there exists $\overline{f}\pz \in \Hahn_{\vert \Gamma}$ such that $L(\overline{f}\pz)=g\pz$. We split the proof in two main steps. 
\vskip 5pt
\noindent {\bf Step 1. } Let us first prove that, for all $y \in \Gamma$, there exists $f_{y}\pz \in \Hahn_{\vert \Gamma}$ such that $\val (L(f_y\pz)-g\pz)>\minP(y)$. We argue by contradiction: we assume that this is not true, {\it i.e.}, that the set 
$$
Y=\{y \in \Gamma \ \vert \ \forall f\pz \in \Hahn_{\vert \Gamma}, \val (L(f\pz)-g\pz ) \leq \minP(y) \}
$$ 
is nonempty. Since $Y$ is a nonempty subset of the well-ordered set $\Gamma$, it has a minimal element $y_{\min}$. 

We claim that there exists $f\pz \in \Hahn_{\vert \Gamma \cap ]-\infty , y_{\min}[}$ such that
\begin{equation}\label{ineg val psi x}
\val (L(f\pz)-g\pz) > \minP(x) 
\end{equation}
for all $x \in \Gamma \cap ]-\infty , y_{\min}[$, and such that 
\begin{equation}\label{val dans lambda}
\val (L(f\pz)-g\pz) \in \Lambda. 
\end{equation}
Indeed, for all $x \in \Gamma \cap ]-\infty,y_{\min}[$, we have $x \in \Gamma \setminus Y$ and, hence, there exists $f_x \in \Hahn_{\vert \Gamma}$ such that 
\begin{equation}\label{ineq val Lf-g>minPx}
\val (L(f_x\pz )-g\pz )>\minP(x).  
\end{equation}
According to Lemma \ref{lem:if val strict then eg}, we have, for any $x,x' \in \Gamma \cap ]-\infty,y_{\min}[$, $f_{x}\pz=f_{x'}\pz$ on $\Gamma \cap ]-\infty , \min \{x,x'\}]$. Lemma \ref{lem:exists unique} ensures that there exists  $f\pz\in\Hahn_{\vert \Gamma \cap ]-\infty , y_{\min}[}$ such that, for all $x \in \Gamma \cap ]-\infty , y_{\min}[$, $f\pz=f_{x}\pz$ on $\Gamma \cap ]-\infty , x]$. Let us prove that $f\pz$ satisfies \eqref{ineg val psi x} and \eqref{val dans lambda}.  Let us first note that, for all $x \in \Gamma \cap ]-\infty , y_{\min}[$, 
\begin{equation}\label{ineq L(f-fx)}
\val L(f\pz-f_{x}\pz) \geq \minP(\val (f\pz-f_{x}\pz)) > \minP(x);  
\end{equation}
indeed, the first inequality follows from Lemma \ref{lem:val Lf}, the second inequality follows from the facts that $\val (f\pz-f_{x}\pz)>x$ and that $\minP$ is \strictementcroissant{} by Lemma \ref{lem proprietes pi minP}. 
Using \eqref{val ult metric} and, then, the inequalities \eqref{ineq val Lf-g>minPx} and \eqref{ineq L(f-fx)}, we get, for all $x \in \Gamma \cap ]-\infty , y_{\min}[$, 
\begin{multline}
\val (L(f\pz)-g\pz)=\val(L(f\pz-f_{x}\pz)+L(f_{x}\pz)-g\pz) \\
\geq \min \{\val L(f\pz-f_{x}\pz), \val (L(f_{x}\pz)-g\pz)  \} > \minP(x). 
\end{multline}
This justifies \eqref{ineg val psi x}. Moreover, we have $L(f\pz)-g\pz \neq 0$ because $Y$ is nonempty, so 
$$\val (L(f\pz)-g\pz) \in \supp (L(f\pz)-g\pz)\subset \supp L(f\pz) \cup \supp g\pz \subset \bigcup_{\gamma \in \Gamma}\Psi(\gamma) \cup \Lambda \subset \Lambda,$$ the latter two inclusions following from Lemma \ref{lem:supp Lf sub P supp f} and 
Lemma \ref{lem Gamma Lambda} respectively. This justifies \eqref{val dans lambda} and, hence, our claim toward the existence of $f$. 

We fix $f\pz \in \Hahn_{\vert \Gamma \cap ]-\infty , y_{\min}[}$ satisfying \eqref{ineg val psi x} and \eqref{val dans lambda}. We can apply Lemma \ref{lem:prop minp} to $\lambda=\val (L(f\pz)-g\pz)$ and to $\gamma=y_{\min}$ and we obtain  
$$
\val (L(f\pz)-g\pz) \geq \minP(y_{\min}).
$$ 
We have already seen that $L(f\pz)-g\pz \neq 0$ and that $\supp (L(f\pz)-g\pz)\subset \Lambda$. So, Lemma \ref{for all g there exists a gamma} ensures that there exists $a \in \bK$ and $\gamma \in \Gamma$ such that
$$\val (L(az^{\gamma})+L(f\pz)-g\pz)>\val (L(f\pz)-g\pz)\geq \minP(y_{\min}).$$
Therefore, $f'\pz=az^{\gamma}+f\pz \in \Hahn_{\vert \Gamma}$ satisfies 
$$\val (L(f'\pz)-g\pz) > \minP(y_{\min}).
$$ 
This contradicts the fact that $y_{\min}$ belongs to $Y$. So, $Y$ is empty and, hence, we have proved that, 
for all $y \in \Gamma$, there exists $f_{y}\pz \in \Hahn_{\Gamma}$ such that $\val (L(f_y\pz)-g\pz)>\minP(y)$. 
\vskip 5pt
\noindent {\bf Step 2. } Lemma \ref{lem:if val strict then eg} ensures that, for all $y,y' \in \Gamma$, we have $f_{y'}\pz=f_{y}\pz$ on $\Gamma \cap ]-\infty , \min\{y,y'\}]$. According to Lemma \ref{lem:exists unique} applied with $y=+\infty$, there exists  $\overline{f}\pz \in \Hahn_{\Gamma}$ such that, for all $y \in \Gamma$, $\overline{f}\pz=f_{y}\pz$ on  $\Gamma \cap ]-\infty , y]$. Arguing as we did above for proving \eqref{ineg val psi x}, we see that, for all $y \in \Gamma$, 
$
\val(L(\overline{f}\pz)-g\pz) > \minP(y).
$
This implies that $L(\overline{f}\pz)-g\pz=0$ because, otherwise, $\val(L(\overline{f}\pz)-g\pz)$ would belong to $\supp(L(\overline{f}\pz)-g\pz) \subset \Lambda$ but not to $\minP(\Gamma)$ and this would contradict the fact that $\Lambda=\minP(\Gamma)$ by definition.
\end{proof}

\subsection{Proof of Theorem \ref{theo: si star alors isom}}\label{preuve theo Lf eg g} We recall the  following notations introduced in section \ref{sec: prelim res}:  
$$\Gamma=\V \setminus \Rr, \ \ \Lambda = \minP(\Gamma)
$$
and 
$$
\Cond{\Rr} = \{f\pz  \in \Hahn_{\vert \Rr} \ \vert \ \supp(L(f\pz )) \cap \minP(\Rr) = \emptyset\}.
$$
Proving Theorem \ref{theo: si star alors isom} is equivalent to proving the following properties relative to the $\bK$-linear map $\res{\Rr}$ defined by \eqref{eq:res hahn to rr}:  $\res{\Rr}(\Sol(L,\Hahn))=\Cond{\Rr}$ and $\ker(\res{\Rr}) \cap \Sol(L,\Hahn)=\{0\}$. 
		Before proving these properties, note that 
\begin{equation}\label{CR in terms of Lambda}
 \Cond{\Rr} 
 =\{f\pz  \in \Hahn_{\vert \Rr} \ \vert \ \supp L(f\pz ) \subset  \Lambda\}
=  \{f\pz  \in \Hahn_{\vert \Rr} \ \vert \ L(f\pz ) \in \Hahn_{\vert \Lambda}\}.
\end{equation}
Indeed, since $\Lambda=\minP(\Gamma)=\minP(\V\setminus \Rr)=\minP(\V)\setminus \minP(\Rr)$, in order to prove the equality \eqref{CR in terms of Lambda}, it is sufficient to prove that, for any $f\pz \in \Hahn_{\vert \Rr}$, we have $\supp L(f\pz) \subset \minP(\V)$. As a matter of fact, the latter property is true since $\supp L(f) \subset \bigcup_{v \in \supp f} \Psi(v)$ by Lemma~\ref{lem:supp Lf sub P supp f}, $\bigcup_{v \in \supp f} \Psi(v)\subset \bigcup_{v \in \V} \Psi(v)$ because $\supp f \subset \Rr \subset \V$ and $\bigcup_{v \in \V} \Psi(v) = \psi(\V)$ by (\ref{hypo pi psi dans V}) of Hypothesis \ref{hyp sol}.
\vskip 5pt
\noindent {\it Proof of $\res{\Rr}(\Sol(L,\Hahn))\subset \Cond{\Rr}$.} Consider $f\pz \in \Sol(L,\Hahn)$. 
Consider the decomposition $f\pz=f_{\vert \Rr}\pz+f_{\vert \gamma}\pz$. Applying $L$ to this equality, we get $0=L(f_{\vert \Rr}\pz)+L(f_{\vert\gamma}\pz)$, so 
\begin{equation}
 L(f_{\vert \Rr}\pz)=-L(f_{\vert \gamma}\pz).
\end{equation}
It follows from Lemma \ref{lem:supp Lf sub P supp f} that $\supp L(f_{\vert \gamma}\pz) \subset \bigcup_{v \in \gamma}\Psi(v)=\Lambda$, the latter equality coming from Lemma \ref{lem Gamma Lambda}. Therefore, 
$$
L(f_{\vert \Rr}\pz)= -L(f_{\vert \gamma}\pz) \in \Hahn_{\vert \Lambda}.
$$ 
Using \eqref{CR in terms of Lambda}, this proves that the image of $f\pz$ by $\res{\Rr}$ belongs to $\Cond{\Rr}$. 
\vskip 5pt
\noindent {\it Proof of $\Cond{\Rr} \subset \res{\Rr}(\Sol(L,\Hahn))$. } Consider $f_0\pz \in \Cond{\Rr}$. It follows from \eqref{CR in terms of Lambda} that $L(f_0\pz) \in \Hahn_{\vert \Lambda}$. 
Proposition \ref{prop:for all g there is f} ensures that there exists $f_{1}\pz \in \Hahn_{\vert \Gamma}$ such that 
$
L(f_{1}\pz)=-L(f_0\pz)
$. 
Then $f\pz= f_{0}\pz+f_{1}\pz$ belongs to $\Sol(L,\Hahn)$ and its image by $\res{\Rr}$ is $f_{0}\pz$.
\vskip 5pt
\noindent {\it Proof of $\ker(\res{\Rr}) \cap \Sol(L,\Hahn)=\{0\}$. } Let $f\pz \in \ker(\res{\Rr}) \cap \Sol(L,\Hahn)$. Then, $f\pz$ belongs to $\Hahn_{\vert \Gamma}$ and satisfies $L(f\pz)=0$. Lemma \ref{lem:val Lf eg g} ensures that $\val f\pz \in - \S(L) \cup\{+\infty\}$. But, $\val f\pz \in \Gamma \cup \{+\infty\}$ and $ - \S(L) \cap \Gamma = \emptyset$ by Lemma \ref{lem Gamma Lambda}. So, $\val f\pz =+\infty$ and, hence, $f\pz=0$.  
\vskip 5pt
This concludes the proof of Theorem \ref{theo: si star alors isom}.

\section{Computing an $\Rr$ containing $\E$ and satisfying $\star_{\V}$.}\label{sec constr R}

We use the notations of section \ref{sec:approx sol}: we consider the operator $L$ given by \eqref{operateur L bla} and  we let $\V$ be a subset of $\Q$ satisfying Hypothesis \ref{hyp sol}. 
Moreover, we let $\E$ be a finite subset of $\V$. 

\begin{defi}
We say that a set $\Rr$ satisfies property $\star_{\E,\V}$ if $\E \subset \Rr$ and $\Rr$ satisfies property $\star_\V$.  
\end{defi}

We shall now give a recursive construction of a finite set satisfying $\star_{\E,\V}$.

\begin{theo}\label{th R star}
The sequence $(\Rr_{i})_{i\geq 0}$ of subsets of $\V$ recursively defined by 
	$$\Rr_0= \Ecal \cup -\S(L)$$
and, for all $i \geq 0$, 
\begin{equation}\label{eq:def des Ri}
 \Rr_{i+1} =\{v \in \V \ \vert \ \pi(\Psi(v)) \cap \Rr_i \neq \emptyset \}
\end{equation}
is an eventually constant \croissant{} sequence of finite sets. 
The above recursive definition formula can be rewritten as follows, for all $i \geq 0$: 
\begin{equation}
\Rr_{i+1}=\bigcup_{(\ellmahl^{\alpha},\beta) \in \P(L)}  \ellmahl^{-\alpha} (\minP (\Rr_i)-\beta)  \cap \V   \label{form alt Ri+1}.
\end{equation}
Moreover, $\Rrinfty=\bigcup_{i \geq 0} \Rr_{i}$ is a finite set which satisfies $\star_{\E,\V}$. 
\end{theo}

\begin{ex}\label{ex:Ri et R pour L pour example}
The sets $\Rr_{i}$ and $\Rr$ are computed in section \ref{main algo on RS example} for the operator $L$ given by \eqref{L pour example}.
\end{ex}

\begin{proof}$_{}$\\
\noindent {\it Proof of the equality \eqref{form alt Ri+1}.} This equality follows from the following chain of equalities:
\begin{eqnarray*}
\{v \in \V \ \vert \ \pi(\Psi(v)) \cap \Rr_i \neq \emptyset \} &=& \{v \in \V \ \vert \ \Psi(v) \cap \minP (\Rr_i) \neq \emptyset \} \\
&=& \{v \in \V \ \vert \ \exists (\ellmahl^{\alpha},\beta) \in \P(L), v \ellmahl^{\alpha}+\beta  \in \minP (\Rr_i)  \}\\
&=&\bigcup_{(\ellmahl^{\alpha},\beta) \in \P(L)} \{v \in \V \ \vert \  v \ellmahl^{\alpha}+\beta  \in \minP (\Rr_i)  \}\\
&=&\bigcup_{(\ellmahl^{\alpha},\beta) \in \P(L)}  \ellmahl^{-\alpha} (\minP (\Rr_i)-\beta) \cap \V.
\end{eqnarray*}

$_{}$\\
\noindent {\it Proof of the fact that the $\Rr_{i}$ are finite sets.}   We argue by induction on $i\geq 0$. The base case $i=0$ is clear. We now proceed with the inductive step. Let us assume that $\Rr_{i}$ is finite for some $i\geq 0$. Then, $\minP (\Rr_i)$ is finite, so, for all $(\ellmahl^{\alpha},\beta) \in \P(L)$, $\ellmahl^{-\alpha} (\minP (\Rr_i)-\beta) \cap \V$ is finite. Since $\P(L)$ is finite, the equation \eqref{form alt Ri+1} shows that $\Rr_{i+1}$ is finite as well. This concludes the proof.

\vskip 5pt
\noindent {\it Proof of the fact that the sequence $(\Rr_{i})_{i\geq 0}$ is \croissant.} Consider $i \in \Z_{\geq 0}$. Lemma \ref{min Psi(s)=s} ensures that, for all $v \in \Rr_{i}$, we have $v \in \pi(\Psi(v))$, so $v \in \pi(\Psi(v)) \cap \Rr_{i}$ and, hence, $v \in \Rr_{i+1}$. This shows that $\Rr_{i} \subset \Rr_{i+1}$.
\vskip 5pt
\noindent {\it Proof of the fact that $(\Rr_{i})_{i \geq 0}$ is eventually constant.} We argue by contradiction: we assume that $(\Rr_{i})_{i \geq 0}$ is not eventually constant. Then, $\Rrinfty = \bigcup_{i\geq 0} \Rr_i$ is infinite because $(\Rr_{i})_{i\geq 0}$ is \croissant. 
For any $v \in \V$, we consider the sequence $(\Rr_{i}(v))_{i \geq 0}$ of subsets of $\V$ defined by  
\begin{itemize}
	\item $\Rr_0(v)= \{v\}$;
	\item $\forall i \geq 0$, $\Rr_{i+1}(v) = \{w \in \V \ \vert \ \pi(\Psi(w)) \cap \Rr_i(v) \neq \emptyset \}$.
	\end{itemize}
Then, $(\Rr_{i}(v))_{i \geq 0}$ is a \croissant{} sequence of finite sets (for the same reasons that $(\Rr_{i})_{i \geq 0}$ is an \croissant{} sequence of finite sets). We set $\Rrinfty(v)=\bigcup_{i\geq 0} \Rr_i(v)$. Note that : 
\begin{itemize}
\item we have $\Rrinfty = \bigcup_{v \in \Rr_{0}}  \Rrinfty(v)$; since $\Rrinfty$ is infinite and $\Rr_{0}$ is finite, there exists $v_{0} \in \Rr_{0}$ such that $\Rrinfty(v_{0})$ is infinite; 
\item we have $\Rrinfty(v_{0})=\{v_{0}\} \cup \bigcup_{v \in \Rr_{1}(v_{0}) \setminus \{v_{0}\}}  \Rrinfty(v)$;  since $\Rrinfty(v_{0})$ is infinite and $\Rr_{1}(v_{0})$ is finite, there exists $v_{1} \in \Rr_{1}(v_{0})\setminus \{v_{0}\}$ such that $\Rrinfty(v_{1})$ is infinite;
\item we have $\Rrinfty(v_{1})=\{v_{1}\} \cup \bigcup_{v \in \Rr_{1}(v_{1}) \setminus \{v_{1}\}}  \Rrinfty(v)$; since $\Rrinfty(v_{1})$ is infinite and $\Rr_{1}(v_{1})$ is finite, there exists $v_{2} \in \Rr_{1}(v_{1})\setminus \{v_{1}\}$ such that $\Rrinfty(v_{2})$ is infinite. 
\end{itemize}
Iterating this construction, we see that there exists a sequence $(v_{i})_{i\geq 0}$ of elements of $\V$ such that, for all $i\geq 0$, $v_{i+1} \in \Rr_{1}(v_{i})\setminus \{v_{i}\}$. Therefore, we have $v_{i} \in \pi(\Psi(v_{i+1}))$ and $v_{i+1} \neq v_{i}$ so $v_{i+1} < v_{i}$ by Lemma \ref{lem:Delta} applied to $w=v_{i}$ and $w'=v_{i+1}$ (we draw the reader's attention to the fact that ${v}_{i} \in \pi(\Psi({v}_{i+1}))$ and not the opposite, as it was in section \ref{sec:recept supp sol L bis}). So, the sequence $(v_{i})_{i\geq 0}$ is  \strictementdecroissant. This contradicts the fact that $\V$ is well-ordered. Thus, the sequence $(\Rr_i)_{i \in \Z_{\geq 0}}$ is eventually constant. In particular, $\Rrinfty = \Rr_{i_0}$ for some $i_0 \geq 0$ and, hence,  $\Rrinfty$ is a finite set.
\vskip 5pt
\noindent {\it Proof of the fact that $\Rrinfty$ satisfies property $\star_\V$.} The fact that $\Rrinfty$ satisfies property a. of Definition \ref{defi:prop star s} is obvious. Moreover, if $v \in \V \setminus \Rrinfty$, then, for all $i \geq 0$, we have $\pi(\Psi(v)) \cap \Rr_i = \emptyset$ and, hence, $\pi(\Psi(v)) \cap \Rrinfty = \emptyset$. Thus,
\begin{equation}\label{eq utile 1}
\bigcup_{v \in \V\setminus \Rr} \pi(\Psi(v)) \cap \Rrinfty = \emptyset.  
\end{equation}
But,
\begin{equation}\label{eq utile 2}
\V\setminus \Rr \subset \bigcup_{v \in \V\setminus \Rr} \pi(\Psi(v)) \subset \V; 
\end{equation}
indeed, the first inclusion follows from the first assertion of Lemma \ref{min Psi(s)=s} and the second inclusion follows from property (\ref{hypo pi psi dans V}) of Hypothesis \ref{hyp sol}.
Combining \eqref{eq utile 1} and \eqref{eq utile 2}, we get 
$$
\bigcup_{v \in \V\setminus \Rr} \pi(\Psi(v)) = \V \setminus \Rr.
$$	
So, $\Rrinfty$ satisfies property b. of Definition \ref{defi:prop star s}. 
\end{proof}

 Let $d \in \Z_{\geq 1}$ be a common denominator of the slopes of $L$ and let $\V$ be the set given by Theorem \ref{theo: un support bis}.  The following result gives an upper bound on the least  $i \in \Z_{\geq 0}$ such that $\Rr_{i}=\Rrinfty$ when $\V$ is the set given by Theorem \ref{theo: un support bis}.

\begin{prop}\label{prop bound R}
We let $\V$ be the set given by Theorem \ref{theo: un support bis}, we let $\E \subset \V$ be a finite set and we let $(\Rr_{i})_{i\geq 0}$ and $\Rrinfty$ be the sets given by Theorem \ref{th R star}. We let $H \in \Z_{\geq 0}$ be such that\footnote{We recall that $\V \subset \Z_{d,\ellmahl}$ by Lemma \ref{lem:V_Zdl},  so that any element of $\V$ and, hence, of the finite set $\E \subset \V$ is of the form $\frac{a}{d\ellmahl^h}$ for some $a \in \Z$ and $h \in \Z_{\geq 0}$.} $d\ellmahl^H \E \subset \Z$ and $\hgt\in \Q_{\geq 0}$ be such that $\E\cup -\S(L) \subset \Q_{\leq \hgt}$. Let $\minepsilon$ be the number defined by \eqref{eq:epsilon}, namely 
$$
\minepsilon= \min\big\{\epsilon(-\mu_1),\ldots,\epsilon(-\mu_{\kappa}),(d\ellmahl^n)^{-1}\big\} \in \Q_{>0}.
$$ We set
	$$
	c=\left\lfloor (n+1)\frac{\hgt+\mu_{\kappa}}{\minepsilon} \right\rfloor + H.
	$$
	Then, for any $i \in \Z_{\geq c}$, we have $\Rr_{i}=\Rrinfty$.
\end{prop}

\begin{proof}
	Theorem \ref{th R star} guarantees that the sequence $(\Rr_{i})_{i\geq 0}$ is \croissant{} and eventually constant. Let $M$ be the least element of $\Z_{\geq 0}$ such that $\Rr_{M}=\Rr_{M+1}$. It follows easily from the definition of $(\Rr_{i})_{i\geq 0}$ that, for all $i \in \Z_{\geq M}$, $\Rr_{i}=\Rrinfty$. In order to conclude the proof, it is thus sufficient to prove that $c \geq M$. Let us prove this. We have $c \geq 0$, so the result holds if $M=0$. Otherwise, assume $M \geq 1$. 
	The set $\Rr_{M} \setminus \Rr_{M-1}$ being nonempty, one can consider $v_{0} \in \Rr_{M} \setminus \Rr_{M-1}$.  It follows from the definition of $(\Rr_{i})_{i\geq 0}$, that there exist $v_{1},\ldots,v_{M} \in \V$ such that, for all $i \in \{0,\ldots,M-1\}$, 
		$$
		v_{i+1} \in \pi(\Psi(v_{i})) \cap \Rr_{M-1-i}. 
		$$
		We claim that, for any $i \in \{0,\ldots,M-1\}$, we have $v_i \notin \Rr_{M-1-i}$. Indeed, assume on the contrary that there exists $i \in \{0,\ldots,M-1\}$ such that $v_i \in \Rr_{M-1-i}$. Without loss of generality, we can assume that $i$ is the least element of $\{0,\ldots,M-1\}$ satisfying the latter property. Our choice of $v_{0}$ guaranties that $i \neq 0$.	 We have $v_{i} \in \Rr_{M-1-i}$ and, by construction, $v_{i} \in \pi(\Psi(v_{i-1})) \cap \Rr_{M-i}$, so $v_{i} \in \pi(\Psi(v_{i-1})) \cap \Rr_{M-1-i}$ and it follows from the definition of $\Rr_{M-i}$ that $v_{i-1} \in \Rr_{M-i}$. This contradicts the minimality of $i$ and concludes the proof of our claim.				
It follows that, for any $i \in \{0,\ldots,M-1\}$, we have $v_{i+1}\neq v_i$ because $v_{i+1} \in \Rr_{M-1-i}$ by construction and $v_{i} \notin \Rr_{M-1-i}$ by the previous claim.
Then, Lemma \ref{lem:Delta} applied to $w=v_{i+1}$ and $w'=v_{i}$ ensures that $v_{i+1} > v_{i}$. 
In conclusion,  $v_{0},v_{1},\ldots,v_{M} \in \V$ satisfy the following properties: 
	\begin{itemize}
	\item $v_{M} \in \Rr_{0} = \E \cup -\S(L)$;
	\item $v_{i+1} \in \pi( \Psi(v_{i}))$ for all $i \in \{0,\ldots,M-1\}$; 
	\item $v_{0} < v_1 < \cdots <v_{M}$. 
\end{itemize}
It follows from Lemma \ref{lem:longueur-chaine1} that 
$$
M \leq (n+1)\frac{v_{M}+\mu_{\kappa}}\minepsilon + h(v_{M}) \,.
$$
Since $v_{M} \leq \hgt$ and $h(v_{M})\leq H$ and since $M$ is an integer, we have $M \leq c$. This concludes the proof. 
\end{proof}

We note the following results  for further use. 

\begin{lem}\label{lem:infT}
With the notations of Theorem \ref{th R star}, 
$\max \Rrinfty = \max \E \cup -\S(L)$.
\end{lem}

\begin{proof}
The inequality $\max \Rrinfty \geq \max \E \cup -\S(L)$ follows from the fact that $\E \cup -\S(L)=\Rr_0 \subset \Rrinfty$. Proving the converse inequality $\max \Rrinfty \leq \max \E \cup -\S(L)$ is equivalent to proving that, for all $i \in \Z_{\geq 0}$, 
\begin{equation}\label{ineg pour rec}
 \max \Rrinfty_i  \leq  \max \E \cup -\S(L).
\end{equation}
 Let us prove this by induction on $i$.  The inequality \eqref{ineg pour rec} is obvious when $i=0$ since $\Rr_0=\E \cup -\S(L)$. Suppose that the inequality \eqref{ineg pour rec} is proved for some $i\in \Z_{\geq 0}$. Let $v \in \Rr_{i+1}$. By definition of $\Rr_{i+1}$, there exists $v' \in \Rr_i$ such that $v' \in \pi(\Psi(v))$. By induction hypothesis, $v' \leq \max \E \cup -\S(L)$. By Lemma \ref{min Psi(s)=s}, $v = \min \pi(\Psi(v))$, so $v \leq v'\leq \max \E \cup -\S(L)$. Therefore, $\max \Rr_{i+1} \leq \max \E \cup -\S(L)$. This concludes the induction.
\end{proof}

\begin{prop}\label{calc Ri avec VM}
We make the same assumptions and use the same notations as in Proposition~\ref{prop bound R}. Let $\minepsilonlb$ be a positive lower bound of $\minepsilon$. Then, the sequence $(\Rr_{i})_{i\geq 0}$ can also be recursively computed as follows:
	$$\Rr_0= (\Ecal \cup -\S(L)) \cap \V_{M}$$
and, for all $i \geq 0$,
\begin{equation}
\Rr_{i+1}=\bigcup_{(\ellmahl^{\alpha},\beta) \in \P(L)}  \ellmahl^{-\alpha} (\minP (\Rr_i)-\beta)  \cap \V_{M}   \label{form alt Ri+1 avec VM}
\end{equation}
where 
$$
M = (n+1) \left(\left\lfloor (n+1) \frac{\hgt+\mu_{\kappa}}{\minepsilonlb} \right\rfloor + H\right).
$$
\end{prop}

\begin{proof}
We have seen in Theorem \ref{th R star} that the sequence $(\Rr_{i})_{i\geq 0}$ can be recursively computed as follows:  
	$$\Rr_0= \Ecal \cup -\S(L)$$
and, for all $i \geq 0$, 
\begin{equation*}
\Rr_{i+1}=\bigcup_{(\ellmahl^{\alpha},\beta) \in \P(L)}  \ellmahl^{-\alpha} (\minP (\Rr_i)-\beta).
\end{equation*}
Intersecting with $\V_{M}$, we obtain: 
$$\Rr_0 \cap \V_{M}= (\Ecal \cup -\S(L)) \cap \V_{M}$$
and, for all $i \geq 0$, 
\begin{equation*}
\Rr_{i+1} \cap \V_{M}=\bigcup_{(\ellmahl^{\alpha},\beta) \in \P(L)}  \ellmahl^{-\alpha} (\minP (\Rr_i)-\beta) \cap \V_{M}.
\end{equation*}
Given this formula, in order to prove the Proposition, it is clearly sufficient to prove that, for all $i \geq 0$, we have $\Rr_{i} \subset \V_{M}$. Let us prove this.
 
We claim that it is sufficient to prove that $\Rr_{d} \subset \V_{M}$ where 
$$
	d=\left\lfloor (n+1)\frac{\hgt+\mu_{\kappa}}{\minepsilonlb} \right\rfloor + H.
	$$
Indeed, Proposition \ref{prop bound R} guarantees that, for all $i \in \Z_{\geq 0}$, $\Rr_{i} \subset \Rr_c=\Rr$ where 
 	$$
	c=\left\lfloor (n+1)\frac{\hgt+\mu_{\kappa}}{\minepsilon} \right\rfloor + H.
	$$
But, since $\minepsilonlb$ is a positive lower bound of $\minepsilon$, we have 
	$$
	d=\left\lfloor (n+1)\frac{\hgt+\mu_{\kappa}}{\minepsilonlb} \right\rfloor + H \geq c.
	$$
So, $\Rr_{c} = \Rr_d=\Rr$ and, for all $i \in \Z_{\geq 0}$, $\Rr_{i} \subset \Rr_c=\Rr_d$. This justifies our claim.

We now claim that, in order to conclude the proof, it is sufficient to prove that, for all $v \in \Rr_{d}$, we have  
\begin{enumerate}
 \item \label{v leq hgt} $v \leq \hgt$;
  \item \label{h leq H nd} $h(v) \leq H+nd$. 
\end{enumerate}
Indeed, Proposition \ref{prop:majoration-iterationV} ensures that, for all $v \in \Rr_{d}$, we have 
	 $$
	 v \in \V_{\left\lfloor (n+1)\frac{v+\mu_{\kappa}}{\minepsilon} + h(v) \right\rfloor}.
	 $$ 
But, if \eqref{v leq hgt} and \eqref{h leq H nd}  are true, then we have, for all $v \in \Rr_{d}$,
$$
 \left\lfloor (n+1)\frac{v+\mu_{\kappa}}{\minepsilon} + h(v) \right\rfloor 
 \leq \left\lfloor (n+1)\frac{N+\mu_{\kappa}}{\minepsilonlb} + H+nd \right\rfloor = (n+1)d=M.
$$ 
So $\Rr_{d} \subset \V_{M}$. This proves our claim. 
 
In order to complete the proof, it only remains to prove \eqref{v leq hgt} and \eqref{h leq H nd}. The inequality \eqref{v leq hgt} is a direct consequence of Lemma \ref{lem:infT} and of our choice of $N$. To justify inequality \eqref{h leq H nd}, we prove more generally that, for all $i \in \Z_{\geq 0}$, for all $v \in \Rr_{i}$,  $h(v) \leq H+ni$. We proceed by induction on $i$. The base case $i=0$ is true by our choice of $H$. We now assume that, for some $i \in \Z_{\geq 0}$, we have, for all $v \in \Rr_{i}$,  $h(v) \leq H+ni$. Consider $v \in  \Rr_{i+1}$. By definition of the sequence $(\Rr_{i})_{i \geq 0}$, the set $\pi(\Psi(v)) \cap \Rr_i$ is nonempty (see \eqref{eq:def des Ri}); let $w$ be in this intersection. By Lemma \ref{lem:V_Zdl}, we have $h(v) \leq h(w) + n$. But, by the inductive hypothesis, we have $h(w) \leq H+ni$. So $h(v) \leq  H+ni +n=H+n(i+1)$. This concludes the induction.  
	\end{proof}

\section{Answer to Question \ref{main question: un algo pour sol hahn?}}\label{sec: ancsw quest}

In this section, we consider a Mahler operator 
\begin{equation*}\label{operateur L bla}
L= a_{n}\pz \malop{\ellmahl}^{n} + a_{n-1}\pz\malop{\ellmahl}^{n-1} + \cdots + a_{0}\pz 
\end{equation*}
with coefficients $a_{0}\pz,\ldots,a_{n}\pz \in \bK[z]$ such that $a_{0}\pz a_{n}\pz \neq 0$.

\subsection{The algorithm}

\begin{algorithm}\label{algo: rep main question}
	$_{}$ \vskip 1 pt
	\noindent \rule{\linewidth}{1mm}
	{\bf Input}: $L$ a Mahler operator with coefficients in $\bK[z]$, $\E$  a finite subset of $ \Q$. \\
	{\bf Output}: the image under $\res{\E}$ of a basis of $\Sol(L,\Hahn)$; it is a generating family of the $\bK$-vector space made of the $f \in \Hahn_{\vert \mathcal E}$  for which there exists a solution 
	$\widetilde{f} \in \Hahn$ 
	of $L$ such that $
	\widetilde{f}_{\vert \mathcal E} 
	=f$ and even a basis if $-\S(L)\subset \mathcal{E}$. \vskip 1pt
	\noindent \rule{\linewidth}{1mm}
	\begin{itemize}
		\item[] set $\Rr_{-1}=\emptyset$ 
		\item[]  set $\breve{\tau}=$Lower\textunderscore Bound\textunderscore $\minepsilon$ ($L$) (see Algorithm \ref{algo:minepsilon})
			\item[] let $H$ be the least integer such that $\ellmahl^H \mathcal E\subset \frac{1}{d}\Z$
			\item[] set $N=\max \E\cup-\S(L)$ 
			\item[] set $M=(n+1)\left(\left\lfloor (n+1)\frac{\hgt+\mu_{\kappa}}{\minepsilonlb} \right\rfloor + H\right)$
		\item[] compute $\Rr_0 = (\E\cup -\S(L)) \cap \V_{\blue M}$ 
		\item[] set $i=0$
		\item[]  \textbf{while} $\Rr_{i}\neq \Rr_{i-1}$ 
		\begin{itemize}
			\item[] compute $\Rr_{i+1}=\bigcup_{(\ellmahl^{\alpha},\beta) \in \P(L)}  \ellmahl^{-\alpha} (\minP (\Rr_i)-\beta)  \cap \V_{\blue M} $
			\item[] increment $i$ by $1$
		\end{itemize}	 
		\item[] \textbf{end while}
		\item[] set $\Rrinfty = \Rr_i$;
		\item[] compute a basis $(f_1\pz,\ldots,f_t\pz)$ of the $\bK$-vector space $\mathcal C_{\Rrinfty}$
		\item[] return $(\res{\E}(f_1\pz),\ldots,\res{\E}(f_t\pz))$.
	\end{itemize}
	\noindent \rule{\linewidth}{0.5mm}
	\end{algorithm}

\begin{theo}\label{theo algo}
	Algorithm \ref{algo: rep main question} answers Question \ref{main question: un algo pour sol hahn?} by the positive.
\end{theo}

\begin{proof}
	Let us first consider the computability issues. 
We can compute $\V_M$ using the recursive formula from Theorem \ref{theo: un support bis}. Then, we can compute  $\Rr_{0}= (\E\cup -\S(L)) \cap \V_M$ because $\E\cup -\S(L)$ and $\V_M$ are explicit finite sets. We can compute $\Rr_{i+1}$ from $\Rr_{i}$ 
		with the formula $\Rr_{i+1}=\bigcup_{(\ellmahl^{\alpha},\beta) \in \P(L)} \ellmahl^{-\alpha} (\minP (\Rr_i)-\beta)   \cap  \V_M$ 
		because $\P(L)$, $\V_M$ and, for any $(\ellmahl^{\alpha},\beta) \in \P(L)$, $\ellmahl^{-\alpha} (\minP (\Rr_i)-\beta)$ are explicit finite sets. 
Proposition \ref{calc Ri avec VM} and Theorem \ref{th R star} guaranty that $\Rr_{i}=\Rr_{i-1}$ for some $i \in \Z_{\geq 1}$, so the ``while'' loop will stop after finitely many steps. 
Once $\Rrinfty$ has been calculated, computing a basis  $(f_1\pz,\ldots,f_t\pz)$ of $\Cond{\Rrinfty}$ amounts to compute a basis of solutions of an explicit system of linear equations. 
	Indeed one can compute explicit linear maps $(F_{\delta} : \bK^{\Rrinfty} \rightarrow \bK)_{\delta \in \minP(\Rrinfty)}$ such that, for any $f\pz=\sum_{\gamma \in \Rrinfty} f_{\gamma} z^{\gamma} \in \Hahn_{\vert \Rrinfty}$, 
	$$
	L(f\pz)=\sum_{\delta \in \minP(\Rrinfty)} F_{\delta}((f_{\gamma})_{\gamma \in \Rrinfty}) z^{\delta} +\text{terms whose support is disjoint from } \minP(\Rrinfty).
	$$
	So, $f\pz$ belongs to $\Cond{\Rrinfty}$ if and only if, for all 
	\begin{equation}\label{eq:Fdelta}
	F_{\delta}((f_{\gamma})_{\gamma \in \Rrinfty}) =0, \ \ \ \text{ for all } \delta \in \minP(\Rrinfty). 
	\end{equation}
	Finding a basis of $\Cond{\Rrinfty}$ amounts to finding a basis of solutions of this system of linear equations. This can be done algorithmically.

	Let us now justify that this algorithm returns the correct output, namely the image under $\res{\E}$ of a basis of $\Sol(L,\Hahn)$.

	It follows from Theorem \ref{theo: un support bis} that $\Sol(L,\Hahn) \subset \Hahn_{\vert \V}$.  Since $\Rr$ satisfies $\star_{\Rr_0,\V}$ by Proposition \ref{calc Ri avec VM} and Theorem \ref{th R star}, it follows from Theorem \ref{theo: si star alors isom} that the map $\res{\Rrinfty}: \Hahn \rightarrow \Hahn_{\vert \Rrinfty}$ induces an isomorphism between $\Sol(L,\Hahn)$ and $\mathcal C_{\Rrinfty}$. So, $(\res{\Rrinfty}^{-1}(f_1\pz),\ldots,\res{\Rrinfty}^{-1}(f_t\pz))$ is a basis of $\Sol(L,\Hahn)$. But, for all $i \in \{1,\ldots,t\}$, after setting $g_{i}=\res{\Rrinfty}^{-1}(f_i\pz)$, which is an element of $\Sol(L,\Hahn)$, we have $\supp(g_{i}\pz) \cap \E \subset \supp(g_{i}\pz) \cap \Rrinfty$ and, hence, 
	$$
	\res{\E}(\res{\Rrinfty}^{-1}(f_i\pz))=\res{\E}(g_i\pz)=\res{\E}(\res{\Rrinfty}(g_i\pz))=\res{\E}(f_i\pz).
	$$ Thus, 
	$$
	(\res{\E}(f_1\pz),\ldots,\res{\E}(f_t\pz))=(\res{\E}(\res{\Rrinfty}^{-1}(f_1\pz)),\ldots,\res{\E}(\res{\Rrinfty}^{-1}(f_t\pz)))
	$$ 
	is the image under $\res{\E}$ of a basis of $\Sol(L,\Hahn)$. Thus, $(\res{\E}(f_1\pz),\ldots,\res{\E}(f_t\pz))$ is a generating family of the $\bK$-vector space $\res{\E}(\Sol(L,\Hahn))$, which is nothing but the $\bK$-vector space made of the $f \in \Hahn_{\vert \mathcal E}$  for which there exists a solution 
	$\widetilde{f} \in \Hahn$ 
	of $L$ such that $
	\widetilde{f}_{\vert \mathcal E} 
	=f$. 
	Last, if $-\S(L) \subset \mathcal E$, then the restriction of $\res{\E}$ to $\Sol(L,\Hahn)$ is injective as a consequence of Corollary \ref{sol eq si eq on - pentes} and, hence,  $(\res{\E}(f_1\pz),\ldots,\res{\E}(f_t\pz))$ is a basis of the $\bK$-vector space $\res{\E}(\Sol(L,\Hahn))$. \end{proof}

\subsection{On the complexity of Algorithm \ref{algo: rep main question}}

In this section, by ``complexity'' we mean the number of basic operations ($+$, $-$, $\times$, $\div$) in $\bK$ and comparisons in $\mathbb Q \cup \{-\infty,+\infty\}$ performed by an algorithm.

To estimate the complexity of Algorithm \ref{algo: rep main question}, we shall suppose that $n\geq 2$. Indeed, if $n=1$, then \eqref{eq mahl intro} has a nonzero solution $f\pz \in \Hahn$ if and only if the coefficient of $z^{\val(a_{0}\pz)}$ in $a_{0}\pz$ is the opposite of the coefficient of $z^{\val(a_{1}\pz)}$ in $a_{1}\pz$; in this case, we have 
$$
f\pz=\lambda z^{\frac{\val(a_1)-\val(a_0)}{\ellmahl-1}} \prod_{k=0}^\infty \frac{-a_1(z^{\ellmahl^k})z^{-\ellmahl^k\val(a_1)}}{a_0(z^{\ellmahl^k})z^{-\ellmahl^k\val(a_0)}}
$$
for some $\lambda \in \bK \setminus \{0\}$. Thus, there is no need using Algorithm \ref{algo: rep main question}. 

\begin{prop}\label{prop compl algo}
	Suppose that $n\geq 2$. Let $\minepsilon$ be defined by \eqref{eq:epsilon} and $\minepsilonlb>0$ be a lower bound on $\minepsilon$. Let $\hgt$ be an integer and let 
	$$
	\E=\E_{\hgt}=\left\{ \frac{a}{b} \, \vert \, a\in\Z,b \in \Z\setminus\{0\}, \max\{\vert a \vert,\vert b \vert\} \leq \hgt\right\} 
	$$
	Suppose that $\hgt$ is large enough so that $\S(L) \subset \E$. Then, Algorithm \ref{algo: rep main question} has complexity
	\begin{equation}\label{eq:complexity}
	\mathcal O\left((\delta n)^{3n^2\hgt/\minepsilonlb} \right)
	\end{equation}
	 when one does not take into account the complexity of computing $\minepsilonlb$.
\end{prop}

\begin{rem}
	1. The complexity of the algorithm in Theorem \ref{theo algo} depends strongly on the lower bound $\minepsilonlb$ computed by Algorithm \ref{algo:minepsilon}.
	
	2. Of course, the complexity of this algorithm depends on the choice of the integer $\ellmahl$. In \eqref{eq:complexity}, this dependency  is hidden in the parameter $\minepsilonlb$ which is bounded from above by $\ellmahl^n$.
	
	3. In comparison, the algorithm given in \cite{CompSolMahlEq} to find Puiseux solutions has complexity $\tilde{\mathcal O}(n^2\hgt d\ellmahl^n)$, where $d$ is defined as in Section~\ref{sec:recept supp sol L bis}.
\end{rem}

\begin{proof}
		We note that, with the notation of the algorithm, we have $H\leq \lceil \log \hgt/\log \ellmahl \rceil$.	
	\subsubsection*{Computation of the Newton polygon}
	One can compute the set $\P(L)$ in $\mathcal O(\card \P(L))$ operations and the set 
		$$\{(j,\val a_j(z)) \,\vert \, j \in \{0,\ldots,n\}\}$$
		with the same complexity. Then, one can compute the set $\S(L)$ of slopes of $\mathcal N(L)$ and the endpoints of these slopes by performing $\mathcal O(n)$ comparisons and operations. Furthermore, one may return the set of slopes as an ordered list of rational numbers with the same complexity.
	
	\subsubsection*{Computation of $\V_M$} The set $\V_0=-\S(L)$ has $\kappa$ elements. This set can be computed in $\mathcal O(\kappa)$ operations, once $\S(L)$ is known. Then, for any $i$, the set $\mathcal V_i$ has at most $(\card \P(L))^i\kappa$ elements. Suppose that the set $\mathcal V_i$ has been computed for some $i$ and that it is given as an ordered list of rational numbers. Let us compute the ordered list of all elements of $\mathcal V_{i+1}$. Fix a point $\x=(\ellmahl^\alpha,\beta) \in \P(L)$. Then,  since $\mathcal V_i$ is given as an ordered list, one may compute the ordered list of elements of the set
	$$
	\mathcal V_{i}(\x)	=\left\{\pi(v\ellmahl^\alpha+\beta)\,\vert \, v \in \mathcal V_i\right\}
	$$
	in $\mathcal O((\card \P(L))^i\kappa)$.
	Thus, the computation of the $\card \P(L)$ lists $\mathcal V_{i}(\x)$, $\x \in \P(L)$, requires $\mathcal O((\card \P(L))^{i+1}\kappa)$ operations. Given $k$ ordered lists containing $m$ elements each, one may order the union of the $k$ lists in $\mathcal O(km\log(k))$ operations.
	Thus, once the ordered lists of elements of $\mathcal V_{i}(\x)$, $\x \in \P(L)$, are computed, the ordered list of elements of
	$$
	\mathcal V_{i+1}=\bigcup_{\x \in \P(L)} 	\mathcal V_{i}(\x)
	$$ 
	can be computed in $\mathcal O((\card \P(L))^{i+1}\kappa \log(\card P(L)))$. \textit{In fine}, the computation of the ordered list of elements of $\V_{M}$ can be performed with
	\begin{equation}
	\label{eq:complex_V}
	\mathcal O((\card \P(L))^{M}\kappa \log(\card P(L)))
	\end{equation} 
	operations. Furthermore, one has
	\begin{equation}\label{eq:cardV}
	\card \V_M\leq  (\card \P(L))^{M}\kappa.
	\end{equation}
	
	\subsubsection*{Computation of $\Rrinfty$} We now have an ordered list of all elements of $\V_M$. One may compute an ordered list of all elements of $\mathcal E$ in $\mathcal O(\card \mathcal E\log(\card \mathcal E))$. Two ordered lists of rational numbers being given, one can compute the ordered list of elements belonging to both lists making a number of comparisons at most equal to the maximum of the size of these lists. 
	Thus, one may compute the ordered list of elements of $\Rr_0 = \mathcal E \cap \V_M$ by making
	$\mathcal O(\card \mathcal E+\card \V_M)$ comparisons. Suppose that the ordered list of elements of $\Rr_i$ has been computed, for some integer $i\geq 0$. Note that $\card \Rr_i \leq \card \mathcal E (\card \P(L))^i$. Then, for each $\x=(\ellmahl^\alpha,\beta) \in \P(L)$ one can compute the ordered list of elements of
	$$
	\ellmahl^\alpha(\minP(\Rr_i)-\beta) \cap \V_M
	$$
	in $\mathcal O(\card \mathcal E (\card \P(L))^i+\card \V_M)$ operations. Once these $\card \P(L)$ lists are stored, one may compute the ordered list of elements of $\Rr_{i+1}$ by making $$\mathcal O(\card \E(\card \P(L))^{i+1}\log(\card \P(L)))$$ comparisons. In fine, the total complexity of the computation of $\Rrinfty=\Rr_c$ is in
	\begin{equation}\label{eq:complex_R}
	\mathcal O(\card \E(\card \P(L))^c\log(\card \P(L)) + c\card \P(L)\card \V_M)
	\end{equation}
	Furthermore, we have
	\begin{equation}\label{eq:card_R}
	\card \Rrinfty \leq \card \E(\card \P(L))^c\,.
	\end{equation}
	
	\subsubsection*{Computation of a basis of $\mathcal C_{\Rrinfty}$} Now that the ordered list of elements of $\Rrinfty$ has been computed, one may compute the ordered list of elements of $\minP(\Rr)$ with $\mathcal O(\card \Rrinfty)$ operations. The coefficients of the system of linear equations defining $\mathcal C_{\Rrinfty}$ can be computed with $\mathcal O(\card \P(L)\card \Rrinfty)$ operations. Solving a linear system with $m$ indeterminates necessitates $\mathcal O(m^{\rm{\Theta}})$ operations, for some ${\rm{\Theta}} < 3$ (for example, one may take ${\rm{\Theta}}=\log_2(7)\simeq 2,81$; see \cite{BCGLLSS}). Thus, once the coefficients of the system of linear equations is known, one may compute a basis of $\mathcal C_{\Rrinfty}$ in
	\begin{equation}\label{eq:complex_C}
	\mathcal O((\card \Rrinfty)^\Theta)
	\end{equation}
	operations. 
	
	\subsubsection*{Total complexity of the algorithm}
	Once a basis of $\mathcal C_{\Rrinfty}$ is known, 
	the complexity of the end of the algorithm is negligible with respect to the number of operations performed so far.
	Combining \eqref{eq:complex_V}, \eqref{eq:cardV}, \eqref{eq:complex_R}, \eqref{eq:card_R} and \eqref{eq:complex_C}, the complexity of this algorithm is 
	\begin{equation}\label{eq:complex_tot1}
	\mathcal O( \card \E^{\Theta}(\card \P(L))^{c\Theta}+c\card \P(L)^{M+1}\kappa)\,.
	\end{equation}
	Furthermore, we have the following bounds
	$$
	\left\{\begin{array}{rcl}
	\card \P(L)&\leq& (\delta+1)(n+1)\,,\\
	\card \E& \leq& 2\hgt^2\,,
	\\ \mu_{\kappa}&\leq& \hgt\,,
	\\c &\leq&2(n+1)\hgt\minepsilonlb^{-1} +\log \hgt\,,
	\\ M & \leq & 2(n+1)^2\hgt\minepsilonlb^{-1} + (n+1) \log \hgt\,,
	\\ \kappa &\leq& n\,,
	\\
	\Theta & < &n+1 \,.
	\end{array}\right.
	$$
	Thus, the second term in \eqref{eq:complex_tot1} dominates the first one. It follows that the complexity of the algorithm is
	$$
	\mathcal O\left( \left(2(n+1)\hgt\minepsilonlb^{-1} + \log \hgt\right) \big((\delta+1)(n+1)\big)^{2(n+1)^2\hgt\minepsilonlb^{-1} + (n+1) \log \hgt+1}n\right)
	$$
	Now, the result follows from the fact that the quantity above is in
	$$
	\mathcal O\left((\delta n)^{3n^2\hgt/\minepsilonlb} \right).
	$$
\end{proof}

\section{An example : the Rudin-Shapiro Mahler equation}\label{illustration main algo on RS}

Consider the following $2$-Mahler equation:
\begin{equation}
\label{eq:RudinS}
-2zy(z^4)+(z-1)y(z^2)+y(z)=0\,.
\end{equation}
One can prove that: 
\begin{enumerate}[label=\roman*)]
\item \label{enum:prop 1 RS} up to a multiplicative constant, the only solution of \eqref{eq:RudinS} in $\Hahn$ is actually a power series, namely the generating series of the Rudin-Shapiro sequence, see \cite{AS03} for instance; 
\item \label{enum:prop 2 RS} \eqref{eq:RudinS} has another nonzero solution of the form $fe_{-\frac{1}{2}}$ where $f \in \Hahn$ and where $e_{-\frac{1}{2}}$ satisfies $\malop{2} (e_{-\frac{1}{2}})=-\frac{1}{2}e_{-\frac{1}{2}}$. 
\end{enumerate}
We won't be proving property \ref{enum:prop 1 RS} here, as that would take us too far from our objective to illustrate Algorithm \ref{algo: rep main question}. Property \ref{enum:prop 2 RS} could be proved using \cite{RoquesFrobForMahler} but we will give another proof based on Algorithm \ref{algo: rep main question}, which has the advantage of also giving as many coefficients of $f$ as we like.

Let us first note that $fe_{-\frac{1}{2}}$ is a solution of  \eqref{eq:RudinS} if and only if $f$ is a solution of the $2$-Mahler equation
\begin{equation}\label{eq:RudinStwist}
zy(z^4)+(z-1)y(z^2)-2y(z)=0, 
\end{equation}
which is nothing but the $2$-Mahler equation associated to the  $2$-Mahler operator 
\begin{equation}
L=z\malop{2}^{2}+(z-1)\malop{2}-2
\end{equation}
defined by \eqref{L pour example}. 
In what follows,  we will run Algorithm \ref{algo: rep main question} for this $L$ and for the set $\E$ defined by
\begin{equation}\label{E pour example}
\E_8:=\left\{\frac{a}{b}\in \mathbb Q \ \vert \ \max\{ \vert a\vert,\vert b \vert\} \leq 8 \right\}\,. 
\end{equation}
We will see that the output of this algorithm is 
\begin{multline}\label{output algo ex RS}
	z^{-\frac{1}{2}}-2z^{-\frac{1}{4}}+4z^{-\frac{1}{8}}
	\\  -\frac{1}{3}+z^{\frac{1}{2}}-2z^{\frac{3}{4}}+4z^{\frac{7}{8}}-\frac{5}{6}z+z^{\frac{3}{2}}-2z^{\frac{7}{4}}  +\frac{11}{12}z^{2}-z^{\frac{5}{2}}
	\\  
	-\frac{5}{12}z^{3}
	+z^{\frac{7}{2}}
	-\frac{23}{24}z^{4}+\frac{13}{24}z^5-\frac{7}{24}z^6-\frac{5}{24}z^7-\frac{1}{48}z^8.
	\end{multline}
Since $-\S(L) \subset \E$, this shows that the $\bK$-vector space $\res{\E}(\Sol(L,\Hahn))$ has dimension $1$ and is generated by \eqref{output algo ex RS}. 
Moreover, since the restriction of $\res{\E}$ to $\Sol(L,\Hahn)$ is injective by Corollary \ref{sol eq si eq on - pentes}, this proves that  
 $$
 \Sol(L,\Hahn)=\C f
 $$ 
 for some $f \in \Hahn$ such that  $\res{\E}(f)=$\eqref{output algo ex RS}.
 Therefore, $fe_{-\frac{1}{2}}$ is a solution of \eqref{eq:RudinS} such that $\res{\E}(f)=$\eqref{output algo ex RS}. This justifies property \ref{enum:prop 2 RS} above and gives moreover  the value of the coefficients of $f$ corresponding to the indices in $\E$.

Let us take a close look at how Algorithm \ref{algo: rep main question} works when we take as input the operator $L$ given by \eqref{L pour example} and the set $\E$ given by \eqref{E pour example}. This is done in section \ref{main algo on RS example}, after some preliminaries. 

\subsection{Newton polygon and slopes of $L$}\label{sec:example Newton et slopes}
  The Newton polygon $\mathcal N(L)$ of $L$ is the lower convex hull of the set
$$
\mathcal P(L)=\{(1,0),(2,0),(2,1),(4,1)\}.
$$
We have 
$$\S(L)=\{\mu_{1},\mu_{2}\} \text{ with } \mu_{1}=0 \text{ and}  \mu_{2}=\frac{1}{2}.$$ 
A common denominator of the slopes of $L$ is thus $d=2$. 

The vertices, ordered by \strictementcroissant{} abscissa, of the polygon $\mathcal N(L)$ are
$$
\x_{0}=(1,0), \ \x_{1}=(2,0) \text{ and } \x_{2}=(4,1). 
$$
For any $k \in \{0,1,2\}$, we have
$$
\x_k=(\ellmahl^{\alpha_{k}},\beta_k)=(\ellmahl^{\alpha_{k}},\val a_{\alpha_{k}}\pz ) 
$$
with 
$$
\alpha_{0}=0, \beta_{0}=0,  
\alpha_{1}=1, \beta_{1}=0,  
\alpha_{2}=2, \beta_{2}=1.  
$$

The set $\P(L)$, the Newton polygon $\mathcal N(L)$ and the vertices $\x _{k}$ 
are represented in Figure \ref{fig:rudin shapiro polygon}.

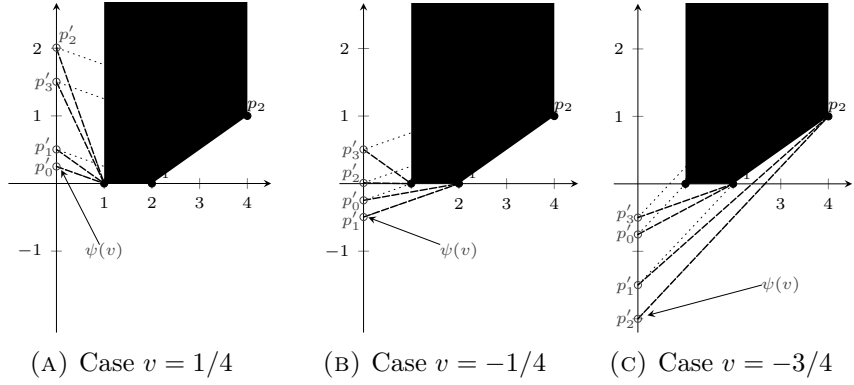
\begin{figure}
\begin{subfigure}[b]{0.31\textwidth}
        \centering
        \resizebox{\linewidth}{!}{
        	\definecolor{uuuuuu}{rgb}{0.26666666666666666,0.26666666666666666,0.26666666666666666}
			\definecolor{ffffff}{rgb}{1,1,1}
\begin{tikzpicture}[line cap=round,line join=round,>=triangle 45,x=1cm,y=1cm,font=footnotesize]
\begin{axis}[
x=0.7cm,y=1cm,
axis lines=middle,
xmin=-1,
xmax=4.5,
ymin=-2.2,
ymax=2.67084059715681,
xtick={0,1,...,4},
ytick={-1,0,1,...,2},font=\tiny]
\clip(-4.707418941998403,-7.112691098179308) rectangle (5.912988815477123,3.137237318919191);
\fill[line width=2pt,fill=black,fill opacity=0.05] (1,7) -- (1,0) -- (2,0) -- (4,1) -- (4,7) -- cycle;
\draw [line width=0.4pt,dash pattern=on 1pt off 1pt on 1pt off 4pt,color=ffffff] (4,7)-- (1,7);
\draw [line width=0.4pt,dotted] (2,1)-- (0,1.5043839298499524);
\draw [line width=0.4pt,dotted] (2,0)-- (0,0.5043839298499524);
\draw [line width=0.4pt,dotted] (1,0)-- (0,0.2521919649249762);
\draw [line width=0.4pt,dotted] (4,1)-- (0,2.008767859699905);
\draw [line width=0.6pt,dash pattern=on 3pt off 1pt] (0,2)-- (1,0);
\draw [line width=0.6pt,dash pattern=on 3pt off 1pt] (0,1.5043839298499524)-- (1,0);
\draw [line width=0.6pt,dash pattern=on 3pt off 1pt] (0,0.5043839298499524)-- (1,0);
\draw [line width=0.6pt,dash pattern=on 3pt off 1pt] (0,0.2521919649249762)-- (1,0);
\begin{scriptsize}
\draw [fill=black] (1,0) circle [radius=1.5pt];
\draw[color=black] (1.2,0.14) node {$\x_0$};
\draw [fill=black] (2,0) circle [radius=1.5pt];
\draw[color=black] (2.2,0.14) node {$\x_1$};
\draw [fill=black] (4,1) circle [radius=1.5pt];
\draw [fill=black] (2,1) circle [radius=1.5pt];
\draw[color=black] (2.2,1.14) node {$\x_3$};
\draw [fill=black] (4,1) circle [radius=1pt];
\draw[color=black] (4.2,1.14) node {$\x_2$};
\draw [color=uuuuuu] (0,1.5043839298499524) circle [radius=1.5pt];
\draw[color=uuuuuu] (-0.25,1.5043839298499524) node {$\x_3'$};
\draw [color=uuuuuu] (0,0.5043839298499524) circle [radius=1.5pt];
\draw[color=uuuuuu] (-0.25,0.55) node {$\x_1'$};
\draw [color=uuuuuu] (0,0.2521919649249762) circle [radius=1.5pt];
\draw[color=uuuuuu] (-0.25,0.2521919649249762) node {$\x_0'$};
\draw [stealth-]  (0.1,0.2) --  (0.9,-0.9);
\draw[color=uuuuuu] (1,-1) node {$\minP(v)$};
\draw [color=uuuuuu] (0,2.008767859699905) circle [radius=1.5pt];
\draw[color=uuuuuu] (0.25,2.2) node {$\x_2'$};
\end{scriptsize}
\end{axis}
\end{tikzpicture}
}
        \caption{Case $v=1/4$}
        \label{fig:subfigA}
    \end{subfigure}
\begin{subfigure}[b]{0.3\textwidth}
    \centering
        \resizebox{\linewidth}{!}{
\definecolor{uuuuuu}{rgb}{0.26666666666666666,0.26666666666666666,0.26666666666666666}
\definecolor{ffffff}{rgb}{1,1,1}
\begin{tikzpicture}[line cap=round,line join=round,>=triangle 45,x=1cm,y=1cm,font=\footnotesize]
\begin{axis}[
x=0.7cm,y=1cm,
axis lines=middle,
xmin=-0.5,
xmax=4.5,
ymin=-2.2,
ymax=2.67084059715681,
xtick={0,1,...,4},
xticklabels={,,2,3,4},
ytick={-1,0,1,...,2},font=\tiny]
\clip(-4.201369057657044,-6.1685681796319765) rectangle (4.862870732195982,2.67084059715681);
\fill[line width=2pt,fill=black,fill opacity=0.05] (1,7) -- (1,0) -- (2,0) -- (4,1) -- (4,7) -- cycle;
\draw [line width=0.6pt,dash pattern=on 1pt off 1pt on 1pt off 4pt,color=ffffff] (4,7)-- (1,7);
\draw [line width=0.4pt,dotted] (2,1)-- (0,0.5053629450117656);
\draw [line width=0.4pt,dotted] (2,0)-- (0,-0.4946370549882344);
\draw [line width=0.4pt,dotted] (1,0)-- (0,-0.2473185274941172);
\draw [line width=0.4pt,dotted] (4,1)-- (0,0.010725890023531193);
\draw [line width=0.6pt,dash pattern=on 3pt off 1pt] (0,0.5053629450117656)-- (1,0);
\draw [line width=0.6pt,dash pattern=on 3pt off 1pt] (0,-0.4946370549882344)-- (2,0);
\draw [line width=0.6pt,dash pattern=on 3pt off 1pt] (0,0.010725890023531193)-- (1,0);
\draw [line width=0.6pt,dash pattern=on 3pt off 1pt] (0,-0.2473185274941172)-- (2,0);
\begin{scriptsize}
\draw [fill=black] (1,0) circle [radius=1.5pt];
\draw[color=black] (1.2,0.14) node {$\x_0$};
\draw [fill=black] (2,0) circle [radius=1.5pt];
\draw[color=black] (2.2,0.14) node {$\x_1$};
\draw [fill=black] (4,1) circle [radius=1.5pt];
\draw [fill=black] (2,1) circle [radius=1.5pt];
\draw[color=black] (2.2,1.14) node {$\x_3$};
\draw [fill=black] (4,1) circle [radius=1pt];
\draw[color=black] (4.2,1.14) node {$\x_2$};
\draw [color=uuuuuu] (0,0.5053629450117656) circle [radius=1.5pt];
\draw[color=uuuuuu] (-0.25,0.5053629450117656) node {$\x_3'$};
\draw [color=uuuuuu] (0,-0.4946370549882344) circle [radius=1.5pt];
\draw[color=uuuuuu] (-0.25,-0.54) node {$\x_1'$};
\draw [stealth-]  (0.1,-0.5) --  (1.6,-0.9);
\draw[color=uuuuuu] (2,-1) node {$\minP(v)$};
\draw [color=uuuuuu] (0,-0.2473185274941172) circle [radius=1.5pt];
\draw[color=uuuuuu] (-0.25,-0.2473185274941172) node {$\x_0'$};
\draw [color=uuuuuu] (0,0.010725890023531193) circle [radius=1.5pt];
\draw[color=uuuuuu] (-0.25,0.16) node {$\x_2'$};
\end{scriptsize}
\end{axis}
\end{tikzpicture}}
        \caption{Case $v=-1/4$}   
        \label{fig:subfigB}
    \end{subfigure}
\begin{subfigure}[b]{0.275\textwidth}
    \centering
        \resizebox{\linewidth}{!}{\definecolor{uuuuuu}{rgb}{0.26666666666666666,0.26666666666666666,0.26666666666666666}
\definecolor{ffffff}{rgb}{1,1,1}
\begin{tikzpicture}[line cap=round,line join=round,>=triangle 45,x=1cm,y=1cm,font=\footnotesize]
\begin{axis}[
x=0.7cm,y=1cm,
axis lines=middle,
xmin=-0.5,
xmax=4.5,
ymin=-2.2,
ymax=2.67084059715681,
xtick={0,1,...,4},
xticklabels={,,,3,4},
ytick={-2,-1,0,1,...,2},
yticklabels={,,,1,2},
font=\tiny]
\clip(-4.201369057657044,-6.168568179631977) rectangle (4.957536421907239,2.67084059715681);
\fill[line width=2pt,fill=black,fill opacity=0.05] (1,7) -- (1,0) -- (2,0) -- (4,1) -- (4,7) -- cycle;
\draw [line width=0.4pt,dash pattern=on 1pt off 1pt on 1pt off 4pt,color=ffffff] (4,7)-- (1,7);
\draw [line width=0.4pt,dotted] (2,1)-- (0,-0.5004600081703587);
\draw [line width=0.4pt,dotted] (2,0)-- (0,-1.5004600081703587);
\draw [line width=0.4pt,dotted] (1,0)-- (0,-0.7502300040851794);
\draw [line width=0.4pt,dotted] (4,1)-- (0,-2.0009200163407175);
\draw [line width=0.6pt,dash pattern=on 3pt off 1pt] (0,-0.5004600081703587)-- (2,0);
\draw [line width=0.6pt,dash pattern=on 3pt off 1pt] (0,-0.7502300040851794) -- (2,0);
\draw [line width=0.6pt,dash pattern=on 3pt off 1pt] (0,-1.5004600081703587)-- (4,1);
\draw [line width=0.6pt,dash pattern=on 3pt off 1pt] (0,-2.0009200163407175)-- (4,1);
\begin{scriptsize}
\draw [fill=black] (1,0) circle [radius=1.5pt];
\draw[color=black] (1.2,0.14) node {$\x_0$};
\draw [fill=black] (2,0) circle [radius=1.5pt];
\draw[color=black] (2.2,0.14) node {$\x_1$};
\draw [fill=black] (4,1) circle [radius=1.5pt];
\draw [fill=black] (2,1) circle [radius=1.5pt];
\draw[color=black] (2.2,1.14) node {$\x_3$};
\draw [fill=black] (4,1) circle [radius=1pt];
\draw[color=black] (4.2,1.14) node {$\x_2$};
\draw [color=uuuuuu] (0,-0.5004600081703587) circle [radius=1.5pt];
\draw[color=uuuuuu] (-0.25,-0.44) node {$\x_3'$};
\draw [color=uuuuuu] (0,-1.5004600081703587) circle [radius=1.5pt];
\draw[color=uuuuuu] (-0.25,-1.5004600081703587) node {$\x_1'$};
\draw [color=uuuuuu] (0,-0.7502300040851794) circle [radius=1.5pt];
\draw[color=uuuuuu] (-0.25,-0.7402300040851794) node {$\x_0'$};
\draw [color=uuuuuu] (0,-2.0009200163407175) circle [radius=1.5pt];
\draw[color=uuuuuu] (-0.25,-2.0009200163407175) node {$\x_2'$};
\draw [stealth-]  (0.2,-1.95) --  (2.6,-1.5);
\draw[color=uuuuuu] (3,-1.5) node {$\minP(v)$};
\end{scriptsize}
\end{axis}
\end{tikzpicture}}
 \caption{Case $v=-3/4$}   
        \label{fig:subfigC}
    \end{subfigure}
\caption{\small{This figure is relative to the operator $L$ given by~\eqref{L pour example}. 
	We have $\P(L)=\{\x_{0},\x_{1},\x_{2},\x_{3}\}$. The Newton polygon $\mathcal N(L)$ is the shaded area. Its vertices are $\x_{0},\x_{1}$ and $\x_{2}$ and we have $\S(L)=\{0,1/2\}$. In each subfigure, we consider a specific $v \in \mathbb Q$. The point $\x_{k}'$ is the projection of $\x_{k}$ along a line of slope $-v$ onto the $y$-axis; so, the dotted segments have slope $-v$ and $\Psi(v)=\{\x_0',\x_1',\x_2',\x_3'\}$.
	The dashed segment with left extremity $\x_{k}'$ is the segment with lowest slope among those linking $\x_{k}'$ to an element of $\P(L)$. The slope of this segment is thus the opposite of  $\pi(q_{k})$ where $q_{k}$ is the ordinate of $\x_{k}'$ and, hence,  $\pi(\Psi(v))$ is the set of the opposite of the  slopes of the four dashed segments.
  }
} 
\label{fig:rudin shapiro polygon}
\end{figure}

\subsection{The maps $\pi$ and $\Psi$}\label{sec:calc pi Psi RS}
Straightforward calculations show that
$$
\Psi(v)=\left\{v,2v,2v+1,4v+1\right\} \text{ and } \pi(q)=\left\{\begin{array}{rl}
q/2 & \text{if }-1 \leq q<0,
\\ q & \text{if }q\geq 0. \end{array}
\right.
$$
We do not need to specify $\pi(q)$ when $q < -1$.

\subsection{The sets $\V_i$}\label{the Vi for RS example}
One can compte as many $\V_{i}$ as necessary using their recursive definition. For instance, 
\begin{eqnarray*}
 \mathcal V_0&=&-\S(L)=\left\{-\frac{1}{2},0\right\},\\ 
\V_1&=&\pi\left(\Psi\left(-\frac{1}{2}\right)\cup\Psi(0)\right)
= \pi\left(\left\{-1,-\frac{1}{2},0,1\right\}\right)
=\left\{-\frac{1}{2},-\frac{1}{4},0,1\right\}, \\
\V_2&=& \pi\left(\Psi\left(-\frac{1}{2}\right) \cup\Psi\left(-\frac{1}{4}\right) \cup \Psi\left(0\right) \cup \Psi(1)\right)
\\ &=&\left\{-\frac{1}{2},-\frac{1}{4},-\frac{1}{8},0,\frac{1}{2},1,2,3,5\right\} \\
\cdots & \cdots & \cdots
\end{eqnarray*}

\begin{rem}
 In this very peculiar example, we could prove that
\begin{equation*}
\label{eq:V_rudin_shapiro}
\V=\left\{ k-\frac{1}{2^n} \ \vert \  k,n\in \Z_{\geq 0},\, (k,n)\neq (0,0)\right\}.
\end{equation*}
As this property will not be used, we only briefly indicate how to prove the most useful inclusion, namely the inclusion of $\V$ in the right-hand side of the latter equality, and leave the details and proof of the other inclusion to the reader.   Let us denote this right-hand side by $\mathcal W$. One can easily check that $-\S(L)\subset \mathcal W$ and that, for any $w \in \mathcal W$, $\pi(\Psi(w))\subset \mathcal W$. Given the definition of $\V$, this clearly implies that $\V \subset \mathcal W$.  
\end{rem}

When running Algorithm \ref{algo: rep main question}, we will first need to call Algorithm \ref{algo:minepsilon} in order to calculate a positive lower bound $\minepsilonlb$ on $\minepsilon$, which is defined by \eqref{eq:epsilon}.

\subsection{Computation of a positive lower bound on $\minepsilon$}\label{comput vareps example RS}

We will now explain how  Lower\textunderscore Bound\textunderscore $\minepsilon$($L$) described in Algorithm \ref{algo:minepsilon} runs to compute a lower bound $\minepsilonlb$ on $\minepsilon$. It takes the following steps: 
\begin{enumerate}
 \item it computes a lower bound 
 $$\theta_{2}=\text{LB\textunderscore $\epsilon$\textunderscore p$(L,2,(),-\mu_{2})$=LB\textunderscore $\epsilon$\textunderscore p$(L,2,(),-\frac{1}{2})$}
 $$ on $\epsilon(-\mu_{2})=\epsilon(-\frac{1}{2})$; 
 \item it computes a lower bound 
 $$\theta_{1}=\text{LB\textunderscore $\epsilon$\textunderscore p$(L,1,(\theta_{2}),-\mu_{1})$=LB\textunderscore $\epsilon$\textunderscore p$(L,1,(\theta_{2}),0)$}
 $$ on $\epsilon(-\mu_{1})=\epsilon(0)$;
 \item it returns 
 $$
 \minepsilonlb=\min\{\theta_1,\theta_2,\frac{1}{2\cdot 2^2}\};
 $$
\end{enumerate}
where we have written LB\textunderscore $\epsilon$\textunderscore p for Lower\textunderscore Bound\textunderscore $\epsilon$\textunderscore param defined in Algorithm \ref{algo:lower bound eps avec param}, in order to avoid heavy notations. These three steps are detailed in the following three sections. 

\subsubsection{Execution of LB\textunderscore $\epsilon$\textunderscore p$(L,2,(),-\mu_{2})=$LB\textunderscore $\epsilon$\textunderscore p$(L,2,(),-\frac{1}{2})$}
\label{item:theta_2} 
We are in the situation starting in row 4 of Algorithm \ref{algo:lower bound eps avec param} since $-\frac{1}{2}=-\mu_{2}$. 
Thus, LB\textunderscore $\epsilon$\textunderscore p($L,2,(),-\frac{1}{2})$ returns
	$$
\min\left( \V_1 \setminus\left\{-\frac{1}{2}\right\}\right)+ \frac{1}{2} = -\frac{1}{4}+\frac{1}{2} = \frac{1}{4}\,.
	$$
This latter value $\frac{1}{4}$ is stored in $\theta_{2}$. 
\subsubsection{Execution of LB\textunderscore $\epsilon$\textunderscore p$(L,1,(\theta_{2}),-\mu_{1})=$LB\textunderscore $\epsilon$\textunderscore p$(L,1,(\frac{1}{4}),0)$} 
It takes the following steps. 
	\begin{enumerate}
		\item Since $0 = -\mu_1$, we are in the situation starting in row \ref{debutStep3} of Algorithm \ref{algo:lower bound eps avec param} with $k=2$.
		\item For any $w' \in \Delta(-\mu_{1})=\Delta(0)=\{-\frac{1}{2},-\frac{1}{4}\}$, it computes
		$$m_{w'}=\text{LB\textunderscore $\epsilon$\textunderscore p$(L,1,(\frac{1}{4}),w')$}.
		$$
		\item LB\textunderscore $\epsilon$\textunderscore p$(L,1,(\frac{1}{4}),-\frac{1}{2})$ returns $\frac{1}{4}$ after calculations similar  to those described in \ref{item:theta_2}.
		\item LB\textunderscore $\epsilon$\textunderscore p($L,1,(\frac{1}{4}),-\frac{1}{4}$) runs as follows.
\begin{enumerate}
 \item Since $-\frac{1}{4} \in ]-\frac{1}{2},0[=]-\mu_{2},-\mu_{1}[$, we are in the situation starting in row \ref{debutStep2} of Algorithm \ref{algo:lower bound eps avec param} with $k=2$. So, LB\textunderscore $\epsilon$\textunderscore p($L,1,(\frac{1}{4}),-\frac{1}{4}$) calls LB\textunderscore $\epsilon$\textunderscore i $(L,1,(\frac{1}{4}), -\frac{1}{4})$, where we have written LB\textunderscore $\epsilon$\textunderscore i for Lower\textunderscore Bound\textunderscore $\epsilon$\textunderscore interval defined in Algorithm \ref{algo:lower bound eps avec param interval}, in order to avoid heavy notations.
 \item LB\textunderscore $\epsilon$\textunderscore i $(L,1,(\frac{1}{4}), -\frac{1}{4})$ runs as follows.
 \begin{enumerate}
 \item Since $-\frac{1}{4} \geq -\mu_{2}+\theta_{2}=-\frac{1}{2}+\frac{1}{4}=-\frac{1}{4}$, we are in the situation starting in row \ref{finStep2d} of Algorithm \ref{algo:lower bound eps avec param interval}. Therefore, for each $v \in \Delta(-\frac{1}{4})=\{-\frac{1}{2},-\frac{3}{4},-\frac{3}{8}\}$, it computes $\boldsymbol b_{v}=\text{LB\textunderscore $\epsilon$\textunderscore i$(L,1,(\frac{1}{4}),v)$}$.
\begin{enumerate}
 \item LB\textunderscore $\epsilon$\textunderscore i$(L,1,(\frac{1}{4}),-\frac{3}{4})$ returns $(-\frac{3}{4},\frac{1}{4})$  because $-\frac{3}{4}<-\frac{1}{2}=-\mu_{2}$ so we are in the situation starting with row \ref{finStep2b} of Algorithm \ref{algo:lower bound eps avec param interval} so it calls LB\textunderscore $\epsilon$\textunderscore p$(L,2,(),-\frac{3}{4})$ which returns $-\mu_{2}-(-\frac{3}{4})=-\frac{1}{2}+\frac{3}{4}=\frac{1}{4}$ because we are in the situation starting in  row \ref{debutStep1} of Algorithm \ref{algo:lower bound eps avec param}.
  \item LB\textunderscore $\epsilon$\textunderscore i$(L,1,(\frac{1}{4}),-\frac{1}{2})$ returns $(-\frac{1}{2},\frac{1}{4})$ because $-\frac{1}{2} < -\mu_{2}+\theta_{2}=-\frac{1}{2}+\frac{1}{4}=-\frac{1}{4}$ and $-\frac{1}{2} \geq -\mu_{2}=-\frac{1}{2}$ so we are in the situation starting with row \ref{finStep2c} of Algorithm \ref{algo:lower bound eps avec param interval}  and  $-\mu_{2}+\theta_{2}-(-\frac{1}{2})=-\frac{1}{2}+\frac{1}{4}+\frac{1}{2}=\frac{1}{4}$.
 \item LB\textunderscore $\epsilon$\textunderscore i$(L,1,(\frac{1}{4}),-\frac{3}{8})$ returns $(-\frac{3}{8},\frac{1}{8})$ because $-\frac{3}{8} < -\mu_{2}+\theta_{2}=-\frac{1}{2}+\frac{1}{4}=-\frac{1}{4}$ and $-\frac{3}{8} \geq -\mu_{2}=-\frac{1}{2}$ so we are in the situation starting with row \ref{finStep2c} of Algorithm \ref{algo:lower bound eps avec param interval} and $-\mu_{2}+\theta_{2}-(-\frac{3}{8})=-\frac{1}{2}+\frac{1}{4}+\frac{3}{8}=\frac{1}{8}$.
\end{enumerate}
\item It computes the minimum $m$ of the set \eqref{eq:setpourminvertexnotleaf}, which is, as explained in details in Example~\ref{ex:rec step case 1}, the minimum of  
				\begin{itemize}[label=$\bullet$]
					\item  $\frac{1}{4}\times2^{0-1}=\frac{1}{8}$, 
					\item $\frac{1}{4}\times2^{1-1}=\frac{1}{4}$, 
					\item $\frac{1}{8}\times2^{2-1}=\frac{1}{4}$, 
					\item $0 - \left(- \frac{1}{4}\right)=\frac{1}{4}$,
					\item $\min\{\pi(\Psi(-\frac{1}{4}))\setminus \{-\frac{1}{4}\}\} + \frac{1}{4}=-\frac{1}{8}+\frac{1}{4} = \frac{1}{8}$.
				\end{itemize}
We thus have $m=\frac{1}{8}$.
\item LB\textunderscore $\epsilon$\textunderscore i($L,1,(\frac{1}{4}),-\frac{1}{4}$) returns $(-\frac{1}{4},\frac{1}{8})$.
\end{enumerate}
		\item LB\textunderscore $\epsilon$\textunderscore p$(L,1,(\frac{1}{4}),-\frac{1}{4})$ returns $\frac{1}{8}$.  
					\end{enumerate}
		\item Last, it runs row \ref{finalstepalgo} of Algorithm \ref{algo:lower bound eps avec param} and, after calculations already presented in details in Example \ref{ex:step 3 pour RS}, we find that LB\textunderscore $\epsilon$\textunderscore p$(L,1,(\frac{1}{4}),0)$ returns $\frac{1}{2}$.  	
\end{enumerate}
The latter value $\frac{1}{2}$ is stored in $\theta_1$.

\subsubsection{Last step of Lower\textunderscore Bound\textunderscore $\minepsilon$($L$) } Finally,  Lower\textunderscore Bound\textunderscore $\minepsilon$($L$) returns 
$$
\minepsilonlb=\min\left\{\theta_1,\theta_2,\frac{1}{8}\right\} =\frac{1}{8}\,.
$$

\subsection{Algorithm \ref{algo: rep main question}}\label{main algo on RS example}

 We already computed the positive lower bound $\minepsilonlb=\frac{1}{8}$ on $\minepsilon$. We check that $H=2$ and $N=8$ by definition of $\mathcal E$. Then we have
	$$
	M=(n+1)\left(\left\lfloor (n+1)\frac{\hgt+\mu_{\kappa}}{\minepsilonlb} \right\rfloor + H\right)=618\,.
	$$
	We could compute the set $\V_M$ but, in practice, we can save a few calculations by exploiting the fact that, according to Lemma \ref{lem:infT}, any element of $\Rr$ (and, hence, of any $\Rr_{i}$) is lower than or equal to $\max \E \cup -\S(L) = 8$. Indeed, this remark entails that the intersections with $\V_{M}$ involved in Algorithm \ref{algo: rep main question} for calulating the $\Rr_{i}$ can be replaced  by intersections with $\V_M\cap \Q_{\leq 8}$ without changing the result of the calculations, {\it i.e.}, we have 
	$$\Rr_0 = (\E\cup -\S(L)) \cap \V_{M}=(\E\cup -\S(L)) \cap \V_{M} \cap \Q_{\leq 8}$$
			and  $$\Rr_{i+1}=\bigcup_{(\ellmahl^{\alpha},\beta) \in \P(L)}  \ellmahl^{-\alpha} (\minP (\Rr_i)-\beta)  \cap \V_{M}= \bigcup_{(\ellmahl^{\alpha},\beta) \in \P(L)}  \ellmahl^{-\alpha} (\minP (\Rr_i)-\beta)  \cap \V_{M} \cap \Q_{\leq 8}.$$
Thus, we  only need to compute the set $\V_M\cap \Q_{\leq 8}$, not the whole $\V_{M}$. Let us now explain how one can compute the sets $\Vmaj{i}=\V_i\cap \Q_{\leq 8}$.  We recall that, by definition, the $(\V_{i})_{i \geq 0}$ can be computed recursively as follows: $$
\V_{0}=-\S(L)$$ 
and, for all $i \in \Z_{\geq 0}$,  
$$
\V_{i+1}=\bigcup_{v \in \V_i} \pi(\Psi(v)). 
$$ 
Intersecting with $\Q_{\leq 8}$, we obtain: 
\begin{equation}\label{eq:init Vmaj}
 \Vmaj{0}=-\S(L) \cap \Q_{\leq 8} =-\S(L)
\end{equation}
and, for all $i \in \Z_{\geq 0}$,  
\begin{equation}\label{eq rec inters with Q leq 8}
\Vmaj{i+1}=\bigcup_{v \in \V_i} \pi(\Psi(v)) \cap \Q_{\leq 8}.  
\end{equation}
But, according to \eqref{eq:min Psi(s)=s}, we have, for all $v \in \Q$,  $v=\min \pi(\Psi(v))$. So, for any $v \in \V_{i} \setminus \Vmaj{i}$, we have $\pi(\Psi(v)) \cap \Q_{\leq 8}=\emptyset$ and, hence, \eqref{eq rec inters with Q leq 8} can be rewritten as follows:  
\begin{equation}\label{eq rec Vmaj}
\Vmaj{i+1}=\bigcup_{v \in \Vmaj{i}} \pi(\Psi(v)) \cap \Q_{\leq 8}.  
\end{equation}
Now, \eqref{eq:init Vmaj} and \eqref{eq rec Vmaj} allow us to recursively calculate as many $\Vmaj{i}$ as we like. Note also that, for computing the right-hand side of \eqref{eq rec Vmaj}, we only need to compute $\pi(\Psi(v))$ for $v \in \Vmaj{i} \setminus \Vmaj{i-1}$ because 
$$
\Vmaj{i+1} = \Vmaj{i} \cup \bigcup_{v \in \Vmaj{i}\setminus \Vmaj{i-1}} \pi(\Psi(v)) \cap \Q_{\leq 8}.
	$$
Using this, we may compute $\Vmaj{M}$, which has $5512$ elements, in a fair time\footnote{With a basic desktop computer and the computer algebra software Giac/Xcas it took us less than a minute.}. 	We then compute
\begin{multline*}
\mathcal R_0= (\E\cup -\S(L)) \cap \V_{M}=\E \cap \V_{M} \cap \Q_{\leq 8}=\E \cap \Vmaj{M}\\ 
=	\left\{-\frac{1}{2},-\frac{1}{4},-\frac{1}{8},0,\frac{1}{2},\frac{3}{4},\frac{7}{8},1,\frac{3}{2},\frac{7}{4},2,\frac{5}{2},3,\frac{7}{2},4,5,6,7,8\right\}. 
\end{multline*}
Then, in order to compute $\Rr_{1}$, we first compute (the finite set) $$\bigcup_{(\ellmahl^{\alpha},\beta) \in \P(L)}  \ellmahl^{-\alpha} (\minP (\Rr_0)-\beta)$$
and, then, we compute its intersection with $\Vmaj{M}$. We obtain
$$
\Rr_1 = \Rr_0 \cup \left\{-\frac{1}{16},-\frac{1}{32}\right\}.
$$
 Iterating this process, we find $\Rr_2=\Rr_1$. Thus, $\Rr=\Rr_1$ is a set with $21$ elements.

The next step of Algorithm \ref{algo: rep main question} consists in computing a basis of $\mathcal C_{\Rr}$. 
As explained in the proof of Theorem \ref{theo algo}, this amounts to solve the linear system 
$$F_{\delta}((f_{\gamma})_{\gamma \in \Rr})=0,\quad \delta \in \minP(\Rr)
$$
with $\card \minP(\Rr)=\card \Rr=21$ equations given by \eqref{eq:Fdelta}.  
 We may gather these equations in a matrix whose columns are indexed by the elements of $\Rr$ and whose rows are indexed by the elements of $\minP(\Rr)$. The coefficient of the entry $(\lambda,\gamma)$ of this matrix is the coefficient of $z^\lambda$ in 
$$
L(z^\gamma)=z^{1+4\gamma}+(z-1)z^{2\gamma}-2z^\gamma.
$$
To save space, we shall not reproduce here the square matrix of this system.
After calculations, we find that the kernel of this matrix has dimension $1$, generated by some explicit $(f_{\gamma})_{\gamma \in \Rr}$. This means that $\mathcal C_{\Rr}$ has dimension $1$ and is generated 
by $\sum_{\gamma \in \Rr} f_{\gamma} z^{\gamma}$. 

Last, Algorithm \ref{algo: rep main question} returns $\sum_{\gamma \in \E} f_{\gamma} z^{\gamma}$. Replacing the $f_{\gamma}$ by their explicit values, we find \eqref{output algo ex RS}.

\bibliographystyle{alpha}
\bibliography{biblio}

\end{document}